\PassOptionsToPackage{unicode}{hyperref}
\PassOptionsToPackage{hyphens}{url}
\PassOptionsToPackage{dvipsnames,svgnames,x11names}{xcolor}
\documentclass[
  12pt]{article}
\usepackage{tabularx,threeparttable,array}
\usepackage{ragged2e}
\usepackage{natbib}
\usepackage{amsmath,amssymb}
\usepackage{algorithmicx, algpseudocode}
\usepackage{algorithm}
\usepackage{float}
\usepackage{amsthm}
\newtheorem{theorem}{Theorem}
\newtheorem{proposition}{Proposition}
\newtheorem{assumption}{Assumption}
\newtheorem{remark}{Remark}
\usepackage[table]{xcolor} % row_spec in R
\usepackage{multirow} % pack_rows in R
\usepackage{threeparttablex}
\usepackage[normalem]{ulem} % strike-through
\newcommand{\bbeta}{\boldsymbol{\beta}}
 \newcommand{\btheta}{\boldsymbol{\Theta}}
 \newcommand{\bstheta}{\boldsymbol{\vartheta}}
 
 \newcommand{\bphia}{{\boldsymbol{\beta}_0}}
 \newcommand{\bphib}{{\boldsymbol{\Gamma}_0}}
 \newcommand{\bgamma}{\boldsymbol{\Gamma}}

 \newcommand{\bomega}{\boldsymbol{\Omega}}
 \newcommand{\bomegac}{\boldsymbol{\bar{\Omega}}}
 \newcommand{\bepsilon}{\boldsymbol{\epsilon}}
\newcommand{\bX}{\boldsymbol{X}}
 \newcommand{\bM}{\boldsymbol{M}}
 \newcommand{\bA}{\boldsymbol{A}}
 \newcommand{\bB}{\boldsymbol{B}}
 
 \newcommand{\bS}{\boldsymbol{S}}
 \newcommand{\bSa}{{\boldsymbol{S}_r}}
 \newcommand{\bSb}{{\boldsymbol{S}_c}}

 \newcommand{\bE}{\boldsymbol{E}}
 \newcommand{\bY}{\boldsymbol{Y}}
 \newcommand{\bZ}{\boldsymbol{Z}}
 \newcommand{\bU}{\boldsymbol{U}}
 \newcommand{\bV}{\boldsymbol{V}}
 \newcommand{\bD}{\boldsymbol{D}}
\newcommand{\bone}{\boldsymbol{1}}
\newcommand{\iden}[1]{\mathbf{I}_{#1}}
\newcommand{\zeromat}[1]{\mathbf{0}_{#1}}

\newcommand{\lambb}{\lambda_\beta}
\newcommand{\lambm}{\lambda_m}
 \newcommand{\lambg}{\lambda_\Gamma}

\newcommand{\trace}[1]{\operatorname{tr}\left(#1\right)}
\newcommand{\estim}[1]{\widehat{{#1}}}
\newcommand{\stept}[1]{{#1}^{(t)}}
\newcommand{\steptmin}[1]{{#1}^{(t-1)}}

\newcommand{\newv}[1]{{#1}^{(t)}} %% for now it'll be the same as stept
\newcommand{\indexx}[3]{{#1}_{(#2#3)}}
\newcommand{\indexxt}[3]{\stept{{#1}}_{(#2#3)}}

\newcommand{\truev}[1]{{#1}^{\bracks{*}}}
\newcommand{\cxb}{c_{x\beta}}
\newcommand{\cgz}{c_{\gamma z}}

 \newcommand{\R}[2]{\mathbb{R}^{#1\times #2}}
 \newcommand{\Ra}[1]{\mathbb{R}^{#1}}
\newcommand{\bracks}[1]{\left(#1\right)}
\newcommand{\bracksb}[1]{\left\{#1\right\}}
\newcommand{\bracksc}[1]{\left[#1\right]}
\newcommand{\tr}[1]{{#1}^\top}

\newcommand{\wtilde}[1]{\widetilde{#1}}

 \newcommand{\loss}[1]{\mathcal{L}\left(#1\right)}
  
 \newcommand{\lossst}[1]{\mathcal{L}^{(t)}\left(#1\right)}

 \newcommand{\losssf}[1]{\mathcal{L}_{\mathcal{F}}\left(#1\right)}
 \newcommand{\losssft}[1]{\mathcal{L}^{(t)}_{\mathcal{F}}\left(#1\right)}
 \newcommand{\lpenalty}[1]{\mathcal{P}{\bracks{#1}}}
 \newcommand{\deriv}[2]{\frac{\delta {#1}}{\delta {#2} }}

 \newcommand{\softthresh}[2]{\mathbb{S}_{#1}\left(#2\right)}
 
 \newcommand{\opr}[1]{\mathcal{P}_\Omega \left(#1\right)}
\newcommand{\normfrob}[1]{\left\|#1\right\|_F^2}
\newcommand{\normnuc}[1]{\left\|#1\right\|_*}
\newcommand{\normlass}[1]{\left\|#1\right\|_1}
\newcommand{\norminf}[1]{\left\|#1\right\|_\infty}
\newcommand{\normop}[1]{\left\|#1\right\|_{\text{op}}}
\newcommand{\normf}[1]{\left\|#1\right\|_F}

\newcommand{\mse}[1]{\operatorname{MSE\bracks{#1}}}

\allowdisplaybreaks
\usepackage{iftex}
\ifPDFTeX
  \usepackage[T1]{fontenc}
  \usepackage[utf8]{inputenc}
  \usepackage{textcomp} % provide euro and other symbols
\else % if luatex or xetex
  \usepackage{unicode-math}
  \defaultfontfeatures{Scale=MatchLowercase}
  \defaultfontfeatures[\rmfamily]{Ligatures=TeX,Scale=1}
\fi
\usepackage{lmodern}
\ifPDFTeX\else  
\fi
\IfFileExists{upquote.sty}{\usepackage{upquote}}{}
\IfFileExists{microtype.sty}{% use microtype if available
  \usepackage[]{microtype}
  \UseMicrotypeSet[protrusion]{basicmath} % disable protrusion for tt fonts
}{}
\makeatletter
\@ifundefined{KOMAClassName}{% if non-KOMA class
  \IfFileExists{parskip.sty}{%
    \usepackage{parskip}
  }{% else
    \setlength{\parindent}{0pt}
    \setlength{\parskip}{6pt plus 2pt minus 1pt}}
}{% if KOMA class
  \KOMAoptions{parskip=half}}
\makeatother
\makeatletter
\ifx\paragraph\undefined\else
  \let\oldparagraph\paragraph
  \renewcommand{\paragraph}{
    \@ifstar
      \xxxParagraphStar
      \xxxParagraphNoStar
  }
  \newcommand{\xxxParagraphStar}[1]{\oldparagraph*{#1}\mbox{}}
  \newcommand{\xxxParagraphNoStar}[1]{\oldparagraph{#1}\mbox{}}
\fi
\ifx\subparagraph\undefined\else
  \let\oldsubparagraph\subparagraph
  \renewcommand{\subparagraph}{
    \@ifstar
      \xxxSubParagraphStar
      \xxxSubParagraphNoStar
  }
  \newcommand{\xxxSubParagraphStar}[1]{\oldsubparagraph*{#1}\mbox{}}
  \newcommand{\xxxSubParagraphNoStar}[1]{\oldsubparagraph{#1}\mbox{}}
\fi
\makeatother

\usepackage{longtable,booktabs,array}
\usepackage{calc} % for calculating minipage widths
\usepackage{etoolbox}
\makeatletter
\patchcmd\longtable{\par}{\if@noskipsec\mbox{}\fi\par}{}{}
\makeatother
\IfFileExists{footnotehyper.sty}{\usepackage{footnotehyper}}{\usepackage{footnote}}
\makesavenoteenv{longtable}
\usepackage{graphicx}
\makeatletter
\def\maxwidth{\ifdim\Gin@nat@width>\linewidth\linewidth\else\Gin@nat@width\fi}
\def\maxheight{\ifdim\Gin@nat@height>\textheight\textheight\else\Gin@nat@height\fi}
\makeatother
\setkeys{Gin}{width=\maxwidth,height=\maxheight,keepaspectratio}
\makeatletter
\def\fps@figure{htbp}
\makeatother

\makeatletter
\@ifpackageloaded{caption}{}{\usepackage{caption}}
\AtBeginDocument{%
\ifdefined\contentsname
  \renewcommand*\contentsname{Table of contents}
\else
  \newcommand\contentsname{Table of contents}
\fi
\ifdefined\listfigurename
  \renewcommand*\listfigurename{List of Figures}
\else
  \newcommand\listfigurename{List of Figures}
\fi
\ifdefined\listtablename
  \renewcommand*\listtablename{List of Tables}
\else
  \newcommand\listtablename{List of Tables}
\fi
\ifdefined\figurename
  \renewcommand*\figurename{Figure}
\else
  \newcommand\figurename{Figure}
\fi
\ifdefined\tablename
  \renewcommand*\tablename{Table}
\else
  \newcommand\tablename{Table}
\fi
}
\@ifpackageloaded{float}{}{\usepackage{float}}
\floatstyle{ruled}
\@ifundefined{c@chapter}{\newfloat{codelisting}{h}{lop}}{\newfloat{codelisting}{h}{lop}[chapter]}
\floatname{codelisting}{Listing}

\makeatother
\makeatletter
\@ifpackageloaded{caption}{}{\usepackage{caption}}
\@ifpackageloaded{subcaption}{}{\usepackage{subcaption}}
\makeatother

\ifLuaTeX
  \usepackage{selnolig}  % disable illegal ligatures
\fi

\usepackage{bookmark}

\IfFileExists{xurl.sty}{\usepackage{xurl}}{} % add URL line breaks if available
\hypersetup{
  pdftitle={Incomplete Matrix Regression},
  pdfauthor={Khaled Fouda, Aurélie Labbe, Karim Oualkacha},
  pdfkeywords={Matrix completion, Collaborative filtering, Nuclear norm, Low-rank estimation, Side information},
  colorlinks=true,
  linkcolor={blue},
  filecolor={Maroon},
  citecolor={Blue},
  urlcolor={Blue},
  pdfcreator={LaTeX via pandoc}}

\newcommand{\anon}{1}

\begin{document}

\def\spacingset#1{\renewcommand{\baselinestretch}%
{#1}\small\normalsize} \spacingset{1}

%%%%%%%%%%%%%%%%%%%%%%%%%%%%%%%%%%%%%%%%%%%%%%%%%%%%%%%%%%%%%%%%%%%%%%%%%%%%%%

\if1\anon
{
  \title{\bf Incomplete Matrix Regression}
    \author{Khaled Fouda\thanks{Correspondence to \href{mailto:khaled.fouda@hec.ca}{khaled.fouda@hec.ca}.},\\
    % Department of Decision Sciences, HEC Montréal,\\ Montreal, Quebec, H3T 2A7, Canada\\
    % \and
    Aurélie Labbe, \\
    Department of Decision Sciences, HEC Montréal,\\ Montreal, Quebec, H3T 2A7, Canada\\
    \and
    %and\\
    Karim Oualkacha \\
    Department of Mathematics, University of Quebec in Montreal,\\ Montreal, Quebec, H2X 3Y7, Canada
    }
  \maketitle
} \fi
\if0\anon
{
  \bigskip
  \bigskip
  \bigskip
  \begin{center}
    {\LARGE\bf Incomplete Matrix Regression}
\end{center}
  \medskip
} \fi

\bigskip
\begin{abstract}
Matrix completion seeks to recover a low-rank matrix from a sparse and noisy subset of its entries. In many applications, such as recommendation systems and urban mobility, the observed matrix is accompanied by auxiliary covariates on its rows and columns and exhibits dependence across them. We propose Incomplete Matrix Regression (IMR), a distribution-free penalized regression framework that integrates such information into matrix completion. The target matrix is modeled as the sum of intercepts, covariate effects regularized by a Lasso penalty, and a low-rank latent component that captures structure unexplained by the covariates. Known similarity structures, such as spatial and temporal kernels, are incorporated through ridge-type penalties on the latent factors. For estimation, we provide a scalable alternating least-squares algorithm whose modular form allows us to include or exclude individual model components without rederiving the updates. We establish non-asymptotic error bounds for both the Lasso and matrix completion estimators that are consistent with standard rates in their respective literature. Through simulation studies and two real-data applications, we demonstrate that the proposed method attains predictive accuracy competitive with more complex methods at a small fraction of their computational cost. The methodology is implemented in the \texttt{R} package \texttt{IMR}.
\end{abstract}

\noindent%
\vfill

\newpage
\spacingset{1.8}

\section{Introduction}\label{sec-intro}

Data arising from recommendation systems, urban mobility records, and drug--target interactions often share a common structure: they can be represented as large, sparsely observed matrices with entries governed by a small number of latent factors. The task of recovering missing entries from observed data is known as matrix completion (MC). Classical MC methods rely exclusively on observed entries and can achieve accurate reconstruction provided that the number of observed entries is sufficiently large, measurements are noiseless, and the underlying low-rank assumption holds. In practice, however, observations are frequently noisy, exhibit autocorrelation, and are accompanied by auxiliary information associated with their rows and columns. This motivates the integration of classical MC with multivariate regression to leverage all available information.

In drug repositioning, for instance, the objective is to predict unknown drug--disease associations based on a sparse matrix of known interactions, where rows correspond to drugs and columns to diseases \citep{meng_drug_2021,sadeghi_Networkbased_2022,cui_drug_2021, meng_weighted_2022,saxena_drug_2020}. The incorporation of auxiliary information (e.g., chemical structures) can enhance the performance of MC models \citep{yang_overlap_2019}.  Similarly, in transportation networks, MC can be applied to traffic flow data to impute missing sensor readings, where rows denote temporal points (e.g., days) and columns denote sensor locations \citep{chen_nonconvex_2020}. Such data typically manifest strong temporal and spatial autocorrelation that needs to be integrated into the model \citep{jia_missing_2021}. While machine learning methods have been successful in these domains, they often function as black boxes and lack interpretability  \citep{you_artificial_2022, yu_predicting_2021}.

In this article, we propose Incomplete Matrix Regression (\texttt{IMR}), a distribution-free penalized regression framework for matrix completion, with four main contributions. First, we introduce a unified model that augments the low-rank matrix completion problem with row and column intercepts and with covariate effects, the latter regularized by a Lasso penalty. The framework further accommodates known dependence across rows and columns, represented by inverse-covariance matrices, through ridge-type penalties on the latent factors. The estimator is fully modular, in that any subset of its components can be switched off to obtain a valid submodel without affecting the estimation of the remaining parameters, and several existing methods, including the models of \citet{hastie_matrix_2015} and \citet{ma_statistical_2025}, are recovered as special cases. Second, we develop a scalable alternating least-squares algorithm whose closed-form updates entail a low per-iteration cost and yield substantial computational savings relative to competing methods, without sacrificing predictive accuracy. Third, we establish nonasymptotic upper bounds on the estimation errors of the parameters under mild regularity conditions. Finally, we provide an \texttt{R} package, \texttt{IMR}\if1\anon\footnote{\href{https://github.com/khaledfouda/IMR}{https://github.com/khaledfouda/IMR}}\fi, that implements the methodology and selects all tuning parameters by cross-validation.

The remainder of the article is organized as follows. Section \ref{sec:method} presents the proposed model, its estimation algorithm, and its theoretical properties, together with a review of the related literature. Section \ref{sec:simulation} evaluates the method through simulation studies, and Section \ref{sec:application} illustrates it on two real-world datasets concerning movie recommendation and bike-sharing demand. We discuss avenues for future work in Section \ref{sec:conc}. Proofs of all theoretical results are deferred to the online Appendices, which are provided in the Supplementary Materials.

\section{Proposed Method and Related Work\label{sec:method}}

\textbf{Notation.} Throughout this article, we adopt the following notation for a matrix $\bA \in \R{n}{m}$. The trace of $\bA$ is denoted by $\trace{\bA}=\sum_i\bA_{ii}$. The squared Frobenius norm is given by $\normfrob{\bA}=\sum_{ij}\bA_{ij}^2$, the entry-wise infinity norm by $\norminf{\bA}=\max_{ij}|\bA_{ij}|$, and the entry-wise $\ell_1$ norm by $\normlass{\bA}=\sum_{ij}|\bA_{ij}|$. We denote the $i$-th singular value of $\bA$ by $\sigma_i$. The nuclear norm is defined as $\normnuc{\bA}=\sum_i\sigma_i$, and the operator norm is given by $\normop{\bA}=\max_{i}\sigma_i$. We use $\estim{\bA}$ to denote an estimator of $\bA$. Finally, $\iden n$ and $\zeromat{n\times n}$ represent the $n\times n$ identity and zero matrices, respectively, and $\bone_n$ denotes the $n$-dimensional vector of ones.

\subsection{The Model}

Let $\btheta = \bracksb{\btheta_{ij}} \in \R{n}{m}$ be the unknown target matrix (the complete data matrix of interest), and let  $\bY  \in \R{n}{m}$ be the observed matrix, a noisy, partially observed version of $\btheta$. We define an indicator matrix $\bomega  \in \R{n}{m}$ such that $\bomega_{ij}=1$ if $\bY_{ij}$ is observed and $\bomega_{ij}=0$ otherwise. We assume the following contamination model:
\[
    \bY_{ij} = \bomega_{ij}  \bracks{\btheta_{ij} + \bepsilon_{ij}},\quad \text{for }i = 1,\dots,n, \quad j=1,\dots,m,
\]
where the errors $\bracksb{\bepsilon_{ij}}$ are independent random variables with zero mean and finite variance. In addition to the incomplete matrix $\bY$, we also observe covariates associated with the rows and columns, denoted by $\bX$ and $\bZ$, respectively. For example, consider a user--movie rating matrix where $\bY_{ij}$ represents the rating user $i$ assigns to movie $j$. The matrix $\bX$ contains user attributes (e.g., age, demographics), while $\bZ$ contains movie attributes (e.g., genre, release year). Since each user rates only a small subset of movies, $\bY$ is sparse.

We suppose that these covariates relate to the target $\btheta$ through the linear model:
\begin{align}\label{eq:decomp}
    \btheta =   \bX\bbeta  +
    \bgamma \bZ + \bone_n \tr\bphib +  \bphia \tr\bone_m + \bM.
\end{align}
Here, 
$\bX = \tr{\bracksc{{\bX_1,\dots, \bX_n}}} \in \R{n}{p}$ is the row covariate matrix, where $\bX_i \in \Ra{p}$ is the feature vector for row $i$. The coefficient matrix $\bbeta = \bracksc{\bbeta_1,\dots,\bbeta_m} \in \R{p}{m}$ captures the effects of the row covariates, with $\bbeta_j \in \Ra{p}$ denoting the effect vector specific to column $j$.
Similarly, $\bZ = \bracksc{\bZ_1,\dots,\bZ_m } \in \R{q}{m}$ is the column covariate matrix, where $\bZ_j \in \Ra{q}$ is the feature vector for column $j$. The coefficient matrix $\bgamma = \bracksc{\bgamma_1,\dots,\bgamma_n}^\top \in \R{n}{q}$ captures the effects of the column covariates, with row $\bgamma_i\in\Ra{q}$ representing the effect vector specific to row $i$. We restrict the feature dimensions such that $p<n$ and $q<m$. The vectors $\bphia \in \Ra{n}$ and $\bphib\in \Ra{m}$ contain the model intercepts for the rows and columns, respectively.

  The matrix $\bM \in \R{n}{m}$ is a latent low-rank component that captures structure unexplained by the covariates. We assume $\bM$ is driven by a small number of latent factors, $r$, where $r \ll \min(n,m)$. In the user--movie rating example, \citet{rennie_fast_2005} explain that the premise behind $\bM$ is that only a small number of factors influence the ratings, and that a user's rating vector is determined by how each factor applies to that user. Thus, for $n$ users and $m$ movies, the ratings are given by the product of an $n\times r$ coefficient matrix $\bA$ (each row representing the extent to which each factor is used) and an $m \times r$ factor matrix $\bB$ whose columns are the factors. The rating matrix is then expressed as $\bM=\bA\tr\bB$. We assume the data are missing completely at random (MCAR); that is, the probability of missingness depends neither on the covariates nor on the response matrix. This assumption allows us to ignore the missing-data mechanism during estimation (see Remark \ref{remark:1} below). 

{Our goal is to estimate the intercepts as well as the covariate coefficients and latent matrix $\bM$. We denote the set of parameters to be estimated as $\bstheta=(\bbeta,\bgamma,\bphia,\bphib, \bA, \bB)$}. %\bphia$ and $\bphib$) ($\bbeta$ and $\bgamma$), and the latent matrix $\bM$. 
To avoid over-parameterization, we penalize the model defined in equation \eqref{eq:decomp} using the following principles. First, we induce sparsity in $\bbeta$ and $\bgamma$ through an $\ell_1$ (Lasso) penalty. Second, to promote a low-rank solution for $\bM$, we introduce a nuclear-norm penalty, which is a convex relaxation of a rank constraint; however, this approach requires a full singular value decomposition (SVD) of a high-dimensional matrix \citep{ma_fixed_2011}. An alternative is to impose a rank constraint and apply squared Frobenius penalties to $\bA$ and $\bB$, following \citet{rennie_fast_2005} and invoking Lemma 6 of \citet{mazumder_spectral_2010}, which states that
\begin{align}
    \normnuc{\bM} = \min_{\bA,\bB:\bM=\bA\tr\bB} \frac{1}{2}\bracks{\normfrob{\bA}+\normfrob{\bB}}. \label{eq:mazm_lemma6}
\end{align}
Finally, we incorporate structural dependencies among rows and columns to account for possible correlation in the rows/columns of $\bY$. For example, in the user--movie rating setting, users belonging to the same social group may exhibit similar tastes in movies, and movies with shared attributes (e.g., same director) may receive similar ratings. We model these dependencies using symmetric positive definite matrices $\bSa \in \R{n}{n}$ and $\bSb \in \R{m}{m}$, which represent graph Laplacians or inverse-covariance matrices associated with the rows and columns of $\bY$, respectively. We impose ridge-type penalties on $\bSa^{1/2}\bA$ and $\bSb^{1/2}\bB$ to encourage similarity in the latent estimates of connected nodes. 

Combining these regularization components, we obtain the following objective function for the proposed model.
\begin{align}
\label{eq:obj}
   \loss{\bstheta;\bY,\bX,\bZ,\bomega,\bSa,\bSb} \, & = \,\losssf{\bstheta;\bY,\bX,\bZ,\bomega} +\lpenalty{\bstheta;\bSa,\bSb},
 \end{align}
 where
 \begin{align} 
    \losssf{\bstheta;\bY,\bX,\bZ,\bomega}  \, & = \, \frac{1}{2}\normfrob{\bomega \circ (\bY-\btheta)}, \label{eq:obj:frob1} \\  
   \lpenalty{\bstheta;\bSa,\bSb}\, & = \,   \frac{1}{2} \lambm \bracks{\trace{\bA^\top\bSa \bA} + \trace{\bB^\top\bSb \bB} }  +
    \lambb \normlass{\bbeta} + 
    \lambg \normlass{\bgamma}.  \notag
\end{align}
Here, $\circ$ denotes the Hadamard (element-wise) product, and $\lambm,\lambb,\lambg$ are nonnegative tuning parameters selected via cross-validation. The rank $r$ is likewise selected via cross-validation. {For the remainder of the paper, we omit arguments $\bY,\bX,\bZ,\bomega,\bSa,\bSb$ from the objective function $\mathcal{L}$ and write it concisely as $\loss{\bstheta}{}$}.

%=========================================================================
\subsection{Related Work \label{sec:literature}}

Here, we review key literature on regression-based matrix completion and establish connections with our proposed model. First, in the noiseless and covariate-free setting, \citet{candes_exact_2009} established strong theoretical guarantees for exact matrix recovery under the model $\btheta=\bM$ via nuclear-norm regularization:
\[
\min_{\bM}\normnuc{\bM} \qquad \text{s.t.} \quad \opr{\bY} = \opr{\bM}.
\]
In practice, however, observations are typically noisy. Consequently, the problem is reformulated as:
\[
\min_{\bM}\normnuc{\bM} \qquad \text{s.t.} \quad \normfrob{\bomega\circ\bracks{\bY-\bM}} < \delta,
\]
where $\delta$ is a regularization parameter controlling the error tolerance. \citet{mazumder_spectral_2010} introduced an iterative algorithm to solve this problem, which converges to a global solution. However, the algorithm requires iterative computation of the SVD of $\bM$, which can become a computational bottleneck. Motivated by \citet{rennie_fast_2005}, \citet{hastie_matrix_2015} proposed a computationally efficient alternating least-squares solution. Their method, called \texttt{Soft-Impute}, solves the optimization problem: 
\[
\min_{\bM=\bA\tr\bB}\bracks{\normfrob{\bA}+\normfrob{\bB}} \qquad \text{s.t.} \quad \normfrob{\bomega\circ\bracks{\bY-\bM}} < \delta.
\]

Recent methods have leveraged side information to improve recovery. For instance, row covariates with regression coefficients are used in \citet{zhu_personalized_2016,  robin_Lowrank_2018, jin_matrix_2022, mao_Matrix_2019, ma_statistical_2025, sun_Noisy_2025}. However, the likelihood-based approaches in the first three studies rely on specific distributional assumptions and homogeneous covariate effects across columns, which complicate inference and limit flexibility in real-world applications. While \citet{mao_Matrix_2019} propose a distribution-free model, their one-step least-squares solution is suboptimal; the method of \citet{meng_CovariateAssisted_2024} is subject to similar limitations. Although \citet{ma_statistical_2025} achieved promising results using alternating least squares, their model lacks regularization, which is critical for controlling model complexity. Finally, \citet{sun_Noisy_2025} proposed a method that alternates between solving a full \texttt{Soft-Impute} problem and a regularized generalized linear model, which is computationally expensive.

Alternative approaches leverage side information to improve recovery without explicitly estimating covariate coefficients. These methods, known as Inductive Matrix Completion \citep{jain_provable_2013, natarajan_inductive_2014, chiang_matrix_2015,yi_regularized_2020,zilber_inductive_2022}, address the following optimization problem:
\[
\min_{\mathbf{K}}\normnuc{\mathbf{K}} \qquad \text{s.t.} \quad \normfrob{\bomega\circ\bracks{\bY-\bX\,\mathbf{K}\,\bZ}} < \delta,
\]
where $\bX$ and $\bZ$ denote row and column covariates, respectively, and $\mathbf{K}\in\R{p}{q}$ is a low-rank latent matrix. While distribution-free, these models focus on imputation rather than inference, limiting their utility when covariate effects are of primary interest. Reduced-rank regression, which constrains covariate coefficients to be low-rank, also provides valuable insights for matrix completion, although it typically assumes the data matrix is fully observed \citep{she_robust_2017,ma_adaptive_2020,tan_sparse_2023, li_supervised_2016}. Finally, row and column similarity matrices have also been incorporated in collaborative filtering to improve matrix recovery \citep{kalofolias_Matrix_2014}.

% \textbf{Remarks}

{The following remarks clarify how the proposed framework relates to the aforementioned approaches in the literature and illustrate the sense in which it encompasses some of them as special cases.}
% KO added this sentence; need to make several remarks as items in one-single remark 
% \begin{enumerate}
% \begin{enumerate}
\begin{remark}

\label{remark:1}
\begin{itemize}
\noindent
\item Our framework assumes an MCAR mechanism, so we do not explicitly model the missingness probabilities. It is worth distinguishing estimation from inference in this respect. The estimator itself is agnostic to the sampling mechanism and can be computed for any missingness pattern. The MCAR assumption enters only through the theoretical guarantees. In contrast, \citet{sun_Noisy_2025} and \citet{ma_statistical_2025} assume that these probabilities depend on the covariates and model them through a logistic regression (propensity-score) function. We adopt the MCAR assumption to keep the framework general-purpose and to avoid the additional computational burden and modeling assumptions that estimating these probabilities entails. Indeed, identifying the missingness mechanism requires prior knowledge of whether the missingness is driven by the response, the covariates, or other confounding factors.
%\end{remark}
%\begin{remark}
\item In the absence of prior knowledge about the structural dependencies ($\bSa=\iden{n}$, $\bSb=\iden{m}$), the penalty on $\bM$ reduces to the standard nuclear-norm proxy in \eqref{eq:mazm_lemma6}: $\trace{\bA^\top\iden n \bA} + \trace{\bB^\top\iden m \bB} = \trace{\bA^\top\bA} + \trace{\bB^\top \bB} = \normfrob{\bA}+\normfrob{\bB}$.
    If we further exclude the covariates and intercepts, the model reduces to \texttt{Soft-Impute} \citep{hastie_matrix_2015}. Thus,  \texttt{IMR}  can be viewed as an extension of \texttt{Soft-Impute} that incorporates penalized multivariate regression on both $\bY$ and $\tr\bY$.
%\end{remark}
%  \begin{remark}
    If only row covariates are retained alongside the low-rank term $\bM$, the model structure corresponds to that of \citet{mao_Matrix_2019}. If all regularization terms are removed, the structure corresponds to  the unpenalized method of \citet{ma_statistical_2025}.
%    \end{remark}
%\begin{remark}
\item In Lemma 8 of online Appendix D, we show that minimizing $\trace{\tr\bA\bSa\bA}+\trace{\tr\bB\bSb\bB}$ is equivalent to minimizing $\normnuc{\bSa^{1/2}\bM\bSb^{1/2}}$. This generalizes \eqref{eq:mazm_lemma6} to the case with similarity matrices.
\end{itemize}
\end{remark}
% \end{enumerate}

\subsection{Inference Process}

The objective function in \eqref{eq:obj} is not convex and does not admit a global closed-form solution. However, it is coordinate-wise convex, and we therefore adopt an alternating least-squares procedure in which we iteratively update each parameter while holding the others fixed.
Specifically, solving \eqref{eq:obj} with respect to $\bbeta$ or $\bgamma$ is a penalized least-squares problem, while solving it with respect to $\bA$ and $\bB$ corresponds to a \texttt{Soft-Impute} problem. We present the derivations of the updates for $\bbeta$ and $\bA$ below, leaving the other parameters to online Appendix A.

First, we rewrite the loss term \eqref{eq:obj:frob1} as
\begin{align}
    \label{eq:obj:frob11}
\losssf{\bstheta} \,  = \, \frac{1}{2}\normfrob{\bomega \circ \bracks{\bY-\btheta}} = \frac{1}{2}\normfrob{\bomega \circ {\bY} + \bomegac\circ\btheta - \btheta},
\end{align}
where $\bomegac$ denotes the complement of $\bomega$ (i.e., the indicator matrix of the missing entries). Now suppose we have estimates for $\btheta$ at iteration $\bracks{t-1}$, and we wish to compute the new estimate at the current iteration. Following \citet{hastie_matrix_2015}, we impute the missing entries at iteration $\bracks{t}$ using the estimate  $\steptmin\btheta$ from the previous iteration. With this approach, the loss term \eqref{eq:obj:frob11} at iteration $\bracks{t}$ is 
\begin{align} \label{eq:obj:frob2}
    \losssft{\bstheta} =\frac{1}{2} \normfrob{
    \bomega \circ (\bY - {\btheta}) +
    \bomegac \circ (\steptmin{\btheta}-{\btheta})},
\end{align}
where $\btheta$ is defined as in \eqref{eq:decomp}, with $\bM$ replaced by $\bA\tr\bB$. The objective function at iteration $(t)$ is:
\begin{align}\label{eq:obj2}
   \lossst{\bstheta} \, & = \, \losssft{\bstheta}  +\mathcal{P}(\bstheta).
\end{align}
 In the following, we use the formulation in \eqref{eq:obj:frob2} to derive the update steps.

\begin{proposition}[Alternating Least-Squares Updates for \texttt{IMR}]
\label{prop:ALS}

Consider the objective function in \eqref{eq:obj2} and the special case where $\bSa=\iden{n}$ and $\bSb = \iden{m}$. Define $\bE$ as the sparse matrix of training residuals, which is updated continuously throughout the algorithm. Given initial values 
{$\bstheta^{(0)}=(\bbeta^{(0)},\bgamma^{(0)},\bphia^{(0)},\bphib^{(0)}, \bA^{(0)},\bB^{(0)})$} and $\bE=\bomega\circ\bracks{\bY-\btheta^{(0)}}$, the alternating least-squares algorithm yields the following closed-form updates at iteration $\bracks{t}$:
\begin{align} \label{eq:update:beta}
    &\newv{\bbeta} = \softthresh{\lambb}{\tr\bX \bE + \steptmin\bbeta},\\ \label{eq:update:gamma}
    &\newv{\bgamma} = \softthresh{\lambg}{\bE\tr\bZ + \steptmin\bgamma} , \\ \label{eq:update:phia}
    &\newv{\bphia} =  \bracks{\bE \,{\bone}_m} \oslash \bracks{\bomega \,\bone_m} +\steptmin\bphia , \\\label{eq:update:phib}
    &\newv{\bphib} =  \bracks{\tr\bE {\bone}_n}\oslash \bracks{\tr\bomega \bone_n}  +\steptmin\bphib , \\\label{eq:update:A}
    &\newv{\bA} = \bracks{\bE{{\steptmin\bB}} + \steptmin\bA\tr{{\steptmin\bB}}{\steptmin\bB}}\bracks{\tr{{\steptmin\bB}}{\steptmin\bB}+\lambm \iden r}^{-1}, \\\label{eq:update:B}
    &\newv{{\bB}}  = \bracks{\tr{\bE}{{\stept\bA}} + \steptmin\bB\tr{{\stept\bA}}{\stept\bA}}\bracks{\tr{{\stept\bA}}{\stept\bA}+\lambm \iden r}^{-1},
\end{align}
where $\oslash$ is the Hadamard (element-wise) division and the soft-thresholding operator $\softthresh{\lambda}{x}$ is defined element-wise as: 
\begin{align*}
    \softthresh{\lambda}{x} = \begin{cases}
    x - \lambda & \text{if } x > \lambda, \\
    x + \lambda & \text{if } x < -\lambda, \\
    0 & \text{if } |x| \leq \lambda.
    \end{cases}
\end{align*}

The updates in \eqref{eq:update:beta}--\eqref{eq:update:B} are executed sequentially, and the residual matrix is updated in place immediately following the update of each parameter block. The exact update formulas for $\bE$ are provided in the proof.  Moreover, we replace \eqref{eq:update:A} and \eqref{eq:update:B} with equivalent yet computationally more efficient updates. See Remark \ref{remark:5} for further details.
\end{proposition}
\begin{proof}
    Below, we derive the updates for $\bbeta$ and $\bA$, deferring the derivations of the remaining parameters to online Appendix A.

\textbf{Updating $\bbeta$.}
The update for $\bbeta$ is obtained by solving
\begin{align*}
 \stept{\bbeta} &  = \arg\min_{\bbeta} \lossst{\bstheta^{(t-1)}_{[-\bbeta]}} \\
 & =  \arg\min_{\bbeta} \losssft{\bstheta^{(t-1)}_{[-\bbeta]}} + \lambb \normlass{\bbeta},
\end{align*}
where $\bstheta^{(t-1)}_{[-\bbeta]}=(\bbeta,\bgamma^{(t-1)}, \bphia^{(t-1)},\bphib^{(t-1)},\bA^{(t-1)},\bB^{(t-1)})$. The first term simplifies to
\begin{align*}
    & \losssft{\bstheta^{(t-1)}_{[-\bbeta]}}  \\
    = &\frac{1}{2} \normfrob{\bomega \circ (\bY - \bX{\bbeta} -\steptmin{\bgamma}{\bZ} - \steptmin{\bphia}\tr{\bone_m} - \bone_n \tr{\steptmin{\bphib}} - \steptmin\bM) +  \overline{\bomega}\circ\bracks{\bX\steptmin\bbeta-\bX\bbeta}} \\
  =& \frac{1}{2} \normfrob{\bE + \bomega \circ (  \bX\steptmin\bbeta - \bX{\bbeta} ) +  \overline{\bomega}\circ\bracks{\bX\steptmin\bbeta-\bX\bbeta}}\\
  %=& \frac{1}{2} \normfrob{\bomega \circ (\bE + \bX\steptmin\bbeta - \bX{\bbeta} ) +  \overline{\bomega}\circ(\bX\steptmin\bbeta-\bX\bbeta)} \\
  =& \frac{1}{2} \normfrob{\bE + \bX\bracks{\steptmin\bbeta - \bbeta} }.
\end{align*}
%Unlike the original loss term \eqref{eq:obj:frob1}, this formulation is differentiable with respect to $\bbeta$. 
Taking the gradient with respect to $\bbeta$ yields:
\begin{align*}
    \deriv{\lossst{\bstheta^{(t-1)}_{[-\bbeta]}}}{\bbeta}
    & = \deriv{\losssft{\bstheta^{(t-1)}_{[-\bbeta]}}}{\bbeta} + \lambb \deriv{\normlass{\bbeta}}{\bbeta} \\
    & = - \tr\bX \bE + \tr\bX\bX \bracks{\bbeta - \steptmin\bbeta} + 
    \lambb  \: \text{sign} \: \bracks{\bbeta}.% = 0.
\end{align*}
Setting this gradient to zero corresponds to a standard Lasso problem. A closed-form solution exists if the columns of $\bX$ are orthonormal, i.e., $\tr\bX\bX=\iden p$.
This can be achieved by orthonormalizing the columns of $\bX$. Let $\bX = \boldsymbol{QR}$ be the QR-decomposition of $\bX$, where $\boldsymbol{Q} \in \R{n}{p}$ has orthonormal columns and $\boldsymbol{R} \in \R{p}{p}$ is upper triangular. By defining transformed variables $\bX^*=\boldsymbol{Q}$ and $\bbeta^* = \boldsymbol{R}\bbeta$, we ensure that $\tr{\bX^*}\bX^*=\iden{p}$. Under this orthonormalization and suppressing the asterisks for notational simplicity, the normal equation becomes:
\begin{align*}
    \deriv{\lossst{\bstheta^{(t-1)}_{[-\bbeta]}}}{\bbeta}
    & = - \tr\bX \bE + {\bbeta - \steptmin\bbeta} + 
    \lambb  \: \text{sign} \: \bracks{\bbeta} = 0.
\end{align*}
Rearranging terms yields the closed-form soft-thresholding update:
\begin{align*}
    \stept\bbeta &=  \tr\bX \bE + \steptmin\bbeta - 
    \lambb  \: \text{sign} \: \bracks{\stept\bbeta} \\
        & = \softthresh{\lambb}{\tr\bX\bE + \steptmin\bbeta}.
\end{align*}

Finally, we update the residual matrix as follows:
\[\bE = \bE +\bomega \circ\bracks{\bX \bracks{\steptmin\bbeta-\stept\bbeta}}.\]

\textbf{Updating ${\bA}$.} Similarly, the update for $\bA$ is obtained by solving
\begin{align*}
 \stept{\bA} & = \arg\min_{\bA} \lossst{\bstheta^{(t-1)}_{[-\bA]}} \\
 & = \arg\min_{\bA} \losssft{\bstheta^{(t-1)}_{[-\bA]}} + \frac{\lambm}{2} \normfrob{\bA},
\end{align*}
where $\bstheta^{(t-1)}_{[-\bA]}=(\stept\bbeta,\bgamma^{(t)}, \bphia^{(t)},\bphib^{(t)},\bA,\bB^{(t-1)})$ is the set of most recently updated parameters before updating $\bA$ at step $t$.
First, for computational efficiency, we express $\bM=\bA\tr\bB$ in its singular value decomposition components: $\bM = \bU \bD^2 \tr{\bV}$, where $\bU$ and $\bV$ have orthonormal columns. We parameterize the factors as $\bA=\bU\bD$ and $\bB = \bV \bD$. Consequently, $\tr\bA\bA=\tr\bB\bB=\bD^2$, which is an $r\times r$ diagonal matrix whose inversion is trivial. The loss simplifies to 
\begin{align*}
    \losssft{\bstheta^{(t-1)}_{[-\bA]}}
  =& \frac{1}{2} \normfrob{\bE + \bracks{\steptmin\bA - \bA}\tr{\steptmin\bB} }.
\end{align*}
Differentiating the objective with respect to $\bA$ and setting the gradient to zero:
\begin{align*}
    \deriv{\lossst{\bstheta^{(t-1)}_{[-\bA]}}}{\bA}
    & = \deriv{\losssft{\bstheta^{(t-1)}_{[-\bA]}}}{\bA} +  \frac{\lambm}{2} \deriv{\normfrob{\bA}}{\bA} \\
    & = - \bE {\steptmin \bB} +
     \bracks{\bA - \steptmin \bA} \tr{\steptmin \bB} \steptmin \bB  + \lambm  \bA = 0.
\end{align*}
Solving for $\bA$ gives:
\begin{align} \label{eq:updateA0}
    \wtilde{\bA} &= \bracks{\bE \steptmin \bB + \steptmin \bA \tr{\steptmin \bB}\steptmin\bB} \bracks{\lambm  \iden r + \tr{\steptmin \bB}\steptmin\bB}^{-1}\\ \label{eq:updateA}
    & = \bracks{\bE \steptmin \bV  + \steptmin \bU \steptmin{\bD^2}} \steptmin\bD \bracks{\lambm \iden{r} +  \steptmin{\bD^2}}^{-1}.
\end{align}
\begin{remark}\label{remark:5}
     By construction, $\wtilde{\bA}$ does not have orthonormal columns. Orthonormality would not be required if formula \eqref{eq:updateA0} were used in lieu of the more computationally efficient formula \eqref{eq:updateA}. To preserve the orthonormality after evaluating \eqref{eq:updateA}, we perform an additional step. Let $\wtilde{\bA} = \wtilde{\bU}\wtilde{\bD}\tr{\wtilde{\bV}}$ be its SVD. We then set $\stept\bA = \wtilde{\bU}\wtilde{\bD}$ and update $\steptmin\bB=\steptmin\bB\wtilde{\bV}$. This preserves orthonormality without altering the product $\stept\bA\tr{\steptmin\bB}$. A similar strategy was employed by \citet{hastie_matrix_2015}. In practice, we directly update the matrices $\stept\bU$, $\stept\bD$, and $\stept\bV$ without explicitly evaluating $\stept\bA$ and $\stept\bB$. See online Appendix A for more details. After obtaining $\stept\bA$, we update the residual matrix as $\bE=\bE + \bomega\circ\bracks{\bracks{\steptmin\bA-\stept\bA}\tr{\steptmin\bB}}$.
\end{remark}
\end{proof}

\begin{proposition}[Update Formulas for $\bA$ and $\bB$ in the General Case]
Consider the general case incorporating similarity matrices $\bSa$ and $\bSb$ for the rows and columns, respectively. The updates \eqref{eq:update:A} and \eqref{eq:update:B} are modified as follows.

Define the matrices:
\begin{align*}
    \mathbf{W}_1 &= \bE\steptmin{\bB} + \steptmin{\bA}\tr{{\steptmin{\bB}}}\steptmin{\bB},\\
    \mathbf{W}_2 &= \tr{\bE}{{\stept\bA}} + \steptmin\bB\tr{{\stept\bA}}{\stept\bA}.
\end{align*}
Let $\steptmin d_j$ be the $j$-th diagonal element of ${\tr{\steptmin\bA}}{\steptmin{\bA}}$. The updates for the $j$-th columns of $\bA$ and $\bB$, denoted $\indexx{\bA}{.}{j}$ and $\indexx{\bB}{.}{j}$ for $j=1,\dots,r$, are given by:
\begin{align}
    \label{eq:update:Sr}
    \indexxt{\bA}{.}{j} =& \left( \lambm \bSa + \steptmin d_j \iden n\right)^{-1} \indexx{\mathbf{W}_1}{.}{j}, \\ 
    \label{eq:update:Sc}
    \indexxt{{{\bB}}}{.}{j} =& 
     \left(
    \lambm \bSb + \stept d_j\iden m\right)^{-1} \indexx{\mathbf{W}_2}{.}{j}.
    \end{align}
All other update equations from Proposition \ref{prop:ALS} remain unchanged. Furthermore, we replace the updates \eqref{eq:update:Sr} and \eqref{eq:update:Sc} with equivalent but more computationally efficient updates. This alternative formulation replaces the matrices $\bracks{\lambm \bSa + \steptmin d_j \iden n}$ and $\bracks{\lambm \bSb + \stept d_j\iden m}$ with diagonal matrices, thereby rendering their inversion trivial. Further details are provided in online Appendix A.
\end{proposition}
\begin{proof}
    See online Appendix A.
\end{proof}
We iterate these update equations until convergence, defined as the Frobenius norm of $\stept\btheta-\steptmin\btheta$ falling below a specified tolerance. Finally, we note that any subset of the parameters in $\bracksb{\bbeta,\bgamma,\bphia,\bphib,\bM}$ may be fixed to zero to define a valid submodel without affecting the estimation of the remaining parameters. %Remarks regarding computational efficiency are provided in online Appendix A.

The alternating least-squares algorithm requires initial values. A straightforward approach is to initialize $\bracksb{\bbeta,\bgamma,\bphia,\bphib}$ to zero and $\bracksb{\bA,\bB}$ to random orthonormal matrices. However, to accelerate convergence, we adopt a different strategy: we first fit the submodel $\bracksb{\bbeta,\bgamma,\bphia,\bphib}$ with $\bM$ fixed at zero until convergence. We use these estimates as initial values for the covariate coefficients and intercepts, and initialize $\bM$ to the rank-$r$ SVD of the resulting residual matrix $\bE$. This strategy converges to the same solution as the random initialization in a smaller number of iterations.

%=============================================================================

\subsection{Upper Bounds on the Estimation Errors}
In this section, we quantify the statistical accuracy of the \texttt{IMR} estimator. For clarity, the analysis focuses on the special case in which no structural dependency information is available (that is, $\bS_r=\iden n$ and $\bS_c=\iden m$). Online Appendix D establishes that Theorem \ref{theorem:1} continues to hold at the same rates in the general case, under a mild spectral condition on the similarity matrices. The proof of Theorem \ref{theorem:1} is deferred to online Appendix B.

We begin by introducing some additional notation. For a parameter matrix $\bA$, we write ${\truev{\bA}}$ for its true (unknown) value and  $\Delta\bA = \estim{\bA} - \truev\bA$ for the corresponding estimation error. For an $n\times m$ matrix, we measure accuracy through the mean squared error, defined as $\operatorname{MSE}\bracks{\Delta\bA}=\bracks{nm}^{-1}\normfrob{\Delta\bA}$. We use ${\|\bA\|}_0=\sum_{ij} \mathbb{I}\bracks{\bA_{ij}\neq 0}$ for the entry-wise $\ell_0$ norm, where $\mathbb{I}\bracks{\cdot}$ is the indicator function. For quantities $x$ and $y$, we use $x\vee y$ and $x \wedge y$ for the maximum and minimum, respectively. Furthermore, $x\asymp y$ (respectively, $x\lesssim y$) indicates that $x$ and $y$ are of the same order (respectively, $x\leq C y$) for a positive constant $C$ independent of the dimensions $n,m,p,q,r$ and the sampling probabilities $\pi_{ij}$.

We study the intercept-free target, $\truev\btheta = \bX\truev\bbeta+\truev\bgamma\bZ + \truev\bM$, and the following estimation problem
%$\truev\btheta = \bX\truev\bbeta+\truev\bgamma\bZ + \truev\bM$ and, 
 \begin{align} \label{problem:thm1}
    \bracks{\estim\bbeta,\estim\bgamma,\estim\bM} \in \underset{\substack{\norminf{\bbeta}\leq c_\beta,%\,\norminf{\bX\bbeta}\leq \cxb  
    \\ \norminf{\bgamma}\leq c_\gamma,%\,\norminf{\bgamma\bZ}\leq \cgz
    \\
    \norminf{\bM} \leq c_m}}{\operatorname{argmin}}
    \frac{1}{2}\normfrob{\bomega \circ (\bY-\btheta)} +\lambm \normnuc{\bM}  +
    \lambb \normlass{\bbeta} + 
    \lambg \normlass{\bgamma},
\end{align}
for positive constants $c_\beta,c_\gamma,c_m$, and we consider the following assumptions.
\begin{assumption} \label{A1}
     Without loss of generality, $\norminf{\bX} \le 1$ and $\norminf{\bZ} \le 1$, and let $s_\beta={\|\truev\bbeta\|}_0$ and $s_\Gamma={\|\truev\bgamma\|}_0$ denote the sparsity of $\truev\bbeta$ and $\truev\bgamma$, respectively. Moreover, there exist finite constants $c_\beta,c_\gamma,c_m,\cxb,\cgz$ such that 
     \[\norminf{\truev\bbeta} \le c_\beta,\quad \norminf{\truev\bgamma} \le c_\gamma,\quad  \norminf{\truev\bM} \le c_m,\quad \norminf{\bX\truev\bbeta}\leq \cxb ,\quad\norminf{\truev\bgamma\bZ}\leq\cgz.\] 
     We assume that $\max{(r,p,q)}< \min{(n,m)}$ and that $\bX$ and $\bZ$ are of full rank. Consequently, their Gram matrices are strictly positive definite, and their eigenvalues are bounded below by positive constants; that is, there exist positive constants $\kappa_x$ and $\kappa_z$ such that $\kappa_x \le n^{-1}\lambda_{\min}\bracks{\tr\bX\bX}$ and $\kappa_z \le m^{-1} \lambda_{\min}\bracks{\bZ\tr\bZ}$, where $\lambda_{\min}(\cdot)$ denotes the smallest eigenvalue. Finally, we define $d_x = \max_{1\le k\leq p}\sum_{i=1}^n|\bX_{ik}|$, $d_z = \max_{1\le k\leq q}\sum_{j=1}^m|\bZ_{kj}|$, $\sigma_x = \normop{\bX}$, and $\sigma_z=\normop{\bZ}$. 
\end{assumption}
\begin{assumption}
    The errors $\bracksb{\bepsilon_{ij}}$ are independent, zero-mean, $\sigma$-sub-Gaussian random variables.
\end{assumption}
\begin{assumption}\label{A3}
     The observation indicators $\bracksb{\bomega_{ij}}$ are independent $\operatorname{Bernoulli}\bracks{\pi_{ij}}$ variables that are independent of $\bracksb{\bepsilon_{ij}}$ and of $\truev\bstheta$. Furthermore, there exist constants $\pi_{\min}$ and $\pi_{\max}$ such that $0 < \pi_{\min} \le \pi_{ij} \le \pi_{\max} \le 1$ for all $(i,j)$.
\end{assumption}

\begin{theorem} \label{theorem:1}
    Let Assumptions \ref{A1}--\ref{A3} hold with $\bSa=\iden n$ and $\bSb = \iden m$, and set
    \begin{align*}
        \lambb &\asymp d_x + \sigma \sqrt{n\,\log\bracks{nm}}, \\
        \lambg &\asymp d_z  +\sigma \sqrt{m\,\log\bracks{nm}}, \\
        \lambm &\asymp \sigma_x + \sigma_z + \sigma \sqrt{\pi_{\max}\bracks{n\vee m}\log\bracks{n+m}}.
    \end{align*}
        The estimator \eqref{problem:thm1} satisfies each of the following with probability at least $1-8\,\bracks{n+m}^{-1}$:
    %Then, for some universal constant $c$,  with probability at least $1-8\,\bracks{n+m}^{-1}$, the estimator \eqref{problem:thm1} satisfies
    \begin{align*}
        (i)&\quad \mse{\Delta\bbeta} \lesssim \frac{\bracks{c_m+\cgz}s_\beta}{\pi_{\min}\,p\,m} + \frac{\sigma\,s_\beta}{\pi_{\min}\,p\,m} \sqrt{\frac{\log\bracks{nm}}{n}} + \frac{\pi_{\max}}{\pi_{\min}^2\bracks{n\wedge m}}, \\
        (ii)&\quad \mse{\Delta\bgamma} \lesssim \frac{\bracks{c_m+\cxb}s_\Gamma}{\pi_{\min}\,n\,q} + \frac{\sigma\,s_\Gamma}{\pi_{\min}\,n\,q} \sqrt{\frac{\log\bracks{nm}}{m}}+ \frac{\pi_{\max}}{\pi_{\min}^2\bracks{n\wedge m}}, \\
        (iii)&\quad \mse{\Delta\bM} \lesssim 
        \frac{r\bracks{\sigma^2_x+\sigma^2_z}}{\,\pi_{\min}^{2}\,n\,m}  +
        %\frac{r\,\pi_{\max}\,\sigma^2\,\bracks{n\vee m}\,\log\bracks{n+m}}{\pi_{\min}^2\,n\,m} = 
        \frac{r\,\pi_{\max}\,\sigma^2\,\log\bracks{n+m}}{\pi_{\min}^2\,\bracks{n\wedge m}}.% \\
    \end{align*}
\end{theorem}
\begin{remark}
{In each bound of Theorem \ref{theorem:1}, the first term reflects the cost of jointly estimating the three components. For example, in the error bound for $\bbeta$ in part (i), this term accounts for the additional uncertainty induced by the simultaneous estimation of $\bgamma$ and $\bM$. In (i) and (ii), this term converges to zero when the true parameter matrices $\bbeta$ and $\bgamma$ are sparse. The second terms in (i) and (ii) scale with the sub-Gaussian noise and vanish in both dense and sparse settings at rate at least $\mathcal{O}\bracks{\sigma,\pi_{\min}^{-1}\sqrt{\log(nm)/(n\wedge m)}}$. The third term in (i) and (ii) arises from lower-bounding the remainder in the first-order Taylor expansion of the loss function around the true parameters, which is related to the restricted strong convexity condition in standard matrix completion theory; see Online Appendix B for further details. Under the almost uniform sampling scheme, $\pi_{\max}\asymp\pi_{\min}$, this term vanishes at the rate $\mathcal{O}\bracks{\pi_{\min}^{-1}(n\wedge m)^{-1}}$. Finally, in (iii), because $\sigma^2_x=\mathcal{O}\bracks{n}$ and $\sigma^2_z=\mathcal{O}\bracks{m}$, whenever $\pi_{\max}\,\sigma^2\,\log\bracks{n+m}>1$, the bound matches the standard rate in matrix completion up to a logarithmic factor \citep{klopp_Matrix_2015}.}
\end{remark}

\section{Simulation Studies \label{sec:simulation}}

In this section, we conduct simulation experiments to evaluate the performance of the \texttt{IMR} framework relative to alternative methods. We consider two simulation settings.

\textbf{Setting 1.} This setting examines performance across different matrix dimensions, with $n=m\in\bracksb{400,600,800,1000}$. Data are generated from the model $\btheta=\bX\bbeta + \bM$, where $\bM = \bA\tr\bB$. The latent matrices $\bA \in \R{n}{r}$ and $\bB \in \R{m}{r}$ are drawn independently from $\operatorname{Uniform}(0,1)$, with $r=10$. The row covariate matrix $\bX \in \R{n}{p}$ (with $p=4$) consists of independent and identically distributed (i.i.d.) entries from $\operatorname{Uniform}(0,1)$. Its coefficient matrix $\bbeta \in \R{p}{m}$ is generated from a multivariate normal distribution, $ N_p(\mu_p, \Sigma_p)$,
where the entries of the mean vector $\mu_p$ are drawn uniformly from $\{x: |x| \in [0.1,1]\}$, and $\Sigma_p$ is a diagonal covariance matrix with entries drawn from $\operatorname{Uniform}(0.025, 1)$. Consequently, the true rank of $\btheta$ is $p+r=14$.

\textbf{Setting 2.} This setting evaluates the impact of missingness rates ranging from $70\%$ to $95\%$ and compares the computational efficiency of  \texttt{IMR}  against \texttt{Soft-Impute}. We fix the dimensions at $n=m=1000$ and generate data from the model $\btheta=\bX\bbeta+\bgamma\bZ+\bM$. The matrices $\bX$, $\bbeta$, and $\bM$ are generated as in Setting 1 (with $r=5$ and $p=5$). The column covariates $\bZ \in \R{q}{m}$ are drawn from $\operatorname{Uniform}(0,1)$, and $\bgamma \in \R{n}{q}$ is generated from the same distribution as $\bbeta$, with $q=5$. Here, the true rank of $\btheta$ becomes $r+p+q=15$.

In both settings, the observed matrix is given by $\bY = \bomega \circ \bracks{\btheta+\bepsilon}$, where $\bomega$ is the indicator matrix of observed entries and $\bepsilon \in \R{n}{m}$ represents noise. The noise entries are i.i.d. Gaussian, $\bepsilon_{ij} \sim \mathcal{N}(0, \sigma^2)$. The noise variance, $\sigma^2$, is chosen such that the signal-to-noise ratio (SNR) is one: $\operatorname{SNR}= \sqrt{\operatorname{Signal}(\btheta)/\sigma^2}=1$, where $\operatorname{Signal}(\btheta) = \sum_{i=1}^n\sum_{j=1}^m \bracks{\btheta_{ij} - \overline{\btheta}}^2/\bracks{nm-1}$ and $\overline{\btheta} = \sum_{i=1}^n\sum_{j=1}^m \btheta_{ij}/\bracks{nm}$.

The missingness indicator $\bomega$ is generated via Bernoulli sampling. In Setting 1, entries are observed with probability $0.2$ ($80\%$ missingness). In Setting 2, we vary the missingness rate from $70\%$ to $95\%$ in $5\%$ increments. To ensure comparability across rates within Setting 2, we first generate a dataset with $70\%$ missingness and incrementally remove observed entries to increase missingness rates. Consequently, the test set (the initially unobserved $70\%$) remains fixed across varying sparsity levels. We perform $500$ independent replicates for each setting.

Performance is evaluated using the relative root mean squared error (RRMSE) of the estimators $\estim\bbeta$, $\estim\bM$, ${\estim\btheta}_{\bomega}$ (training estimates), and ${\estim\btheta}_{\bomegac}$ (test estimates). We define 
\(\operatorname{RRMSE}(\bbeta) = {\normf{\truev\bbeta-\hat \bbeta}}\,\big/\,{\normf{\truev\bbeta}}\) and proceed similarly for the other parameters.
  
We compare the performance of our proposed method with two alternative methods: \texttt{Soft-Impute}\footnote{\texttt{R} package: \href{https://cran.r-project.org/package=softImpute}{cran.r-project.org/package=softImpute}} \citep{hastie_matrix_2015} and  \texttt{MCCI}\footnote{Authors' implementation: \href{https://github.com/mxjki/Matrix_Completion_For_Complex_Survey_with_Multivariate_Missingness}{github.com/mxjki}} \citep{mao_Matrix_2019}. We use the authors' publicly available implementations with tuning parameters selected via cross-validation.

Table \ref{tab:sim1} summarizes the simulation results under the first setting, wherein the proposed method demonstrates superior performance relative to \texttt{Soft-Impute} and \texttt{MCCI}. Furthermore,  \texttt{IMR}  yields more accurate rank estimates than the competing methods, achieving exact rank recovery (with zero variance across all 500 replicates) when the matrix dimension $n=m$ exceeds $400$. Results under the second setting are presented in Figures \ref{fig:sim2_res} and \ref{fig:sim2_time}. The results for \texttt{MCCI} are omitted, as it was outperformed by the other two methods in both accuracy and computational efficiency. \texttt{IMR}  consistently achieves lower RRMSE than \texttt{Soft-Impute}, and both methods exhibit reduced error as the sparsity level decreases. Moreover,  \texttt{IMR}  achieves exact rank recovery (zero variance) for sparsity levels up to $90\%$. In contrast, \texttt{Soft-Impute} consistently overestimates the rank. 

Despite its added model complexity, \texttt{IMR} is at least four times faster than \texttt{Soft-Impute}. For instance, at an $80\%$ sparsity level, \texttt{IMR} achieves a $66\%$ reduction in RRMSE relative to \texttt{Soft-Impute} while being seven times faster, even though both methods require a comparable number of iterations to converge. 

\begin{table}[!tb]
\centering
\caption{\label{tab:sim1}Empirical relative root mean squared errors (RRMSEs), estimated ranks, and standard deviations (in parentheses) under the model $\btheta=\bX\bbeta+\bM$. We consider dimensions $(n,m)=(400,400),(600,600),(800,800),(1000,1000)$ and a fixed missingness rate of $80\%$. The true rank of $\btheta$ is 14.}
\centering
%\fontsize{12}{14}\selectfont
\small
\begin{tabular}[t]{>{}lrrr>{}rr}
\toprule
Model & RRMSE($\beta$) & RRMSE($M$) & RRMSE($\btheta_{\text{train}}$) & RRMSE($\btheta_\text{test}$) & Rank($\btheta$)\\ 
\midrule
\addlinespace[0.3em]
\multicolumn{6}{l}{n = m = 400}\\
\textcolor{black}{\textbf{{\hspace{1em}IMR}}} & \textcolor{black}{{0.709 (0.031)}} & \textcolor{black}{{0.838 (0.059)}} & \textcolor{black}{{0.415 (0.054)}} & \textcolor{black}{\textbf{{0.47 (0.062)}}} & \textcolor{black}{{14.76 (2.169)}}\\
\textcolor{black}{\textbf{\hspace{1em}Soft-Impute}} & \textcolor{black}{—} & \textcolor{black}{—} & \textcolor{black}{0.456 (0.059)} & \textcolor{black}{0.521 (0.07)} & \textcolor{black}{13.54 (1.888)}\\
\textcolor{black}{\textbf{{\hspace{1em}MCCI}}} & \textcolor{black}{{0.913 (0.077)}} & \textcolor{black}{{0.998 (0.006)}} & \textcolor{black}{{0.569 (0.046)}} & \textcolor{black}{{0.576 (0.05)}} & \textcolor{black}{{10.322 (1.399)}}\\
\addlinespace[0.3em]
\multicolumn{6}{l}{n = m = 600}\\
\textcolor{black}{\textbf{\hspace{1em}IMR}} & \textcolor{black}{0.553 (0.022)} & \textcolor{black}{0.647 (0.029)} & \textcolor{black}{0.331 (0.043)} & \textcolor{black}{\textbf{0.358 (0.047)}} & \textcolor{black}{14 (0)}\\
\textcolor{black}{\textbf{{\hspace{1em}Soft-Impute}}} & \textcolor{black}{{—}} & \textcolor{black}{{—}} & \textcolor{black}{{0.4 (0.055)}} & \textcolor{black}{{0.446 (0.063)}} & \textcolor{black}{{14.378 (1.803)}}\\
\textcolor{black}{\textbf{\hspace{1em}MCCI}} & \textcolor{black}{0.745 (0.061)} & \textcolor{black}{0.985 (0.029)} & \textcolor{black}{0.498 (0.044)} & \textcolor{black}{0.507 (0.045)} & \textcolor{black}{12.462 (5.069)}\\
\addlinespace[0.3em]
\multicolumn{6}{l}{n = m = 800}\\
\textcolor{black}{\textbf{{\hspace{1em}IMR}}} & \textcolor{black}{{0.47 (0.019)}} & \textcolor{black}{{0.562 (0.024)}} & \textcolor{black}{{0.29 (0.037)}} & \textcolor{black}{\textbf{{0.308 (0.039)}}} & \textcolor{black}{{14 (0)}}\\
\textcolor{black}{\textbf{\hspace{1em}Soft-Impute}} & \textcolor{black}{—} & \textcolor{black}{—} & \textcolor{black}{0.362 (0.049)} & \textcolor{black}{0.396 (0.055)} & \textcolor{black}{15.194 (1.895)}\\
\textcolor{black}{\textbf{{\hspace{1em}MCCI}}} & \textcolor{black}{{0.643 (0.055)}} & \textcolor{black}{{0.943 (0.062)}} & \textcolor{black}{{0.455 (0.035)}} & \textcolor{black}{{0.462 (0.035)}} & \textcolor{black}{{17.698 (8.802)}}\\
\addlinespace[0.3em]
\multicolumn{6}{l}{n = m = 1000}\\
\textcolor{black}{\textbf{\hspace{1em}IMR}} & \textcolor{black}{0.417 (0.016)} & \textcolor{black}{0.503 (0.022)} & \textcolor{black}{0.261 (0.034)} & \textcolor{black}{\textbf{0.274 (0.036)}} & \textcolor{black}{14 (0)}\\
\textcolor{black}{\textbf{{\hspace{1em}Soft-Impute}}} & \textcolor{black}{{—}} & \textcolor{black}{{—}} & \textcolor{black}{{0.333 (0.047)}} & \textcolor{black}{{0.359 (0.052)}} & \textcolor{black}{{15.69 (1.896)}}\\
\textcolor{black}{\textbf{\hspace{1em}MCCI}} & \textcolor{black}{0.575 (0.049)} & \textcolor{black}{0.895 (0.069)} & \textcolor{black}{0.422 (0.03)} & \textcolor{black}{0.427 (0.031)} & \textcolor{black}{20.754 (8.168)}\\
\bottomrule
\end{tabular}
\end{table}

\begin{figure}[!tb]
    \centering
    \includegraphics[width=1\linewidth]{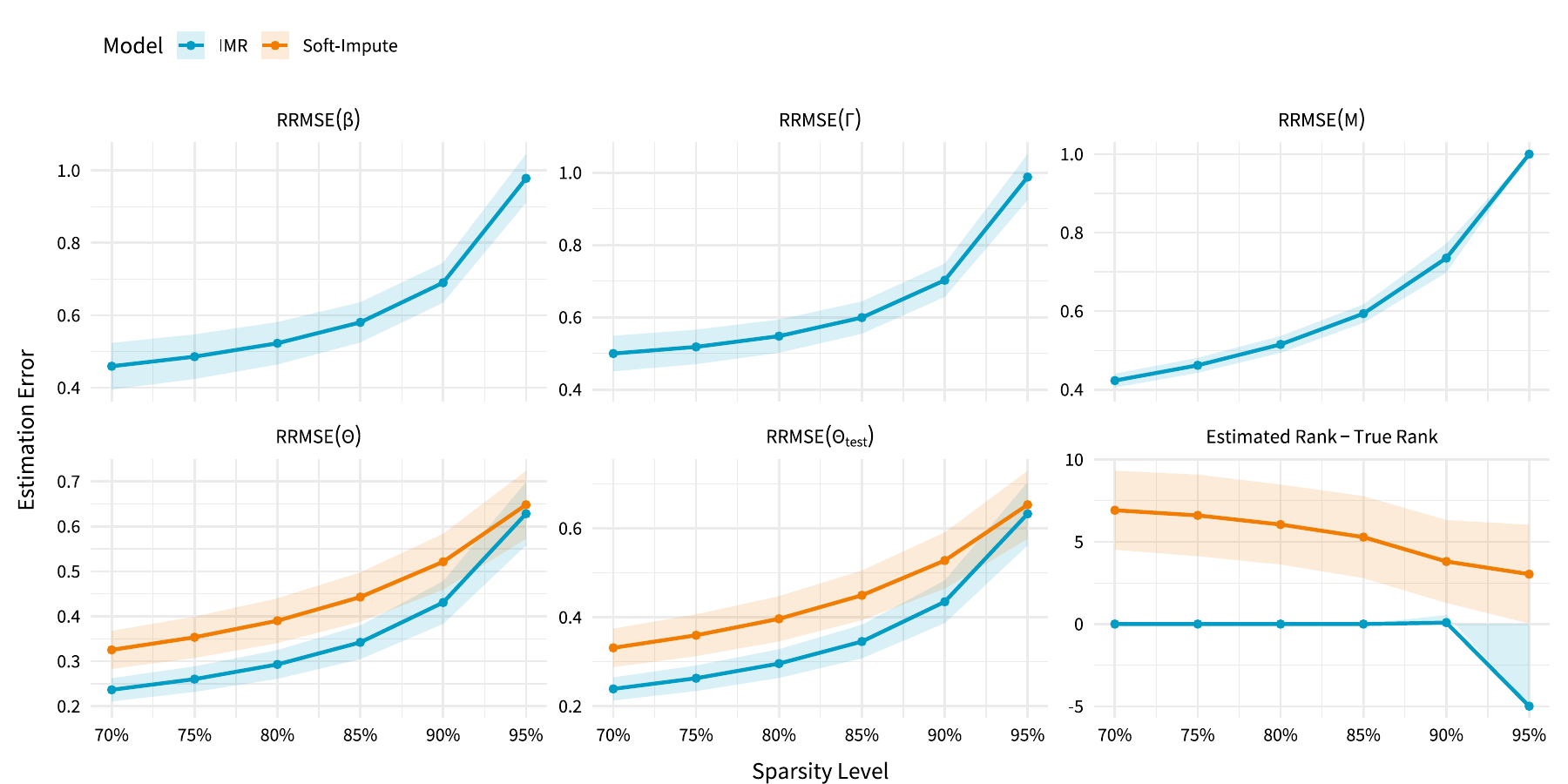}
    \caption{\textbf{Estimation performance and rank recovery in Simulation Setting 2.} Results represent the mean (solid lines) $\pm 1$ standard deviation (shaded regions) over $500$ independent replications. The true rank of $\btheta$ is fixed at $15$.}
    \label{fig:sim2_res}
\end{figure}
\begin{figure}[!tb]
    \centering
    \includegraphics[width=1\linewidth]{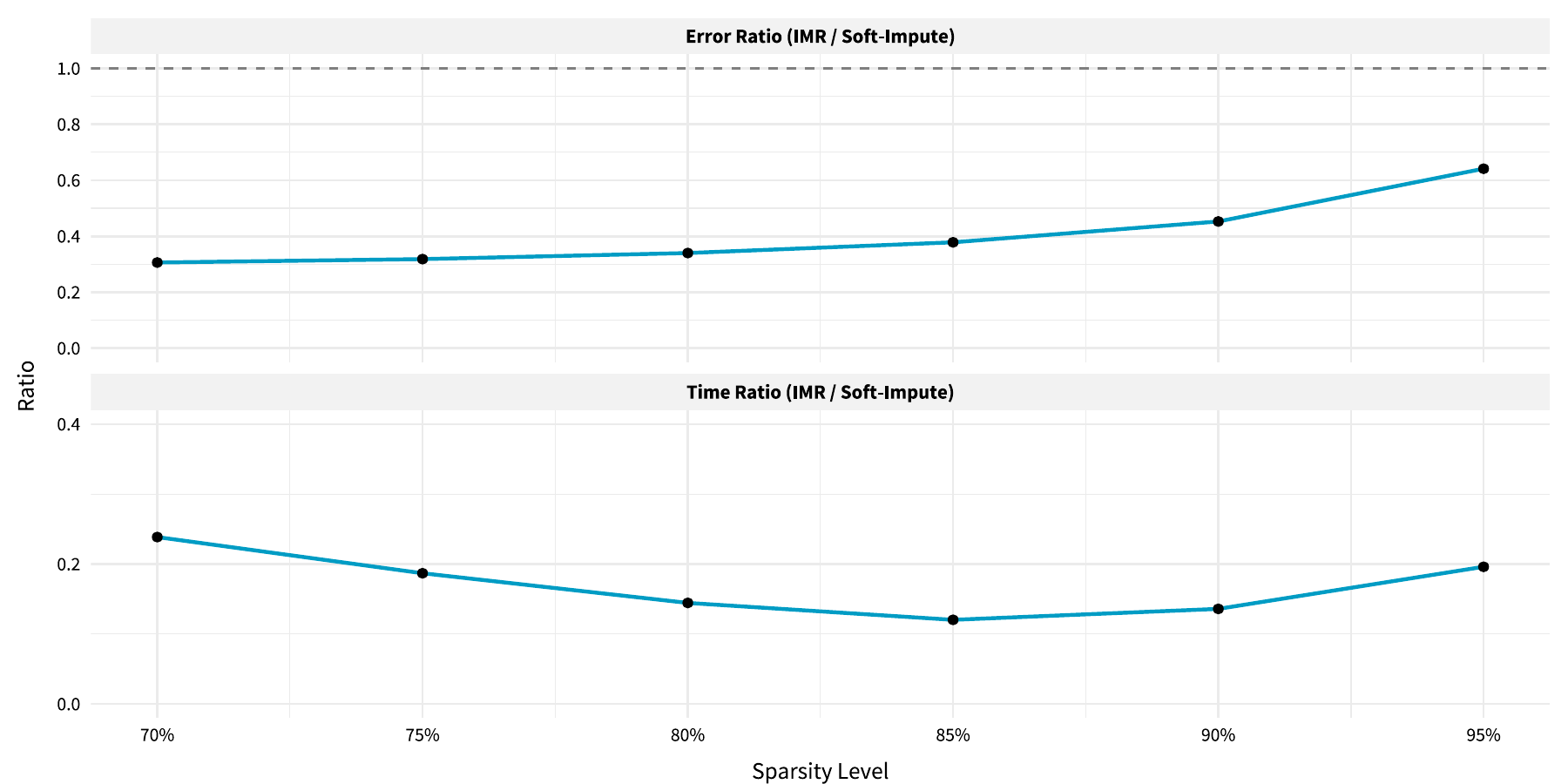}
    \caption{\textbf{Computational cost and performance relative to \texttt{Soft-Impute} (Simulation Setting 2).} Benchmarks were performed on a single CPU core, with results averaged over $500$ independent replications using fixed pre-tuned hyperparameters. \textbf{Top:} Ratio of test RRMSEs, ${\operatorname{RRMSE}_{\texttt{IMR}}}/{\operatorname{RRMSE}_{\texttt{Soft-Impute}}}$.  \textbf{Bottom:} Relative computational cost, ${\operatorname{Time}}_{\texttt{IMR}}/\operatorname{Time}_{\texttt{Soft-Impute}}$ (values below one favor \texttt{IMR}).}
    \label{fig:sim2_time}
\end{figure}
%========================================================
\section{Empirical Applications \label{sec:application}}

In this section, we apply our proposed method to two real-world applications: movie recommendation and bike-sharing demand modeling. For the first application, we use the MovieLens 1M dataset.\footnote{website: \href{https://grouplens.org/datasets/movielens/}{https://grouplens.org/datasets/movielens/}} MovieLens is a website where users rate movies and receive recommendations based on their rating history. This dataset contains 1,000,209 ratings for 3,952 movies by 6,040 users. User demographic information and movie genres are also available. About $96\%$ of the entries in the rating matrix are missing. 

The second application concerns spatiotemporal bike-sharing trips collected from BIXI\footnote{website: \href{https://bixi.com/}{https://bixi.com/}}, a docked bike-sharing service in Montreal, Canada. We use the data compiled by \citet{lei_scalable_2025}, which contain daily departure counts for each of 587 stations over 196 days (April 15 to October 27, 2019). The data also contain the geographic coordinates of the stations and time indices that we will use to construct the similarity matrices.

\subsection{Application to MovieLens 1M}

We use $\bY$ to denote the observed user--movie rating matrix of size $6040\times 3952$, where the $i$-th row corresponds to user $i$ and the $j$-th column corresponds to movie $j$. We include user gender and age as row covariates. Gender is encoded as 0 (female) or 1 (male), and age is grouped into four categories (0--24, 25--34, 35--49, 50+), represented by three dummy variables. These four variables are stored in $\bX \in \R{6040}{4}$, the row covariate matrix. We use movie genres as column covariates, which are encoded as binary indicator variables ($\bZ \in \R{18}{3952}$). A test set of 60,400 ratings (approximately $6\%$ of observed ratings) was held out, spanning all 6,040 users and a subset of 2,464 movies. We use the same train/test split as \citet{ma_statistical_2025}.

We consider two variants of our \texttt{IMR} framework: an intercept-only model (\texttt{IMR-I}) and a full model with intercepts and both row and column covariates (\texttt{IMR-IXZ}). These are compared with three existing methods: \texttt{Soft-Impute} \citep{hastie_matrix_2015}, the method proposed in \citet{ma_statistical_2025} (denoted \texttt{MCAI}), and \texttt{GLocal-K} \citep{han_GLocalK_2021}. The latter is a matrix completion framework that employs autoencoders with convolution kernels and outperforms the state-of-the-art collaborative-filtering baselines on the MovieLens 1M dataset, according to the results reported by its authors. We perform cross-validation to select the hyperparameters for \texttt{IMR} and \texttt{Soft-Impute}. For \texttt{MCAI}, we use the authors' implementation and hyperparameter settings, and we reproduce their reported results. For \texttt{GLocal-K}, we use the authors' implementation on MovieLens without modifications.

For evaluation, we report: (i) the root mean squared error (RMSE) on the training and test sets; (ii) the Pearson correlation between true and predicted ratings in the test set; (iii) the runtime (in minutes); (iv) rank estimates; and (v) the sparsity (proportion of zero entries) of the estimated covariate coefficient matrices. Table \ref{tab:movielens:results} summarizes these results. 

The full \texttt{IMR} model (\texttt{IMR-IXZ}) achieves a test RMSE that is $4.7\%$ lower than that of the \texttt{MCAI} model and $9.2\%$ lower than that of \texttt{Soft-Impute}. \texttt{GLocal-K} yields a further $0.9\%$ improvement in test RMSE over  \texttt{IMR-IXZ}, but  \texttt{IMR-IXZ}  is $99$ times faster. Moreover, the intercept-only variant (\texttt{IMR-I}) is $524$ times faster than \texttt{GLocal-K} while sacrificing only $1.6\%$ in test RMSE. The rank of the fitted matrix $\bM$ in  \texttt{IMR-IXZ}  is 11 (and 13 in  \texttt{IMR-I}), which is substantially smaller than the rank obtained by \texttt{GLocal-K} (63). Finally, although the row and column coefficient matrices in  \texttt{IMR-IXZ}  are full rank, they are $54.7\%$ and $73.8\%$ sparse, respectively.

To examine the interpretability of the covariate effects in the proposed method, we use the full model (\texttt{IMR-IXZ}) to analyze the distribution of predicted ratings across demographic groups for three selected movies from the children's genre. Figure \ref{fig:movielens:full} displays box plots of the estimated ratings for each movie, stratified by age and gender. The model reveals pronounced differences across age and gender groups, and these differences vary by movie. In general, female users exhibit higher predicted ratings than male users for these titles; regarding age, older groups (35+) are associated with higher ratings than younger groups. In contrast, for \textit{Home Alone}, the disparity between male and female ratings is negligible, likely reflecting the movie's broad appeal. Despite the high sparsity of the learned covariate coefficient matrices, the incorporation of demographic information induces noticeable shifts in the predictions and effectively captures preference heterogeneity across age and gender groups.

\begin{table}[!tb]
\centering 
\caption{Performance comparison on the MovieLens 1M dataset. Best values\\ per column are shown in bold, and \texttt{IMR}  models are shaded.} 
\label{tab:movielens:results}
%\resizebox{\ifdim\width>\linewidth\linewidth\else\width\fi}{!}
{
\fontsize{8}{10}\selectfont
\begin{tabular}[t]{>{}p{2.2cm}lllllllll}
\toprule 
\multicolumn{2}{c}{ } & \multicolumn{3}{c}{Performance} & \multicolumn{3}{c}{Rank Estimation} & \multicolumn{2}{c}{Sparsity} \\
\cmidrule(l{3pt}r{3pt}){3-5} \cmidrule(l{3pt}r{3pt}){6-8} \cmidrule(l{3pt}r{3pt}){9-10}
\multicolumn{2}{c}{ } & \multicolumn{2}{c}{Test} & \multicolumn{1}{c}{Train} & \multicolumn{5}{c}{ } \\
\cmidrule(l{3pt}r{3pt}){3-4} \cmidrule(l{3pt}r{3pt}){5-5}
Model & Time (min) & RMSE & Correlation & RMSE & $M$ & $\beta$ & $\Gamma$ & $\beta$ & $\Gamma$\\
\midrule
GLocal-K & 52.36 & \textbf{0.852} & \textbf{0.628} & \textbf{0.702} & 63 & — & — & — & —\\
\cellcolor[HTML]{f7f7f7}{IMR-IXZ} & \cellcolor[HTML]{f7f7f7}{0.53} & \cellcolor[HTML]{f7f7f7}{0.860} & \cellcolor[HTML]{f7f7f7}{0.593} & \cellcolor[HTML]{f7f7f7}{0.766} & \cellcolor[HTML]{f7f7f7}{$11$} & \cellcolor[HTML]{f7f7f7}{$5^{a}$} & \cellcolor[HTML]{f7f7f7}{$19^{a}$} & \cellcolor[HTML]{f7f7f7}{54.7\%} & \cellcolor[HTML]{f7f7f7}{73.8\%}\\
\cellcolor[HTML]{f7f7f7}{IMR-I} & \cellcolor[HTML]{f7f7f7}{\textbf{0.10}} & \cellcolor[HTML]{f7f7f7}{0.866} & \cellcolor[HTML]{f7f7f7}{0.587} & \cellcolor[HTML]{f7f7f7}{0.780} & \cellcolor[HTML]{f7f7f7}{13} & \cellcolor[HTML]{f7f7f7}{$1^{a}$} & \cellcolor[HTML]{f7f7f7}{$1^{a}$} & \cellcolor[HTML]{f7f7f7}{—} & \cellcolor[HTML]{f7f7f7}{—}\\
MCAI & 27.70 & 0.902 & 0.572 & 0.843 & {2} & $5^{a}$ & — & 11.2\% & —\\
Soft-Impute & 11.55 & 0.947 & 0.565 & 0.772 & 9 & — & — & — & —\\
\bottomrule
\end{tabular}}
\begin{tablenotes}[flushleft]\footnotesize
\item \textit{Note:}~$^{a}$~The additional unit in the estimated rank corresponds to the intercept vector.
\end{tablenotes}
\end{table}

\begin{figure}[!tb]
    \centering
    \includegraphics[width=1\linewidth]{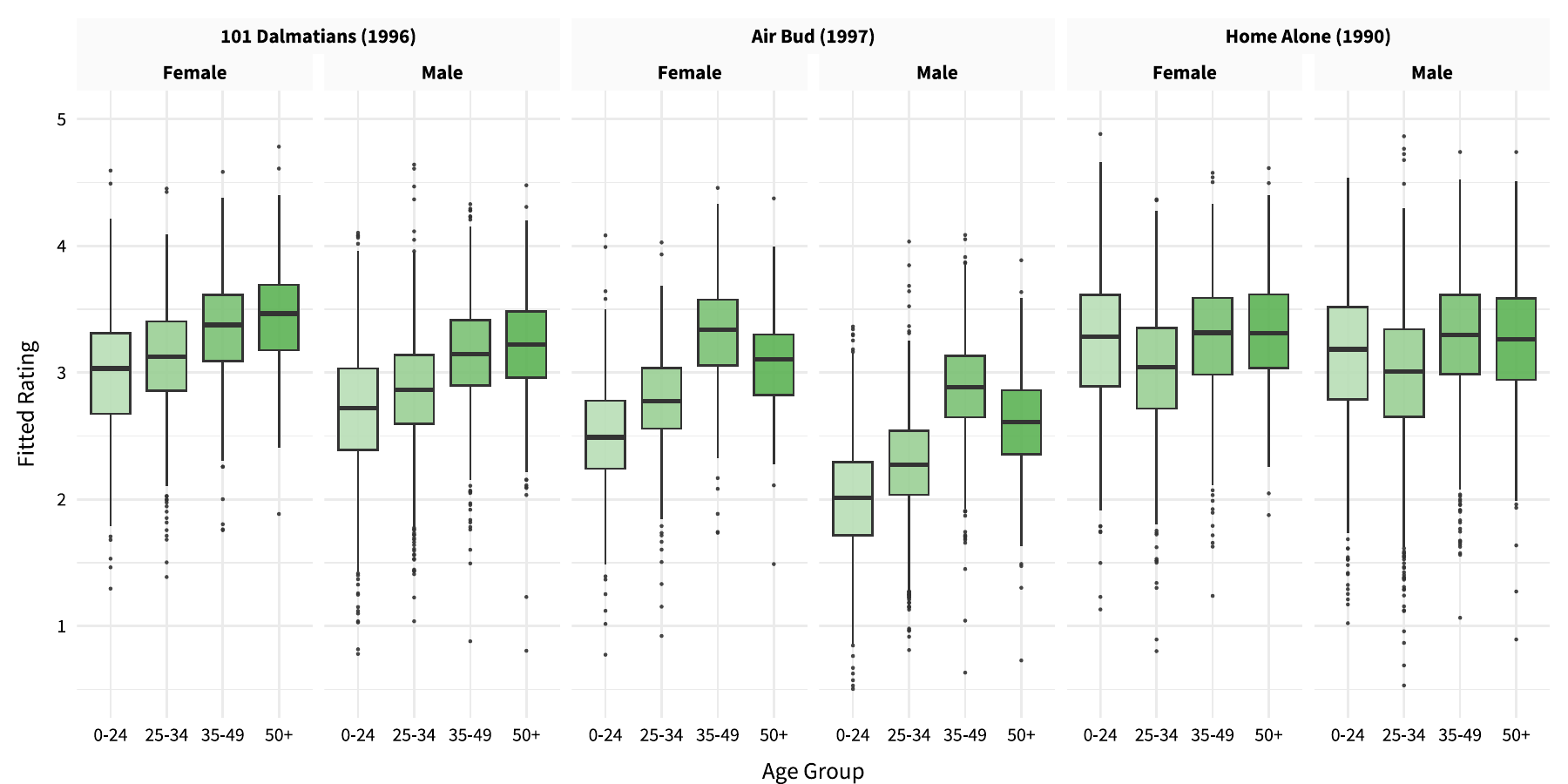}
    \caption{Box plots of estimated ratings from the full model (\texttt{IMR-IXZ}) on three movies, grouped by age and gender.}
    \label{fig:movielens:full}
\end{figure}

%=====================================================================
%=====================================================================
%===================================================================

\subsection{Application to  BIXI}

The BIXI dataset is represented by a ${196}\times{587}$ matrix $\bY$, where $\bY_{ij}$ denotes the number of departures on day $i$ at station $j$. Approximately $11\%$ of the entries in $\bY$ are missing. Furthermore, we held out $14\%$ of the observed entries at random to serve as a test set. The model is trained on a progressively larger fraction of the remaining data: $55\%$, $60\%$, $65\%$, $70\%$, and $75\%$, while the test set remains fixed to ensure comparability across different training sample sizes.  

This dataset exhibits strong dependence across rows (temporal dependence of the daily counts) and across columns (spatial dependence across stations). To account for this correlation structure, we consider two settings for the similarity matrices. In the first setting, denoted \texttt{IMR-S}, we use a Matérn $5/2$ kernel for the spatial dimension (stations) and a product of a squared-exponential kernel and a periodic kernel for the temporal dimension (days), following \citet{lei_scalable_2025}. In the second setting, denoted  \texttt{IMR-N}, we assume identity similarity matrices (i.e., no structural dependencies). Both variants include only the low-rank component; that is, they fit $\estim\bY=\estim\bM$. We compare their performance against the Bayesian Kernelized Tensor Regression (\texttt{BKTR}) model proposed by \citet{lei_scalable_2025}, which also captures row and column correlations.

We replicate the train/test split 50 times and report the average RRMSE on the test set, alongside the average training time. For hyperparameters, we specify $\bracks{\lambm=0.8,r=11}$ for both \texttt{IMR-S} and \texttt{IMR-N}. For \texttt{BKTR}, we adopt the settings used in \citet{lei_scalable_2025}. 

Table \ref{tab:bixi:res} summarizes the results. We observe that \texttt{IMR-S} achieves lower RRMSE values than \texttt{BKTR}, with all models showing improved performance as the training size increases. Furthermore, \texttt{IMR-S} consistently outperforms \texttt{IMR-N}, confirming that the specified similarity matrices improve predictive performance. Notably, while the Markov chain Monte Carlo (MCMC)-based \texttt{BKTR} requires approximately 26 minutes for training, both  \texttt{IMR} models converge in less than a third of a second.

\begin{table}[!tb]
\small
\caption{Predictive performance and computational efficiency on the BIXI dataset. Results are averaged over 50 independent train/test splits, with the test fraction fixed at $14\%$, and are reported as mean (standard deviation).}
 \label{tab:bixi:res}
\centering
\begin{tabularx}{\linewidth}{@{} l *{3}{X} *{3}{X} @{}}
\toprule
 & \multicolumn{3}{c}{{\bfseries Test RRMSE}} & \multicolumn{3}{c}{{\bfseries Computation Time (Seconds)}} \\
\cmidrule(lr){2-4} \cmidrule(lr){5-7}
{\bfseries Train Size} & {\bfseries \texttt{BKTR}} & {\bfseries \texttt{IMR-S} } & {\bfseries \texttt{IMR-N} } & {\bfseries \texttt{BKTR}} & {\bfseries \texttt{IMR-S} } & {\bfseries \texttt{IMR-N} } \\
\midrule\addlinespace[2.5pt]
\textbf{55\%} & 0.1760 (0.0015) & \textbf{0.1759 (0.0014)} & 0.1779 (0.0015) & 1480.18 (227.49) & \textbf{0.28 (0.10)} & \textbf{0.28 (0.11)}\\
\textbf{60\%} & 0.1751 (0.0015) & \textbf{0.1749 (0.0014)} & 0.1766 (0.0015) & 1524.50 (227.76) & \textbf{0.28 (0.09)} & \textbf{0.27 (0.11)}\\
\textbf{65\%} & 0.1744 (0.0015) & \textbf{0.1739 (0.0014)} & 0.1754 (0.0015) & 1614.40 (248.81) & \textbf{0.29 (0.11)} & \textbf{0.27 (0.09)}\\
\textbf{70\%} & 0.1738 (0.0015) & \textbf{0.1730 (0.0014)} & 0.1743 (0.0014) & 1667.30 (223.36) & \textbf{0.27 (0.04)} & \textbf{0.29 (0.15)}\\
\textbf{75\%} & 0.1735 (0.0016) & \textbf{0.1723 (0.0014)} & 0.1733 (0.0014) & 1718.72 (253.08) & \textbf{0.29 (0.08)} & \textbf{0.28 (0.11)}\\
\bottomrule
\end{tabularx}
\end{table}

\section{Discussion}\label{sec:conc}

We have illustrated that the proposed framework accommodates diverse types of side information and achieves competitive predictive performance at a substantially lower computational cost. Several extensions remain open and would further strengthen our framework. First, to preserve computational efficiency, we orthonormalize the covariate matrices through a QR decomposition rather than solving the Lasso in the original feature space. Although this reparameterization alters the geometry of the $\ell_1$ penalty, its primary purpose of guarding against over-parameterization is retained, and the estimator continues to induce sparsity in the original coefficient matrices, as the MovieLens application illustrates.

Second, the $\ell_1$ and nuclear-norm penalties, together with the presence of missing data, introduce bias into the estimates of the covariate coefficients. An important direction for future research is to develop debiasing techniques that would allow us to construct confidence intervals and perform hypothesis tests. Third, although we establish nonasymptotic upper bounds on the estimation errors, we do not derive matching minimax lower bounds tailored to the proposed estimator; to the best of our knowledge, no such bounds currently exist for this class of models. We intend to build on the present analysis to obtain tighter upper bounds and corresponding minimax lower bounds. Finally, our development assumes that entries are missing completely at random. Relaxing this assumption to accommodate informative sampling is another avenue for broadening the applicability of the framework.

\if1\anon
\section*{Funding}
This work was supported by the Natural Sciences and Engineering Research Council of Canada (NSERC) under the Canada Graduate Scholarship -- Doctoral program.
\fi

\section*{Disclosure Statement}\label{disclosure-statement}
The authors report there are no competing interests to declare.

\section*{Software and Data Availability Statement}\label{data-availability-statement}
\if1\anon
The proposed method is implemented in an \texttt{R} package. This package, along with the accompanying vignette is publicly available at \url{https://github.com/khaledfouda/IMR}. The datasets analyzed in Section \ref{sec:application} and the code required to reproduce the results presented in Sections \ref{sec:simulation} and \ref{sec:application} are publicly available at \url{https://github.com/khaledfouda/imr-reproducible-code}.
\fi
\if0\anon
The proposed method is implemented in an \texttt{R} package, which together with its accompanying vignette is provided in the Supplementary Materials. The datasets analyzed in Section \ref{sec:application} and the code required to reproduce the results presented in Sections \ref{sec:simulation} and \ref{sec:application} are also provided in the Supplementary Materials.
\fi

\phantomsection\label{supplementary-material}
\bigskip
\begin{center}
{\large\bf SUPPLEMENTARY MATERIAL}
\end{center}
\begin{description}
\item[Appendices:]
Contains proofs for all theoretical results in Section \ref{sec:method}. (\texttt{.pdf} file)
\item[\texttt{R}-package:] Implements the proposed methodology and includes a vignette detailing hyperparameter tuning. (\texttt{.zip} file)
\item[Reproduction Code:] Contains the code and datasets required to reproduce the tables and figures presented in Sections \ref{sec:simulation} and \ref{sec:application}. (\texttt{.zip} file)
\end{description}

  \bibliography{main_biblio}

\end{document}

\def\spacingset#1{\renewcommand{\baselinestretch}%
{#1}\small\normalsize} \spacingset{1}

%%%%%%%%%%%%%%%%%%%%%%%%%%%%%%%%%%%%%%%%%%%%%%%%%%%%%%%%%%%%%%%%%%%%%%%%%%%%%%

\if1\cAnon
{
  \title{\bf Supplementary Material for ``Incomplete Matrix Regression''}
  \author{Khaled Fouda\footnote{Correspondence to \href{khaled.fouda@hec.ca}{khaled.fouda@hec.ca}.},\\
    % Department of Decision Sciences, HEC Montréal,\\ Montreal, Quebec, H3T 2A7, Canada\\
    % \and
    Aurélie Labbe, \\
    Department of Decision Sciences, HEC Montréal,\\ Montreal, Quebec, H3T 2A7, Canada\\
    \and
    %and \\
    Karim Oualkacha \\
    Department of Mathematics, University of Quebec in Montreal,\\ Montreal, Quebec, H2X 3Y7, Canada
    }
  \maketitle
} \fi

\if0\cAnon
{
  \bigskip
  \bigskip
  \bigskip
  \begin{center}
    {\LARGE\bf Supplementary Material for ``Incomplete Matrix Regression''}
\end{center}
  \medskip
} \fi

\appendix
\setcounter{assumption}{3}
% This supplementary material contains the proposed method in more detail, in Appendix A, the proof of Theorem 1 in Appendix B, and extra theoretical properties are provided in Appendix C.
\section{Proofs of Propositions}
\subsection{Proof of Proposition 1}

In the main article, we provided derivations for the updates of the parameters $\bbeta$ and $\bA$. Here, we derive the remaining updates for parameters $\bgamma$, $\bB$, $\bphia$, and $\bphib$. We begin with $\bgamma$ and $\bB$, whose derivations proceed analogously to those of $\bbeta$ and $\bA$, respectively.

\textbf{Updating $\bgamma$.} 
The update for $\bgamma$ is obtained by solving
\begin{align*}
 \stept{\bgamma} &  = \arg\min_{\bgamma} \lossst{\bstheta^{(t-1)}_{[-\Gamma]}} \\
 & =  \arg\min_{\bgamma} \losssft{\bstheta^{(t-1)}_{[-\Gamma]}} + \lambg \normlass{\bgamma},
\end{align*}
where $\bstheta^{(t-1)}_{[-\Gamma]}=\bracks{\stept\bbeta,\bgamma, \bphia^{(t-1)},\bphib^{(t-1)}, \bA^{(t-1)},\bB^{(t-1)}}$.
The first term simplifies to
\begin{align*}
    & \losssft{\bstheta^{(t-1)}_{[-\Gamma]}}  \\
    = &\frac{1}{2} \normfrob{\bomega \circ (\bY - \bX\stept{\bbeta} - {\bgamma}{\bZ} - \steptmin{\bphia}\tr{\bone_m} - \bone_n \tr{\steptmin{\bphib}} - \steptmin\bM) +  \overline{\bomega}\circ\bracks{\steptmin\bgamma\bZ-\bgamma\bZ}} \\
  =& \frac{1}{2} \normfrob{\bomega \circ (\bE + \steptmin\bgamma\bZ - {\bgamma}\bZ ) +  \overline{\bomega}\circ\bracks{\steptmin\bgamma\bZ-\bgamma\bZ}}\\
  =& \frac{1}{2} \normfrob{\bE + \bracks{\steptmin\bgamma-\bgamma}\bZ  }.
\end{align*}

Differentiating the objective $\eqloss$ with respect to $\bgamma$ and setting the result to zero yields
\begin{align*}
    \deriv{\lossst{\bstheta^{(t-1)}_{[-\Gamma]}}}{\bgamma}
    & = \deriv{\losssft{\bstheta^{(t-1)}_{[-\Gamma]}}}{\bgamma} + \lambg \deriv{\normlass{\bgamma}}{\bgamma} \\
    & = -  \bE\tr\bZ + \bracks{\bgamma - \steptmin\bgamma} \bZ\tr\bZ + 
    \lambg  \operatorname{sign} \bracks{\bgamma} = \boldsymbol{0}.
\end{align*}

For this Lasso problem, a closed-form solution exists if $\bZ\tr\bZ=\iden{q}$, which can be achieved by orthonormalizing the rows of $\bZ$. Under this assumption, the update solution for $\bgamma$ is
\begin{align*}
    \stept\bgamma &=  \bE\tr\bZ + \steptmin\bgamma - 
    \lambg  \operatorname{sign}\bracks{\stept\bgamma} \\
    & = \softthresh{\lambg}{\bE\tr\bZ + \steptmin\bgamma}.
\end{align*}
Finally, the residual matrix is updated as
\[\bE=\bE + \bomega \circ \bracks{{\steptmin\bgamma}\bZ}-\bomega \circ \bracks{{\stept\bgamma}\bZ}.\]
\textbf{Updating $\bB$.}  Similarly, the update for $\bB$ is obtained by solving
\begin{align*}
 \stept{\bB} & = \arg\min_{\bB} \lossst{\bstheta^{(t-1)}_{[-\bB]}} \\
 & = \arg\min_{\bB} \losssft{\bstheta^{(t-1)}_{[-\bB]}} + \frac{\lambm}{2} \normfrob{\bB},
\end{align*}
where $\bstheta^{(t-1)}_{[-\bB]}=\bracks{\stept\bbeta,\stept\bgamma, \bphia^{(t)},\bphib^{(t)}, \bA^{(t)},\bB}$ and
\begin{align*}
     \losssft{\bstheta^{(t-1)}_{[-\bB]}}
  =& \frac{1}{2} \normfrob{\bE + \stept\bA\tr{\bracks{\steptmin\bB - \bB} }}.
\end{align*}
Differentiating the objective with respect to $\bB$  and setting the gradient to zero yields
\begin{align*}
    \deriv{\lossst{\bstheta^{(t-1)}_{[-\bB]}}}{\bB}
    & = \deriv{\losssft{\bstheta^{(t-1)}_{[-\bB]}}}{\bB} +  \frac{\lambm}{2} \deriv{\normfrob{\bB}}{\bB} \\
    & = - \tr{\stept\bA}\bE + \tr{\stept \bA}\stept \bA \bracks{\tr{\ \bB} -\tr{\steptmin \bB}  
       }+ \lambm \tr{ \bB} \\ &= \boldsymbol{0}.
\end{align*}
The solution is then given by
\begin{align}\label{eq:updateB0}
    \wtilde\bB &= \bracks{\tr{\bE}\stept\bA + \steptmin\bB \tr{\stept\bA}\stept\bA }\bracks{\tr{\stept\bA}\stept\bA + \lambm \iden{r}}^{-1} \\\label{eq:updateB}
    & = \bracks{\tr{\bE}\stept\bU + \stept\bV \stept{\bD^2} }\stept\bD\bracks{{\stept{\bD^2}} + \lambm \iden{r}}^{-1}.
\end{align}
As in Remark 2 of the main article, $\wtilde\bB$ does not have orthonormal columns. An additional step is therefore required to preserve orthonormality. As noted in the Remark, this step is only necessary when employing the more computationally efficient expression \eqref{eq:updateB} in lieu of \eqref{eq:updateB0}. Let $\wtilde\bB=\wtilde\bU\wtilde\bD\tr{\wtilde\bV}$ denote its singular value decomposition (SVD).  Setting $\stept\bB = \wtilde\bU\wtilde\bD$ and updating $\stept\bA=\stept\bA\wtilde\bV$ preserves orthogonality without altering the product $\stept\bA\tr{\stept\bB}$. 
Applying \eqref{eq:updateB} requires updating the matrices $\bU$, $\bD$, and $\bV$ after each update of $\bA$ or $\bB$. However, this approach also circumvents the need to explicitly evaluate $\bA$ and $\bB$. To this end, rather than evaluating \eqref{eq:updateB}, we compute
\[\wtilde{\bB\bD} = \bracks{\tr{\bE}\stept\bU + \stept\bV \stept{\bD^2} }\stept{\bD^2}\bracks{{\stept{\bD^2}} + \lambm \iden{r}}^{-1}.\]
Let $\stept{\boldsymbol{d}}$ denote the vector of the diagonal elements of $\stept{\bD^2}$. The rightmost term in the preceding formula, $\stept{\bD^2}\bracks{{\stept{\bD^2}} + \lambm \iden{r}}^{-1}$, is a diagonal matrix with diagonal entries given by $\stept{\boldsymbol{d}}/\bracks{\stept{\boldsymbol{d}}+\lambm}$. The matrices $\bU$, $\bD$, and $\bV$ are then updated as follows. Let $\wtilde\bU\wtilde\bD\tr{\wtilde\bV}$ denote the SVD of $\wtilde{\bB\bD}$, and let $\stept\bU$, $\stept{\bD^2}$, and $\stept\bV$ denote the parameter states after the $\stept\bA$ update step (note that these matrices are updated twice in each iteration: once after the update for $\bA$ and once after the update for $\bB$). The updates are then computed as
\[ 
\stept\bU = \stept{\bU} \wtilde\bV, \qquad
        \stept{\bD^2}=\wtilde\bD, \qquad
\stept\bV = \wtilde \bU.
\]
Finally, the residual matrix is updated as \[\bE = \bE +\bomega\circ\bracks{\steptmin\bU\steptmin{\bD^2}\tr{\steptmin\bV}}-\bomega\circ\bracks{\stept\bU\stept{\bD^2}\tr{\stept\bV}}.\]

\textbf{Updating $\bphib$ and $\bphia$.} As the intercepts are not penalized, their updates admit closed-form solutions obtained by averaging the residuals. While the low-rank latent matrix and the covariate coefficients are updated using the surrogate loss (6), we employ the original empirical loss (4) for the intercept vectors $\bphia$ and $\bphib$, which is restricted to the observed entries. We adopt this approach to avoid normalizing the residual sums by the full dimensions ($m$ for $\bphia$ and $n$ for $\bphib$). For highly sparse matrices, normalizing the residual sums by the full matrix dimensions, rather than by the number of observed entries per row or column, would cause the intercept estimates to shrink toward zero.
\begin{align*}
  \deriv{\losssf{\bstheta^{(t-1)}_{[-\bphia]}}}{\bphia} 
    & = \frac{1}{2}\deriv{\normfrob{\bE + \bomega\circ \bracks{\bracks{\steptmin\bphia - \bphia}\tr\bone_m }}}{\bphia} \\
    & = -  \bE\,\bone_m - \bracks{\bomega\,\bone_m}\circ \bracks{\steptmin\bphia - \bphia} \\ &= 0,
\end{align*}
and
\begin{align*}
  \deriv{\losssf{\bstheta^{(t-1)}_{[-\bphib]}}}{\tr\bphib}
    & = \frac{1}{2}\deriv{\normfrob{\bE + \bomega\circ\bracks{\bone_n\tr{\bracks{\steptmin\bphib - \bphib}} }}}{\tr\bphib}\\
    &= - \tr\bone_n\bE - \bracks{\tr\bone_n\bomega}\circ\bracks{\tr{\steptmin\bphib} - \tr\bphib}\\ &= \boldsymbol{0}.
\end{align*}
This leads to the following parameter updates
\begin{align*}
   \stept \bphia & = \bracks{\bE\,\bone_m}\oslash\bracks{\bomega\,\bone_m} +\steptmin\bphia,\\
       \stept \bphib & =  \tr\bE\bone_n \oslash \bracks{\tr\bomega\bone_n} +\steptmin\bphib.
\end{align*}
The corresponding updates for the residual matrix are given by
\begin{align*}
    \bE &= \bE + \bomega \circ \bracks{\bracks{\steptmin\bphia-\stept\bphia}\tr\bone_ m}, \\
     \bE &= \bE + \bomega \circ \bracks{\bone_n \tr{\bracks{\steptmin\bphib-\stept\bphib}}}.
\end{align*}

This concludes the proofs of equations $\equpdatesa$ of Proposition 1.
\qed
\subsection{Proof of Proposition 2}
In Proposition 2, we consider the general case in which we adjust for autocorrelation within the rows and columns of the matrix $\bY$.

\textbf{Updating $\bA$.} We begin by differentiating the objective function with respect to $\bA$:
\begin{align*}
\deriv{\lossst{\bstheta^{(t-1)}_{[-\bA]}}}{\bA}
    & = \deriv{\losssft{\bstheta^{(t-1)}_{[-\bA]}}}{\bA} +  \frac{\lambm}{2} \deriv{\,\trace{\bA^\top\bSa \bA}}{\bA} \\
    & = - \bE {{\steptmin \bB}} +
     \bracks{\bA - \steptmin \bA} \tr{\steptmin \bB} \steptmin \bB  + \lambm \bSa  \bA \\ &= \boldsymbol{0}.
\end{align*}
Rearranging the terms yields
\begin{align}
    \label{eq:prop2:A}
     \bE {\steptmin \bB} + \steptmin \bA \tr{\steptmin \bB} \steptmin \bB =
     \bA  \tr{\steptmin \bB} \steptmin \bB  + \lambm \bSa  \bA,
\end{align}
or equivalently
\begin{align*}
    \bE {\steptmin \bV} \steptmin\bD + \steptmin \bU \steptmin {\bD^3} =
     \bA  \steptmin{\bD^2} + \lambm \bSa  \bA.
\end{align*}
This is a Sylvester equation and has no direct solution for $\bA$ \citep{higham_Accuracy_2002}. However, we could exploit the diagonal structure of $\bD^2$ and the symmetry of  $\bSa$ to derive closed-form updates for each column of $\bA$. We define $\mathbf{W}_1 \in \R{n}{r}$ as the left-hand side of \eqref{eq:prop2:A}:
\[\mathbf{W}_1 = \bE\steptmin{\bB} + \steptmin{\bA}\tr{{\steptmin{\bB}}}\steptmin{\bB}.\]
The update for the $j$-th column of $\stept\bA$, denoted by $\indexxt{\bA}{.}{j}$, for $j=1,\dots,r$, is
\[\indexxt{{\bA}}{.}{j} = \left( \lambm \bSa + \steptmin{\boldsymbol{d}}_j\mathbf{I}_n\right)^{-1} \indexx{\mathbf{W}_1}{.}{j}.\]
Computing the inverse of the $n\times n$ matrix $\bracks{\lambm \bSa + \steptmin {\boldsymbol{d}}_j\mathbf{I}_n}$ at every iteration is inefficient. To avoid this, we use the eigendecomposition of the symmetric matrix $\bSa$, that is $\bSa=\bU_r\Lambda_r\tr\bU_r$. The inverse term can be then expressed as  
\begin{align*}
    \bracks{\lambm \bSa + \steptmin {\boldsymbol{d}}_j\mathbf{I}_n}^{-1} & =
     \bracks{\bU_r\bracks{\lambm\Lambda_r}\tr\bU_r + \steptmin {\boldsymbol{d}}_j\mathbf{I}_n}^{-1} \\
     & =
     \bracks{\bU_r\bracks{\lambm\Lambda_r}\tr\bU_r +  \steptmin {\boldsymbol{d}}_j\bU_r\tr\bU_r}^{-1} \\
& =
     \bU_r\bracks{{\lambm\Lambda_r+\steptmin {\boldsymbol{d}}_j}\mathbf{I}_n}^{-1}\tr\bU_r,
\end{align*}
where the inverse of the diagonal matrix $\bracks{{\lambm\Lambda_r+\steptmin {\boldsymbol{d}}_j}\mathbf{I}_n}$ is trivial.

\textbf{Updating $\bB$.} Differentiating the objective with respect to $\bB$ yields an analogous Sylvester equation:
\begin{align}
    \label{eq:prop2:B}
     \tr\bE \stept\bA + \steptmin \bB \tr{\stept \bA}\stept \bA = \bB \tr{\stept \bA}\stept \bA  + \lambm \bSb \bB.
\end{align}
Defining $\mathbf{W}_2 \in \R{m}{r}$ as the left-hand side of \eqref{eq:prop2:B},
\[\mathbf{W}_2 = \tr\bE{\stept\bA} + \steptmin \bB \tr{\stept \bA}\stept \bA , \]
the update for the $j$-th column of $\stept\bB$  ($j=1,\dots,r$) is
\[\indexxt{\bB}{.}{j} = \bracks{\lambm\bSb+ \stept {\boldsymbol{d}}_j\mathbf{I}_m}^{-1} \indexx{{\mathbf{W}}_2}{.}{j}.\]
Proceeding as before, let $\bSb=\bU_c\Lambda_c\tr\bU_c$ be the eigendecomposition of $\bSb$. Then,
\begin{align*}
    \bracks{\lambm \bSb +\stept {\boldsymbol{d}}_j\mathbf{I}_m}^{-1} & =
     \bU_c\bracks{{\lambm\Lambda_c+\stept {\boldsymbol{d}}_j}\mathbf{I}_m}^{-1}\tr\bU_c.
\end{align*}

\qed

\section{Proof of Theorem 1}\label{section:B}

We establish the results under the special case of no prior knowledge about the structural dependencies (i.e., $\bSa = \iden n,\bSb=\iden m$) and defer the discussion of the general case to Appendix D. This assumption applies implicitly to all lemmas presented in Appendix B.  Consider the following objective function 
\begin{align} 
    \loss{\bstheta} \, & = \, \losssf{\bstheta} +
    \lambm \normnuc{\bM}  +
    \lambb \normlass{\bbeta} + 
    \lambg \normlass{\bgamma}\label{eq:obj2},\end{align}
    where
    \begin{align}
    \notag
     \losssf{\bstheta} \, &= \,\frac{1}{2}\normfrob{\bomega \circ (\bY-\btheta)},
\end{align}
and the optimization problem
\begin{align} \label{problem:thm1}
    {\estim\bstheta} \in \underset{\substack{\norminf{\bbeta}\leq c_\beta,%\,\norminf{\bX\bbeta}\leq \cxb  
    \\ \norminf{\bgamma}\leq c_\gamma,%\,\norminf{\bgamma\bZ}\leq \cgz
    \\
    \norminf{\bM} \leq c_m}}{\operatorname{argmin}} \loss{\bstheta}.
\end{align}

 The objective is to establish error bounds between any global optimum $\estim\bstheta=\bracks{\estim\bbeta,\estim\bgamma,\estim\bM}$ of \eqref{problem:thm1} and the true parameter matrices $\truev\bstheta=\bracks{\truev\bbeta, \truev\bgamma, \truev\bM}$ for a suitable choice of regularization parameters $\bracks{\lambb,\lambg, \lambm}$, provided that the estimation errors lie in restricted sets. The choice of these regularization parameters is determined by the gradient of the empirical loss (the score function) evaluated at the true parameters. 
 By the first-order condition of the population loss, the expected value of the gradient at the true parameters is zero. Hence, under ideal circumstances, we expect the empirical gradient, evaluated at the true parameters, to be small, with fluctuations bounded in terms of the dual norms of the respective regularizers. To this end, we define the following good events:
\begin{align*}
    \cG_1(\lambb) &= \bracksb{2\norminf{\nabla\losssf{\bstheta_{\bbeta}}} \le \lambb}, \\
    \cG_2(\lambg) &= \bracksb{2\norminf{\nabla\losssf{\bstheta_{\bgamma}}} \le \lambg}, \\
    \cG_3(\lambm) &= \bracksb{2\normop{\nabla\losssf{\bstheta_{\bM}}} \le \lambm},
\end{align*}
where $\bstheta_{\bbeta} = \bracks{\truev\bbeta, \bgamma, \bM}$,  $\bstheta_{\bgamma} = \bracks{\bbeta, \truev\bgamma, \bM}$, $\bstheta_{\bM} = \bracks{\bbeta, \bgamma, \truev\bM}$, with $\bbeta$, $\bgamma$, and $\bM$ denoting arbitrary feasible points of problem \eqref{problem:thm1}.  

The second component of the error analysis is the curvature of the loss function around the true parameters. This curvature is captured by the second derivative evaluated at the true parameters (Hessian matrix). It controls the link between the cost difference $\loss{\truev\bstheta}-\loss{\estim\bstheta}$ and the estimation error $\estim\bstheta-\truev\bstheta$. High curvature implies that when the cost difference is small, the distance between $\estim\bstheta$ and $\truev\bstheta$  is also small. However, with a partially observed matrix $\bY$, the analysis is complicated by the fact that $|\bomega| \ll nm$. Consequently, the empirical Hessian matrix is rank-degenerate as there are $nm-|\bomega|$ directions in which the loss function is flat. However, we are not interested in all directions, but rather only the directions in which the error matrices $\bracks{\estim\bbeta-\truev\bbeta,\estim\bgamma-\truev\bgamma,\estim\bM-\truev\bM}$ can lie. This can be established by lower-bounding the remainder of the first-order Taylor series expansion, a condition known as restricted strong convexity (RSC). The first-order Taylor series expansion of the loss function $\losssf{\cdot}$ is given by
\[\losssf{\truev\bstheta+\Delta} = \losssf{\truev\bstheta} + \langle {\nabla \losssf{\truev\bstheta}},{\Delta}\rangle + \tayerr{\Delta}, \]
where $\tayerr{\Delta}$ denotes the remainder term. Since the loss function is convex, this term is guaranteed to be nonnegative. We now define the RSC condition as follows: for a given regularizer $\Phi(\cdot)$ and an error matrix belonging to a restricted set, the loss function satisfies an RSC condition with a radius $R>0$, curvature $\kappa>0$, and tolerance $\tau^2$ if
\begin{align} \label{eq:RSC}
     \tayerr{\Delta} \ge \frac{\kappa}{2} \normfrob{\Delta} - \tau^2 \Phi^2(\Delta)\qquad \text{for all } \Delta \in \mathbb{B}{(R)},
\end{align}
where $\Delta \in \bracksb{\estim\bbeta-\truev\bbeta,\estim\bgamma-\truev\bgamma,\estim\bM-\truev\bM}$, and $\mathbb{B}(R)$ is a ball of radius $R$ defined within the restricted set. Note that, setting $\tau^2=0$ is equivalent to asserting that the loss function is locally strongly convex, but this strong convexity does not hold in matrix completion due to the aforementioned rank deficiency.

The proof proceeds as follows. First, we show that the decomposability of the Lasso and nuclear regularizers, combined with regularization parameters satisfying the good events defined above, ensures that the estimation errors lie in restricted sets. Let $r = \operatorname{rank}\bracks{\truev\bM}$. For nonnegative  constants $d_1, \dots, d_5$ (to be specified later), we define the following sets:
\begin{align*}
\mathcal{A}(d_1,d_2)
  =& \Bigl\{
      \boldsymbol{\beta}\in\mathbb{R}^{p\times m}\,\big|\,
      \norminf{\bbeta} \leq 2 \,c_\beta, \;
      \norminf{\bX\bbeta} \leq 2 \,\cxb, \;
      \|\boldsymbol{\beta}\|_1 \leq d_1,\;
      \mathbb{E}\|\boldsymbol{\bomega}\circ\mathbf{X}\boldsymbol{\beta}\|_F^2 \ge d_2
    \Bigr\},\\
\mathcal{D}(d_3,d_4)
  =& \Bigl\{
      \bgamma\in\mathbb{R}^{n\times q}\,\big|\,
       \norminf{\bgamma} \leq 2 \,c_\gamma, \;
       \norminf{\bgamma\bZ} \leq 2 \,\cgz, \;
      \normlass{\bgamma} \leq d_3,\;
      \mathbb{E}\|\boldsymbol{\bomega}\circ\bgamma\bZ\|_F^2 \ge d_4
    \Bigr\}, \\
\mathcal{S}(d_5) =& 
\Bigl \{  
\bM \in \mathbb{R}^{n\times m} \, \big | \,
\norminf{\bM} \le 2\,c_m, \;
\E\normfrob{\bomega\circ\bM} \ge d_5, \;
\normnuc{\bM} \leq \sqrt{32\,r} \normf{\bM}  \;
\Bigr\}.
\end{align*}

Second, we establish that the RSC condition holds within these sets with high probability. In particular, for lower-order residual terms $D_\beta$, $D_\Gamma$, and $D_M$ (defined subsequently), the following inequalities hold with high probability
\begin{align}\notag
\tayerr{\Delta\bbeta} &=\| \bomega \circ \bX \Delta \bbeta \|^2_F \geq \frac{\pi_{\min}}{2} \| \bX \Delta \bbeta \|^2_F - D_\beta,\\\notag
\tayerr{\Delta\bgamma} &=\normfrob{\bomega \circ \Delta \bgamma\bZ} \geq \frac{\pi_{\min} }{2} \normfrob{ \Delta \bgamma\bZ} - D_\Gamma,\\
\tayerr{\Delta\bM} &= \| \bomega \circ \Delta \bM \|^2_F \geq \frac{\pi_{\min}}{2} \|  \Delta \bM \|^2_F - D_M. \label{ineq:rsc_M}
\end{align}
Finally, we derive the error bounds using the RSC property alongside the optimality condition: by the definition of $\estim\bstheta$, for any choice of $\bstheta$,
\begin{align}
    \label{ineq:Lossdef}
\loss{\estim\bstheta}\leq \loss{\bstheta}.
\end{align}
Below, we state the Lemmas, leaving their proofs to Appendix C.

\begin{lemma} \label{lemma:norm_bound1}
    
Assume that the regularization parameters satisfy the conditions

\begin{align*}
\lambda_\beta &\geq 2 \, \left(\norminf{\tr\bX\bracks{\bomega\circ\Delta\bM}} +\norminf{\tr\bX\bracks{\bomega\circ\Delta\bgamma\bZ}} + \errinfx \right),\\
\lambg &\geq 2 \, \left(\,\norminf{\bracks{\bomega\circ\Delta\bM}\tr\bZ} +\,\norminf{\bracks{\bomega\circ\bX\Delta\bbeta}\tr\bZ} + \errinfz \right).
\end{align*}
Then, the estimation errors satisfy the bounds
\begin{align*}
\normlass{\Delta \boldsymbol{\beta}} &\leq 4 \normlass{\truev\bbeta}, \\
\normlass{\Delta \bgamma} &\leq 4 \normlass{\truev\bgamma}.
\end{align*}
\end{lemma}

\begin{lemma}\label{lemma:norm_bound2}
    Assume that $\lambm$ satisfies the condition
    \[\lambm  \ge 2 \bracks{\normop{\bomega\circ\bX\Delta\bbeta}+\normop{\bomega\circ\Delta\bgamma\bZ}+\normop{\bomega\circ\bepsilon}},\]
    and let $r=\operatorname{rank}\bracks{\truev\bM}$. Then,
    \[ \normnuc{\Delta\bM} \le\sqrt{32r}\normf{\Delta\bM}. \]
\end{lemma}

Lemma \ref{lemma:norm_bound1} establishes that when $\lambb$ and $\lambg$ satisfy the good events,  $\Delta \bbeta \in \cA(4\|\truev\bbeta\|_1,d_2)$ and $\Delta \bgamma \in \cD(4\|\truev\bgamma\|_1,d_4)$ for some $d_2$ and $d_4$. Furthermore, Lemma \ref{lemma:norm_bound2} demonstrates that when $\lambm$ satisfies the event $\cG_3$,  $\Delta\bM \in \cS(d_5)$ for some $d_5$. Both Lemmas are standard results in the literature on Lasso regression and matrix completion. Below, we give values for $d_2$, $d_4$, and $d_5$ as follows 
\[
    d_2 = \frac{128\,\cxb^2\,  \log(n+m)}{\pi_{\text{min}}\log(6/5)},\quad
    d_4 = \frac{128\,\cgz^2\,  \log(n+m)}{\pi_{\text{min}}\log(6/5)},\quad
    d_5 = \frac{128\,c_m^2\,\log(n+m)}{\pi_{\text{min}}\,\log(6/5)}\:.
\]
These values help to establish the RSC condition with high probability. To ease the notation throughout the remainder of this appendix, we use the relation $x\lesssim y$ to indicate that $x \leq C y$ for some universal constant $C>0$ that is independent of the  dimensions $n$, $m$, $p$, $q$, and $r$, as well as the observation probabilities $\pi_{ij}$. Additionally, we define the quantity \(        N_1 = c^* \bracks{\sqrt{\pi_{\max}\bracks{n\vee m}} + {\sqrt{\log{(n\wedge m)}}}}\), where $c^*$ is a universal constant.

\begin{lemma}\label{lemma:rsc1}
Assume that $\lambb$ and $\lambg$ satisfy
\begin{align*}
\lambda_\beta &\geq 4\, d_x\, \bracks{c_m + \cgz} +2\, \sigma\, \sqrt{2\,n\log\bracks{pm(m+n)/4}}\,, \\
\lambg &\geq4\, d_z\, \bracks{c_m + \cxb} + 2\,\sigma\, \sqrt{2\,m\log\bracks{qn(m+n)/4}}\,,
\end{align*}
where $d_x = \max_{1\leq k \leq p} \,\sum_{i=1}^n |\bX_{ik}|$ and $d_z = \max_{1\leq k \leq q} \,\sum_{j=1}^m |\bZ_{kj}|$.
Then,  for any $\Delta \bbeta \in \cA(4\|\truev\bbeta\|_1,d_2)$, $\Delta \bgamma \in \cD(4\|\truev\bgamma\|_1,d_4)$, and with $d_2$ and $d_4$ defined as above, each of the following inequalities hold with probability at least $1-8(n+m)^{-1}$, 
\begin{align}
\label{ineq: xdbeta}
\normfrob{{\bomega} \circ \bX \Delta{\bbeta}}&\ge\frac{1}{2}\,\mathbb{E}\normfrob{\boldsymbol{\bomega} \circ \mathbf{X} \Delta \boldsymbol{\beta}}
      - D_\beta, \\
      \label{ineq: gammaz}
\normfrob{\boldsymbol{\bomega} \circ \Delta\bgamma\bZ} &\ge \frac{1}{2}\,\mathbb{E}\normfrob{\boldsymbol{\bomega} \circ \Delta\bgamma\bZ}
      - D_\Gamma,
\end{align}
where
\begin{align*}
    D_\beta
        = 768\,\cxb^2\,\pim^{-1} \bracks{ p\bracks{\E\normop{\Sigma_R}}^2+4},\\
      D_\Gamma
        = 768\,\cgz^2\,\pim^{-1} \bracks{ q\bracks{\E\normop{\Sigma_R}}^2+4}.
\end{align*}
 
\end{lemma}

Inequalities \eqref{ineq: xdbeta} and \eqref{ineq: gammaz} establish the RSC condition \eqref{eq:RSC}. In particular, since $\E\normfrob{\bomega\circ\bX\Delta\bbeta}\geq \pi_{\min} \normfrob{\bX\Delta\bbeta}\geq \pi_{\min}\, n\,\kappa_x\normfrob{\Delta\bbeta}$, substitution into \eqref{ineq: xdbeta} yields $\normfrob{\bomega\circ\bX\Delta\bbeta}\ge \pi_{\min}\, n\,\kappa_x\normfrob{\Delta\bbeta}/2 - D_\beta$. An analogous bound holds for $\Delta\bgamma$. In the following lemma, we establish a corresponding result for $\Delta\bM$.

\begin{lemma}\label{lemma:rsc2}
Assume that $\lambm$ satisfies 
\[\lambm \geq 4 \left(\sigma_x\,c_\beta  +   \sigma_z\,c_\gamma +  \sigma \,\sqrt{2\,\pi_{\max}\bracks{n\vee m}\log\bracks{\bracks{m+n}/2}}\right),\]
where $\sigma_x = \normop{\bX}$ and $\sigma_z = \normop{\bZ}$ then for any $\Delta \bM \in \cS(d_5)$, where $d_5$ is defined as above, the following inequality holds with a probability at least $1-8(n+m)^{-1}$
\begin{align}
\label{ineq: M}
\|\boldsymbol{\bomega} \circ \Delta\bM\|_{F}^{2} \ge \frac{1}{2}\,\mathbb{E}\bigl\|\boldsymbol{\bomega} \circ \Delta\bM\bigr\|_{F}^{2}
      - D_M,
\end{align}
where
\[
  D_M
        = 3072\,c_m^2\, \pi_{\min}^{-1} \bracks{8\,r\,\bracks{\E{\normop{\boldsymbol{\Sigma}_R}}}^2+1}.
\]

\end{lemma}

 Prior to deriving the estimation error bounds, we present the following lemma, which bounds the infinity and operator norms of the model residuals.
\begin{lemma} \label{lemma:noise}
For any $0 < \delta < 1 $, with probability at least $1 - \delta$, the following inequalities hold:
\begin{align*}
    %&\bracks{i}\qquad\errinf \leq \sigma\sqrt{4  \log{{\bracks{m+n}}}-2\log\bracks{{\delta}} },\\
    &\bracks{i}\qquad\normop{\bomega\circ\bepsilon} \leq 2\sigma \sqrt{\pi_{\max}\bracks{n\vee m} \log\left( \frac{2 \bracks{m+n}}{\delta} \right)}, \\
    &\bracks{ii}\qquad \norminf{\tr\bX\bracks{\bomega\circ\bepsilon}} \leq
     \sigma \sqrt{2\,n\,{\log\bracks{\frac{2\,p\,m}{\delta}}}}, \\
    &\bracks{iii}\qquad \norminf{\bracks{\bomega\circ\bepsilon}\tr\bZ} \leq
    \sigma \sqrt{2\,m\,{\log\bracks{\frac{2\,q\,n}{\delta}}}}.
\end{align*}

\end{lemma}

Next, we derive upper bounds for the estimation errors, beginning with $\Delta\bbeta$.

\subsection{Upper Bound for \texorpdfstring{$\Delta\bbeta$}{Delta Beta}}

\textbf{Case 1:} Suppose that $\E\|\bomega\circ\bX \Delta \bbeta\|^2_F \leq {128\,\cxb^2\,\log(n+m)\pim^{-1}}/{\log(6/5)}$. Combining this inequality with the bound  $\|\bX \Delta \bbeta\|^2_F \leq \pi_{\min}^{-1} \, \E\|\bomega\circ\bX \Delta \bbeta\|^2_F$ yields
\begin{align*}
\frac{1}{mn}\|\bX \Delta \bbeta\|^2_F &\leq
\frac{128\,\cxb^2\,\log(n+m)}{mn\log(6/5)\,\pi_{\text{min}}^2}\\
&\lesssim \frac{\log(n+m)}{mn\,\pim^2}.
\end{align*}

\textbf{Case 2:} Suppose that $\E\|\bomega\circ\bX \Delta \bbeta\|^2_F \geq {128\,\cxb^2\,\log(n+m)\pim^{-1}}/{\log(6/5)}$. Then, Lemma \ref{lemma:norm_bound1} implies  $\Delta \bbeta \in \cA\bracks{4\|\truev\bbeta\|_1,128\,\cxb^2\,\log(n+m)\pim^{-1}/{\log(6/5)}}$, and by Lemma \ref{lemma:rsc1}, inequality \eqref{ineq: xdbeta} holds.

Applying the optimality condition \eqref{ineq:Lossdef} with $\bstheta=\bracks{\truev\bbeta,\estim\bgamma,\estim\bM}$ yields
\begin{align*} 
    \frac{1}{2} \| \boldsymbol{\bomega} \circ (\mathbf{Y} - \estim\btheta) \|_F^2
    \leq 
    \frac{1}{2} \| \boldsymbol{\bomega} \circ (\mathbf{Y} - \wtilde\btheta) \|_F^2 \,
    + \lambb \bracks{\|\truev\bbeta\|_1 - \|\estim\bbeta\|_1}, 
\end{align*}
where $\wtilde \btheta = \bX \truev\bbeta +\estim\bgamma\bZ + \estim\bM$.
We now bound 
\begin{align*}
    \frac{1}{2} \normfrob{\bomega\circ\bracks{\bX\Delta\bbeta}} & =
    \frac{1}{2} \| \bomega \circ (\wtilde\btheta-\estim\btheta) \|^2_F\\ &= \frac{1}{2} \| \bomega \circ ((\bY-\estim\btheta) - (\bY - \wtilde \btheta)) \|^2_F \\
    &= \frac{1}{2} \bracks{\| \bomega \circ (\bY-\estim\btheta) \|^2_F +  \| \bomega \circ  (\bY - \wtilde \btheta) \|^2_F} \\ &\qquad -  \sum_{i,j} \bomega_{ij} (\bY_{ij}-\estim\btheta_{ij}) (\bY_{ij} - \wtilde \btheta_{ij}) \\
    &\leq 
     \| \boldsymbol{\bomega} \circ (\mathbf{Y} - \wtilde\btheta) \|_F^2 \,
    -  \sum_{i,j} \bomega_{ij} (\bY_{ij}-\estim\btheta_{ij}) (\bY_{ij} - \wtilde \btheta_{ij})
    + \lambda_\beta \|{\Delta{\bbeta}}\|_1  \\
    & =
    \bracks{ \sum_{i,j} \bomega_{ij} (\bY_{ij} - \wtilde \btheta_{ij})^2
    -  \sum_{i,j} \bomega_{ij} (\bY_{ij}-\estim\btheta_{ij}) (\bY_{ij} - \wtilde \btheta_{ij}) }
    + \lambda_\beta \|{\Delta{\bbeta}}\|_1 \\
    & =
     \bracks{   \sum_{i,j} \bomega_{ij}   (\bY_{ij} - \wtilde \btheta_{ij})
    \bracks{\bY_{ij}-\wtilde\btheta_{ij} - \bY_{ij}+\estim\btheta_{ij}}
    }
   + \lambda_\beta \|{\Delta{\bbeta}}\|_1 \\
    & =
     \bracks{  \sum_{i,j} \bomega_{ij}
    (\bepsilon_{ij}- \Delta\bgamma_i \bZ_j-\Delta \bM_{ij})
    (\bX_{i} \Delta \bbeta_{j})
    }
   + \lambda_\beta \|{\Delta{\bbeta}}\|_1 \\
    & \leq
    \|\Delta \bbeta\|_1 \bracks{  {\errinfx
    + \norminf{\tr\bX\bracks{\bomega\circ\Delta\bM}} + \norminf{\tr\bX\bracks{\bomega\circ\Delta\bgamma\bZ}}}
    +  \lambda_\beta  }.
\end{align*}

Assuming
\(\lambda_\beta \geq 2\,\bracks{\norminf{\tr\bX\bracks{\bomega\circ\Delta\bM}} + \norminf{\tr\bX\bracks{\bomega\circ\Delta\bgamma\bZ}} +\errinfx}\), then by Lemma \ref{lemma:norm_bound1}, $\|\Delta \bbeta \|_1 \leq 4 \|{\truev{\bbeta}}\|_1$, and  it follows that
\[
     \| \bomega \circ \bX\Delta \bbeta \|^2_F 
    \leq 3 \lambb \normlass{\Delta\bbeta}   
      \leq 12 \lambb \normlass{\truev\bbeta}.
\]

Applying inequality \eqref{ineq: xdbeta}, we obtain
\begin{align*}
    \normfrob{\bX\Delta\bbeta} & \leq \frac{2}{\pim}\bracks{\normfrob{\bomega\circ\bX\Delta\bbeta}+D_\beta} \\
    &\leq\frac{2}{\pim}\bracks{12\, \lambb \normlass{\truev\bbeta} + 768\,\cxb^2\,\pim^{-1} \bracks{ p\,N_1^2+4} }\\
    &\leq\frac{8}{\pim}\bracks{3\, \lambb \,s_\beta\,c_\beta + 192\,\cxb^2\,\pim^{-1} \bracks{ p\,N_1^2+4} },
\end{align*}
Substituting $\lambb$ with  \(
    \lambda_\beta \asymp \,{d_x\bracks{c_m+\cgz}+\sigma\sqrt{n\log\bracks{nm}}} 
\), we obtain
\begin{align*}
     \normfrob{\bX\Delta\bbeta} 
    &\leq\frac{8}{\pim}\bracks{3\, \bracks{d_x\bracks{c_m+\cgz}+ \sigma\sqrt{n\log\bracks{nm}}}  \,s_\beta\,c_\beta + 192\, \cxb^2\,\pim^{-1} \bracks{p\, N_1^2+4} }\\
    &\lesssim \frac{1}{\pim}\bracks{\, \bracks{d_x\bracks{c_m+\cgz}+ \sigma\sqrt{n\log\bracks{nm}}}  \,s_\beta + p\,\pim^{-1} \bracks{\pi_{\max}(n\vee m) + \log\bracks{n\wedge m}}} .
\end{align*}
Using Assumption 1 we have $\normfrob{\bX\Delta\bbeta}\ge n \,\kappa_x\, \normfrob{\Delta\bbeta}$ and $d_x \leq n \norminf{\bX} \leq n$. Then,
\begin{align*}
     \normfrob{\Delta\bbeta} 
    &\lesssim  \frac{s_\beta}{\pim\, n}{{\left( n\bracks{c_m+\cgz} +\sigma\sqrt{n\log\bracks{nm}}\right)} +  \frac{p}{n\,\pim^2} \bracks{\pi_{\max}\bracks{n\vee m}+\log\bracks{n\wedge m}} }\\
    &\lesssim \frac{s_\beta}{\pim\,n}{\left(n\bracks{c_m+\cgz}+\sigma\sqrt{n\log\bracks{nm}}\right)} +  \frac{p\,\pi_{\max}\bracks{n\vee m}}{n\,\pim^2}  ,
\end{align*}
and subsequently,
\begin{align*}
     \frac{1}{pm}\normfrob{\Delta\bbeta} 
    &\lesssim  \frac{s_\beta}{\pim\,p\,m}{\bracks{\bracks{c_m+\cgz}+ {\sigma}\sqrt{\frac{{\log\bracks{nm}}}{{n}}}} +  \frac{\pi_{\max}}{\pim^2(n\wedge m)}},
\end{align*}
which yields part (i) of Theorem 1. 
\subsection{Upper Bound for \texorpdfstring{$\Delta\bgamma$}{Delta Gamma}}
The error bound for $\estim\bgamma$ follows by analogous arguments. Specifically, we obtain
\begin{align*}
     \normfrob{\Delta\bgamma\bZ} 
    &\leq\frac{8}{\pim}\bracks{3\, \bracks{d_z\bracks{c_m+\cxb}+\sigma\sqrt{m\log\bracks{nm}}}  \,s_\Gamma\,c_\gamma + 192\, \cgz^2\,\pim^{-1} \bracks{q\,N_1^2+4} },
\end{align*}
and substituting $\lambg$ with  \(
    \lambg \asymp \,{d_z\bracks{c_m+\cxb}+\sigma\sqrt{m\log\bracks{nm}}}
\) and using Assumption 1,  we obtain
\begin{align*}
     \frac{1}{qn}\normfrob{\Delta\bgamma} 
    &\lesssim  \frac{s_\Gamma}{\pim\,q\,n}\bracks{\bracks{c_m+\cxb} +{{\sigma\sqrt{\frac{{\log\bracks{nm}}}{{m}}}}}} +  \frac{\pi_{\max}}{\pim^2(n\wedge m)}.
\end{align*}

\subsection{Upper Bound for \texorpdfstring{$\Delta\bM$}{Delta M}}
\label{subsection:bound_M}

Next, we derive an upper bound for the estimation error of $\bM$.

\textbf{Case 1:} Suppose that $\E\|\bomega\circ\Delta\bM\|^2_F \leq {128\,c^2_m\,\log(n+m)\pim^{-1}}/{\log(6/5)}$. It then follows that
\begin{align*}
\frac{1}{mn}\|\Delta\bM\|^2_F &\leq
\frac{128\,c^2_m\,\log(n+m)}{\log(6/5)\,\pi_{\text{min}}^2mn},\\
&\lesssim \frac{\log(n+m)}{\pi_{\text{min}}^2\,mn} .
\end{align*}

\textbf{Case 2:} Suppose that $\E\|\bomega\circ \Delta \bM\|^2_F \geq {128\,c_m^2\,\log(n+m)\pim^{-1}/{\log(6/5)}}$. Then, Lemma \ref{lemma:norm_bound2} implies  \(\Delta \bM \in \cS\bracks{{128\,c^2_m\,\log(n+m)\pim^{-1}}/{\log(6/5)}},\) and by Lemma \ref{lemma:rsc2}, inequality \eqref{ineq: M} holds.

Applying \eqref{ineq:Lossdef} with $\bstheta=\bracks{\estim\bbeta,\estim\bgamma,\truev\bM}$ yields
\begin{align*} 
    \frac{1}{2} \| \boldsymbol{\bomega} \circ (\mathbf{Y} - \estim\btheta) \|_F^2
    \leq 
    \frac{1}{2} \| \boldsymbol{\bomega} \circ (\mathbf{Y} - \wtilde\btheta) \|_F^2 \,
    + \lambm \bracks{\|\truev\bM\|_* - \|\estim\bM\|_*}, 
\end{align*}
where $\wtilde \btheta = \bX \estim\bbeta +\estim\bgamma\bZ + \truev\bM$.
We proceed to bound the term
\begin{align*}
    \frac{1}{2} \normfrob{\bomega\circ\bracks{\Delta\bM}} & =
    \frac{1}{2} \| \bomega \circ (\wtilde\btheta-\estim\btheta) \|^2_F\\ &= \frac{1}{2} \| \bomega \circ ((\bY-\estim\btheta) - (\bY - \wtilde \btheta)) \|^2_F \\
    &= \frac{1}{2} \bracks{\| \bomega \circ (\bY-\estim\btheta) \|^2_F +  \| \bomega \circ  (\bY - \wtilde \btheta) \|^2_F} \\ &\qquad -  \sum_{i,j} \bomega_{ij} (\bY_{ij}-\estim\btheta_{ij}) (\bY_{ij} - \wtilde \btheta_{ij}) \\
    &\leq 
     \| \boldsymbol{\bomega} \circ (\mathbf{Y} - \wtilde\btheta) \|_F^2 \,
    -  \sum_{i,j} \bomega_{ij} (\bY_{ij}-\estim\btheta_{ij}) (\bY_{ij} - \wtilde \btheta_{ij})
    + \lambm \|{\Delta{\bM}}\|_*  \\
    & =
    \bracks{ \sum_{i,j} \bomega_{ij} (\bY_{ij} - \wtilde \btheta_{ij})^2
    -  \sum_{i,j} \bomega_{ij} (\bY_{ij}-\estim\btheta_{ij}) (\bY_{ij} - \wtilde \btheta_{ij}) }
    + \lambm \|{\Delta{\bM}}\|_* \\
    & =
     \bracks{   \sum_{i,j} \bomega_{ij}   (\bY_{ij} - \wtilde \btheta_{ij})
    \bracks{\bY_{ij}-\wtilde\btheta_{ij} - \bY_{ij}+\estim\btheta_{ij}}
    }
   + \lambm \|{\Delta{\bM}}\|_* \\
    & =
     \bracks{  \sum_{i,j} \bomega_{ij}
    (\bepsilon_{ij}- \bX_i\Delta\bbeta_j - \Delta\bgamma_i \bZ_j)
    ( \Delta \bM_{ij})
    }
   + \lambm \|{\Delta{\bM}}\|_* \\
    & \leq
     \bracks{ \normop{\bomega\circ\bepsilon} + \normop{\bomega\circ\Delta\bgamma\bZ} + \normop{\bomega\circ\bX\Delta\bbeta}
    }  \normnuc{\Delta\bM} 
   %\\&\qquad 
   + \lambm \|\Delta\bM\|_* .
\end{align*}

By Lemma \ref{lemma:norm_bound2}, and assuming
\(\lambm \geq 2\bracks{\normop{\bomega\circ\bX\Delta\bbeta} + \normop{\bomega\circ\Delta\bgamma\bZ} +\normop{\bomega\circ\bepsilon}}\), it follows that
\begin{align}
     \| \bomega \circ \Delta \bM \|^2_F 
    &\leq 3 \lambm \normnuc{\Delta\bM}   
     \leq 3\sqrt{32 r}\lambm \normf{\Delta\bM}.
     \label{inequality:bound:M:lambda}
\end{align}

Applying inequality \eqref{ineq: M} we obtain
\begin{align*}
     \normfrob{\Delta\bM} &\leq \pi_{\min}^{-1} \E\normfrob{\bomega\circ\Delta\bM} 
     \\&\leq 2\,{\pi}^{-1}_{\min} \bracks{\normfrob{\bomega\circ\bracks{\Delta\bM}} +D_M}\\
     &\leq 6\sqrt{32 r}\, \pim^{-1} \lambm \normf{\Delta\bM} + 2\,\pim^{-1}\,D_M\\
     &\leq \frac{1}{2}\bracks{6\sqrt{32 r}\, \pim^{-1} \lambm}^2 + \frac{1}{2} \normfrob{\Delta\bM} + 2\,\pim^{-1}\,D_M,
    %&\leq 64\,{\pi}^{-1}_{\min}\bracks{mn}^{-1} \bracks{9\,r\,\lambm^2 + 4\, \pi_{\min}^{-1} \bracks{22\,r\,c^*\,N_1^2+12\,c^2_m}}.
\end{align*}

where the last step follows by applying Young's inequality. Subsequently,

\begin{align*}
     \normfrob{\Delta\bM} &\leq  \bracks{6\sqrt{32 r}\, \pim^{-1} \lambm}^2  + 4\,\pim^{-1}\,D_M \\
     & = 4\, \pim^{-1} \bracks{288 \, r\, \pim^{-1} \lambm^2 + D_M} \\
     & = 4\, \pim^{-1} \bracks{288 \, r\, \pim^{-1} \lambm^2 + 3072\,c_m^2\,\pim^{-1}\,\bracks{8\,r  \,N_1^2 + 1}} \\
     & = 128\, \pim^{-2} \bracks{9 \, r\, \lambm^2 + 96\,c^2_m\,\bracks{8\,r  \,N_1^2 + 1}}.
\end{align*}
Substituting $\lambm^2$ with the bound from Lemma \ref{lemma:rsc2}, we obtain
\begin{align}\notag
     \normfrob{\Delta\bM} &\leq \frac{128}{\,\pim^{2}} \bracks{9 \, r\, \lambm^2 + 96\,c^2_m\,\bracks{8\,r  \,N_1^2 + 1}}\\
     \label{inequality:bound:M:estim}
     &\lesssim \frac{r}{\,\pim^{2}} \bracks{  \lambm^2 + N_1^2}\\\notag
     &\lesssim \frac{r}{\,\pim^{2}}  \bracks{\sigma_x^2 + \sigma_z^2 +  \pi_{\max}\bracks{n\vee m}\bracks{\sigma^2\,\log\bracks{n+m}+1} + \log \bracks{n \wedge m}}\\ \notag
     &\lesssim \frac{r}{\,\pim^{2}}  \bracks{\sigma_x^2 + \sigma_z^2 +\pi_{\max}\, \sigma^2 \bracks{n\vee m}\log\bracks{n+m}}.
\end{align}
Finally, we have
\begin{align}
    \label{inequality:bound:M:estim:last}
    \frac{1}{nm} \normfrob{\Delta\bM} \lesssim \frac{r\bracks{\sigma^2_x+\sigma^2_z}}{\,\pim^{2}\,n\,m} +  \frac{r\,\pi_{\max}\,\sigma^2 \log\bracks{n+m}}{\,\pim^{2}\bracks{n\wedge m}}.
\end{align}
This concludes the proof of Theorem 1.
\qed

\section{Proofs of Lemmas in Appendix B \label{sec:c}}

\textbf{Optimality Conditions.} Consider the optimization problem given in \eqref{problem:thm1} and let $\estim\bstheta=(\estim\bbeta, \estim\bgamma, \estim\bM)$ be a minimizer of it. Then, there exist subgradients $\gradb \in \partial\normlass{\estim\bbeta}$, $\gradg \in \partial\normlass{\estim\bgamma}$, and $\gradm \in \partial \normnuc{\estim\bM}$ such that for any feasible tuple $(\bbeta, \bgamma, \bM)$ the following inequality holds:
\begin{align} \notag
    \bracksd{\nabla_{\estim\bstheta}\bracks{\frac{1}{2}\normfrob{\bomega\circ\bracks{\bY-\estim\btheta}}}}{\bM-\estim\bM+\bX\bracks{\bbeta-\estim\bbeta}+\bracks{\bgamma-\estim\bgamma}\bZ} \ge \\
    \lambb \bracksd{\gradb}{\estim\bbeta-\bbeta} +
    \lambg \bracksd{\gradg}{\estim\bgamma-\bgamma} +
    \lambm \bracksd{\gradm}{\estim\bM-\bM}, \label{eq:optimality}
\end{align}
where the gradient is given by $\nabla_{\estim\bstheta}\bracks{\frac{1}{2}\normfrob{\bomega\circ\bracks{\bY-\estim\btheta}}} = 
\bracks{\bomega\circ \bracks{\estim\btheta-\bY}}$.
%\end{proposition}
%\begin{remark}
The inequality \eqref{eq:optimality} allows us to bound the estimation error. This relies on the convexity of the norms, which implies:
\begin{align*}
    \bracksd{\gradb}{\estim\bbeta-\bbeta} &\ge \normlass{\estim\bbeta}-\normlass{\bbeta} \\
    \bracksd{\gradg}{\estim\bgamma-\bgamma} &\ge \normlass{\estim\bgamma}-\normlass{\bgamma} \\
    \bracksd{\gradm}{\estim\bM-\bM} &\ge \normnuc{\estim\bM}-\normnuc{\bM}.
\end{align*}
%--------------------------------------------------------------------------------------

\subsection{Proof of Lemma \ref{lemma:norm_bound1}}
We establish part (i) only, as part (ii) follows by analogous arguments.
Applying inequality \eqref{eq:optimality} to the tuple $\bracks{\truev\bbeta,\estim\bgamma,\estim\bM}$ and by the convexity of the norms we have 
\begin{align*}
    \lambb \bracks{\normlass{\estim\bbeta} - \normlass{\truev{\bbeta}}} & \leq  \bracksd{\bomega \circ \bracks{\estim\btheta-\bY}}{-\bX\Delta\bbeta} \\
    & = \bracksd{\bomega \circ \bracks{\estim\btheta-\bY\pm\truev\btheta}}{-\bX\Delta\bbeta} \\
    & = \bracksd{\bomega \circ \bracks{\Delta\bM + \bX\Delta\bbeta +\Delta\bgamma\bZ - \bepsilon}}{-\bX\Delta\bbeta} \\
    & \leq  \bracks{\norminf{\tr\bX\bracks{\bomega\circ\Delta\bM}} +\norminf{\tr\bX\bracks{\bomega\circ\Delta\bgamma\bZ}}+\errinfx} 
    \normlass{\Delta\bbeta} \\
    & \leq  \bracks{\norminf{\tr\bX\bracks{\bomega\circ\Delta\bM}} + \norminf{\tr\bX\bracks{\bomega\circ\Delta\bgamma\bZ}}+\errinfx} 
    \bracks{\normlass{\estim\bbeta} + \normlass{\truev{\bbeta}}},
\end{align*}
where we used the fact that $\bracksd{\bomega\circ\bX\Delta\bbeta}{-\bX\Delta\bbeta} = -\normfrob{\bomega\circ\bX\Delta\bbeta}< 0$ to drop the quadratic term. The remaining terms are bounded using Hölder's inequality and the stated assumptions. Rearranging the terms yields
\begin{align*}
 &   \normlass{\estim\bbeta}\bracks{\lambb -  \bracks{\norminf{\tr\bX\bracks{\bomega\circ\Delta\bM}} + \norminf{\tr\bX\bracks{\bomega\circ\Delta\bgamma\bZ}}+\errinfx} }\\ &\qquad\leq  \normlass{\truev{\bbeta}} \bracks{\lambb + 
     \bracks{\norminf{\tr\bX\bracks{\bomega\circ\Delta\bM}} + \norminf{\tr\bX\bracks{\bomega\circ\Delta\bgamma\bZ}}+\errinfx} 
    }.
\end{align*}
Let $\lambb \ge 2 \bracks{\norminf{\tr\bX\bracks{\bomega\circ\Delta\bM}} + \norminf{\tr\bX\bracks{\bomega\circ\Delta\bgamma\bZ}}+\errinfx} $. Substituting this bound into the preceding inequality yields
\begin{align*}
    \normlass{\estim\bbeta} \leq 3 \normlass{\truev{\bbeta}},
\end{align*}
or equivalently,
\[
\|\Delta \boldsymbol{\beta}\|_1 \leq 4 \|{\truev{\bbeta}}\|_1.
\]
\qed

\subsection{Proof of Lemma \ref{lemma:norm_bound2}}
Lemma \ref{lemma:norm_bound2} is a special case of Lemma \ref{lemma:norm_bound3}, the proof of which is given in Appendix D. 
\qed

\subsection{Proof of Lemma \ref{lemma:rsc1}}
We prove part (i); part (ii) follows by analogous arguments.
We show that, with high probability, the estimation error $\Delta\bbeta$ lies in a restricted set. Let $d_1 = 4\|{\truev{\bbeta}}\|_1$. Then by Lemma \ref{lemma:norm_bound1}, it follows that
$\Delta\boldsymbol{\beta}\in\mathcal{A}(d_1,d_2)$ for some $d_2$.  Define 
\[
d_2 = \frac{128\,\cxb^2\,  \log(n+m)}{\pi_{\text{min}}\log(6/5)},
\]
 and consider the event $\mathcal{B}$, which corresponds to the failure of Lemma \ref{lemma:rsc1}. This event characterizes large deviations of the empirical error $\normfrob{\bomega\circ\bX\bbeta}$ from its expectation. Specifically, $\mathcal{B}$ occurs if there exists a candidate $\bbeta$ in the restricted cone for which the observed training error deviates from its expected value by more than half the expectation plus an additive error term $D_\beta$. It suffices to show that $\mathcal{B}$ occurs with low probability to prove Lemma \ref{lemma:rsc1}.  
\[
\mathcal{B} = \Bigl\{
  \exists\,\boldsymbol{\beta}\in\cA\,\Big|\,
  \bigl|
    \|\boldsymbol{\bomega}\circ\mathbf{X}\boldsymbol{\beta}\|_F^2
    - \mathbb{E}\|\boldsymbol{\bomega}\circ\mathbf{X}\boldsymbol{\beta}\|_F^2
  \bigr|
  > \tfrac12\mathbb{E}\|\boldsymbol{\bomega}\circ\mathbf{X}\boldsymbol{\beta}\|_F^2 + D_\beta
\Bigr\}.
\]
To bound the probability of $\mathcal{B}$, we use the peeling argument of \citet{klopp_Matrix_2015}. Let $\eta = 6/5$. For $l \in \mathbb{N}$, define the shells
\[
\cA_1(l)
  = \Bigl\{
      \boldsymbol{\beta}\in\cA\,\Big|\,
      \eta^{\,l-1}d_2 \leq \mathbb{E}\|\boldsymbol{\bomega}\circ\mathbf{X}\boldsymbol{\beta}\|_F^2 \leq \eta^{\,l}d_2
    \Bigr\}.
\]
If the event $\mathcal{B}$ occurs for some $\boldsymbol{\beta} \in \cA$, then \(\boldsymbol{\beta} \in \cA_1(l)\) for some index $l$, and
\begin{align} \label{eq:underB}
\bigl|
\|\boldsymbol{\bomega}\circ\mathbf{X}\boldsymbol{\beta}\|_F^2
  - \mathbb{E}\|\boldsymbol{\bomega}\circ\mathbf{X}\boldsymbol{\beta}\|_F^2
\bigr|
&> \tfrac12\mathbb{E}\|\boldsymbol{\bomega}\circ\mathbf{X}\boldsymbol{\beta}\|_F^2 + D_\beta \\[4pt] \notag
&> \tfrac12\eta^{\,l-1}d_2 + D_\beta \\[4pt]
&= \tfrac{5}{12}\eta^{\,l}d_2 + D_\beta. \notag
\end{align}
For $T > d_2$, consider the set
\[
\cA_2(T)
  = \Bigl\{
      \boldsymbol{\beta}\in\cA\,\Big|\,
      \mathbb{E}\|\boldsymbol{\bomega}\circ\mathbf{X}\boldsymbol{\beta}\|_F^2 \leq  T
    \Bigr\},
\]
and

\[
\mathcal{B}_l
  = \Bigl\{
      \exists\,\boldsymbol{\beta}\in{\mathcal{A}}\,\!\bigl(d_1,\eta^{\,l}d_2\bigr)\,\Big|\,
      \bigl|
        \|\boldsymbol{\bomega}\circ\mathbf{X}\boldsymbol{\beta}\|_F^2
        - \mathbb{E}\|\boldsymbol{\bomega}\circ\mathbf{X}\boldsymbol{\beta}\|_F^2
      \bigr|
      > \tfrac{5}{12}\eta^{\,l}d_2 + D_\beta
    \Bigr\}.
\]
Observe that \(\boldsymbol{\beta} \in \cA_1(l)\) implies \(\boldsymbol{\beta} \in \cA_2( \eta^l d_2)\).  Consequently, \eqref{eq:underB} implies that \(\mathcal{B}_l\) occurs, leading to the inclusion \(
\mathcal{B} \subset \bigcup_{l=1}^{\infty} \mathcal{B}_l
\). It therefore suffices to bound the probability of each event \(\mathcal{B}_l\) and apply the union bound. 

We now derive a concentration bound for $\mathcal{B}_l$. Let
\begin{align*}
Z_T
  &= \sup_{\boldsymbol{\beta}\in\cA_2(T)}
     \bigl|
       \|\boldsymbol{\bomega}\circ\mathbf{X}\boldsymbol{\beta}\|_F^2
       - \mathbb{E}\|\boldsymbol{\bomega}\circ\mathbf{X}\boldsymbol{\beta}\|_F^2
     \bigr| \\[4pt]
  &= \sup_{\boldsymbol{\beta}\in\cA_2(T)}
     \Bigl|
       %\sum_{(i,j):\,\bomega_{ij}=1}
       \sum_{i,j}
       \bigl[\bomega_{ij} (\bX_i\boldsymbol{\beta}_j)^2 - \pi_{ij}(\bX_i\boldsymbol{\beta}_j)^2\bigr]
     \Bigr| \\[4pt]
  &= \sup_{\boldsymbol{\beta}\in\cA_2(T)}
     \bigl| \sum_{i,j}
        (\bomega_{ij} - \pi_{ij})(\bX_i\boldsymbol{\beta}_j)^2\bigr |. %\\[4pt]
\end{align*}
We need an upper bound on $\mathbb{E}\bracksc{Z_T}$. We state the following result from \citet{klopp_Matrix_2015} (see also Theorem 3.3 in \citet{chatterjee_Matrix_2015}).
\begin{lemma}[\citet{klopp_Matrix_2015}, Theorem 11]\label{theorem:klopp}
Suppose that \(f : [-1,1]^N \to \mathbb{R}\) is a convex, Lipschitz
function with Lipschitz constant \(L\).
Let \(\Xi_1,\ldots,\Xi_N\) be independent random variables taking values
in \([-1,1]\), and set \(Z := f(\Xi_1,\ldots,\Xi_N)\). Then, for any
\(t \geq 0\),
\[
\mathbb{P}\,\bigl(|Z - \mathbb{E}Z| \geq 16L + t\bigr)
  \leq 4\exp\!\left(-\frac{t^{2}}{2L^{2}}\right).
\]
\end{lemma}

Define the random variables  \(c_{ij} = \boldsymbol{\bomega}_{ij}\). These variables satisfy the conditions of Lemma \ref{theorem:klopp}. We may express $Z_T$ as
\[
f\,\bigl(c_{11},\ldots,c_{mn}\bigr)
  = Z_T
  =  \sup_{\boldsymbol{\beta}\in\cA_2(T)}
     \left| \sum_{i,j}
        \bracks{c_{ij} - \pi_{ij}}\bracks{\bX_i\boldsymbol{\beta}_j}^2\right |
  = \sup_{\boldsymbol{\beta}\in\cA_2(T)}
      \left |\bracksd{
        \operatorname{vec}\,\bracks{\boldsymbol{c}\boldsymbol{-\pi}}}{
        \operatorname{vec}\,\bigl(\mathbf{X}\boldsymbol{\beta}\bigr)^{2}
      }\right |.
\]
Because the term
\(\bracksd{ \operatorname{vec}\bracks{\boldsymbol{c}-\boldsymbol{\pi}}}{
   \operatorname{vec}\bracks{\mathbf{X}\boldsymbol{\beta}}^{2}}\)
is an affine function of \(\boldsymbol{c}\), and that every real-valued affine
function is convex, the function
\(f\bracks{\operatorname{vec}\,(\boldsymbol{c})}\) is convex. We now bound the Lipschitz constant of $f$ as follows

\begin{align*}
\bigl|f({\boldsymbol{c}})-f({\boldsymbol{b}})\bigr|
&=\Bigl\lvert
   \sup_{\boldsymbol{\beta}\in\cA_2(T)}
       \bigl|\sum_{i=1}^{n}\sum_{j=1}^{m}
           (c_{ij}-\pi_{ij})\,(\mathbf{X}_i\boldsymbol{\beta}_j)^2 \bigr|
   -\sup_{\boldsymbol{\beta}\in\cA_2(T)}
       \bigl| \sum_{i=1}^{n}\sum_{j=1}^{m}
           (b_{ij}-\pi_{ij})\,(\mathbf{X}_i\boldsymbol{\beta}_j)^2
 \bigr|\Bigr\rvert \\[4pt]
 &\leq\sup_{\boldsymbol{\beta}\in\cA_2(T)}
 \Bigl\lvert  
       \bigl|\sum_{i=1}^{n}\sum_{j=1}^{m}
           (c_{ij}-\pi_{ij})\,(\mathbf{X}_i\boldsymbol{\beta}_j)^2 \bigr|
   -
       \bigl| \sum_{i=1}^{n}\sum_{j=1}^{m}
           (b_{ij}-\pi_{ij})\,(\mathbf{X}_i\boldsymbol{\beta}_j)^2
 \bigr|\Bigr\rvert \\[4pt]
 &\leq\sup_{\boldsymbol{\beta}\in\cA_2(T)}
 \Bigl\lvert  
       \sum_{i=1}^{n}\sum_{j=1}^{m}
           (c_{ij}-\pi_{ij})\,(\mathbf{X}_i\boldsymbol{\beta}_j)^2 
   -
        \sum_{i=1}^{n}\sum_{j=1}^{m}
           (b_{ij}-\pi_{ij})\,(\mathbf{X}_i\boldsymbol{\beta}_j)^2
 \Bigr\rvert \\[4pt]
&\leq
   \sup_{\boldsymbol{\beta}\in\cA_2(T)}
   \Bigl\lvert
     \sum_{i=1}^{n}\sum_{j=1}^{m}
       \bigl(c_{ij}-b_{ij}\bigr)\,(\mathbf{X}_i\boldsymbol{\beta}_j)^2
   \Bigr\rvert  \tag{Cauchy--Schwarz}\\[4pt]
&\leq
   \sup_{\boldsymbol{\beta}\in\cA_2(T)}
     \sqrt{\sum_{i=1}^{n}\sum_{j=1}^{m}
            \pi_{ij}^{-1}\bigl(c_{ij}-b_{ij}\bigr)^{2}}
     \sqrt{\sum_{i=1}^{n}\sum_{j=1}^{m}
            \pi_{ij}\,(\mathbf{X}_i\boldsymbol{\beta}_j)^{4}}
           %
           % \tag{$(\sum_{k=1}^{p}\bX_{ik}\boldsymbol{\beta}_{kj})^{2}
        % \leq (\sum_{k=1}^{p}\|\bX_i\|_\infty \|\bbeta_j\|_\infty)^2 \leq 4\cxb^2$}
            \\[4pt]
&\leq
   \sqrt{\pi_{\min}^{-1}}\,
   \normf{{\boldsymbol{c}}-{\boldsymbol{b}}}\,
   \sup_{\boldsymbol{\beta}\in\cA_2(T)}
     \sqrt{\sum_{i=1}^{n}\sum_{j=1}^{m}
           4\cxb^2\,\pi_{ij}\,(\mathbf{X}_i\boldsymbol{\beta}_j)^{2}}
   \\[-2pt]
%&\hphantom{\leq{}}%
%    \\[4pt]
&\leq
   2\cxb\,\pi_{\min}^{-1/2}\,
   \normf{{\boldsymbol{c}}-{\boldsymbol{b}}}\,
   \sqrt{\mathbb{E}\bigl\|\boldsymbol{\bomega}\circ\mathbf{X}\boldsymbol{\beta}\bigr\|_{F}^{2}} \\[4pt]
&\leq
   2\cxb\,\pi_{\min}^{-1/2}\,
   \normf{{\boldsymbol{c}}-{\boldsymbol{b}}}\,\sqrt{T}.
\end{align*}
Consequently, the Lipschitz constant satisfies
\[
\frac{\bigl|f({\boldsymbol{c}})-f({\boldsymbol{b}})\bigr|}
     {\normf{{\boldsymbol{c}}-{\boldsymbol{b}}}}
   \leq 2\cxb\,\pi_{\min}^{-1/2}\sqrt{T}. 
\]
We then set $L=2\cxb\,\pi_{\min}^{-1/2}\sqrt{T}$ and bound $16L \leq 16\bracks{T/192 + 192\, \cxb^2 \pi_{\min}^{-1}}$ by Young's inequality. Substituting this value into Lemma \ref{theorem:klopp} yields
\begin{align*}
\mathbb{P}\,\Bigl(
  \bigl|Z_T - \mathbb{E}Z_T\bigr|
  \geq \frac{1}{12}T + 3072\,\cxb^2\pi_{\min}^{-1} + t
\Bigr)
  &\leq
  4\exp\!\left(
    -\frac{t^{2}\,\pi_{\min}}{8\,T\cxb^2}
  \right).
\end{align*}
Setting \( t =  T/4 \), we obtain
\begin{align}
\mathbb{P}\,\!\Bigl(
  \bigl|Z_T - \mathbb{E}Z_T\bigr|
  \geq \frac{1}{3}T + 3072\, \cxb^2\pi_{\min}^{-1}
\Bigr)
  &\leq
  4\exp\!\left(
    -\frac{T\,\pi_{\min}}{128\, \cxb^2}
  \right).
  \label{eq:bound_Zt_1}
\end{align}
Having established that \( Z_T \) is concentrated around its expectation, we next upper-bound the expectation. Let \(\{E_{ij}\}\) be an i.i.d. Rademacher sequence, and define
\(
\boldsymbol{\Sigma}_R
  = \sum_{i,j} E_{ij}\bomega_{ij}\,\be i \tr{\bff{j}},
\)
where $\be i \in \Ra n$ and $\bff j \in \Ra m$  are standard basis vectors. By symmetry \citep{ledoux_Concentration_2001} and using the contraction inequality (\citealp[Theorem 2.2]{koltchinskii_Oracle_2011}), we have
\begin{align}
    \mathbb{E}(Z_T)
        &= \mathbb{E}\!\left[
                \sup_{\boldsymbol{\beta}\in\cA_2(T)}
                \left\lvert
                    \bigl\lVert\boldsymbol{\bomega}\circ\mathbf{X}\boldsymbol{\bbeta}\bigr\rVert_F^{2}
                    - \mathbb{E}\bigl\lVert\boldsymbol{\bomega}\circ\mathbf{X}\boldsymbol{\bbeta}\bigr\rVert_F^{2}
                \right\rvert
            \right] \notag\\
            & = \mathbb{E}\!\left[
                \sup_{\boldsymbol{\bbeta}\in\cA_2(T)}
                \left\lvert 
                \sum_{i,j} (\bomega_{ij}-\pi_{ij})\,\bigl(X_{i}\bbeta_{j}\bigr)^{2}
                \right\rvert
            \right] \notag\\
        &\leq
            2\,
            \mathbb{E}\!\left[
                \sup_{\boldsymbol{\beta}\in\cA_2(T)}
                \left\lvert
                    \sum_{i,j}^{n,m} E_{ij} \bomega_{ij}\,\bigl(X_{i}\bbeta_{j}\bigr)^{2}
                \right\rvert
            \right] \notag\\
        &\leq
            16\, \cxb\,
            \mathbb{E}\!\left[
                \sup_{\boldsymbol{\beta}\in\cA_2(T)}
                \left\lvert
                    \sum_{i,j}^{n,m} E_{ij} \bomega_{ij}\,X_{i}\bbeta_{j}
                \right\rvert
            \right] \label{ineq:talag}\\
        &= 16\, \cxb\, 
           \mathbb{E}\!\left[
               \sup_{\boldsymbol{\beta}\in\cA_2(T)}
               \left\lvert
                   \langle \boldsymbol{\Sigma}_R,\,\mathbf{X}\boldsymbol{\beta}\rangle
               \right\rvert
           \right] \notag\\
        &\leq
           16\, \cxb\,
           \mathbb{E}\!\left[
               \sup_{\boldsymbol{\beta}\in\cA_2(T)}
               \normop{\Sigma_R} 
               \normnuc{\bX\bbeta}
           \right] \notag\\
           &\leq
           16\, \cxb\,
           \mathbb{E}\!\left[
               \sup_{\boldsymbol{\beta}\in\cA_2(T)}
               \normop{\Sigma_R}
               \sqrt{p}
               \normf{\bX\bbeta}
           \right] \notag\\
        &\leq \label{eq:Expec_bound}
           16\,\cxb\,\sqrt{p\,\pim^{-1}T}\,\E \normop{\Sigma_R},
\end{align}
where \eqref{ineq:talag} follows from applying the contraction inequality. Substituting the expectation bound \eqref{eq:Expec_bound} into the concentration inequality \eqref{eq:bound_Zt_1} and using
\[\frac{1}{3}T+ 16\,\cxb\,\sqrt{p\,\pim^{-1}T}\,\E \normop{\Sigma_R} \leq \bracks{\frac{1}{3}+\frac{1}{12}}T + 768\, \bracks{\E\normop{\Sigma_R}}^2\,p\, \cxb^2\,\pim^{-1},\]
we have 
\begin{align*}
    4\,
        \exp\!\left(
            -\frac{T\,\pi_{\text{min}}}{128\,\cxb^2}
        \right) & \ge
        \mathbb{P}\,\!\Bigl(
  \bigl|Z_T - \mathbb{E}(Z_T)\bigr|
  \geq \frac{1}{3}T + 3072\, \cxb^2\pi_{\min}^{-1}
\Bigr) \\ & \ge
        \mathbb{P}\,\!\Bigl(
  Z_T 
  \geq \mathbb{E}(Z_T) +\frac{1}{3}T + 3072\, \cxb^2\pi_{\min}^{-1}
\Bigr)
       \\ & \ge 
        \mathbb{P}\,\left(
        Z_T
        \geq 
        \frac{5}{12}T
        \;+\; 768\,\cxb^2\,\pim^{-1} \bracks{ p\bracks{\E\normop{\Sigma_R}}^2+4}
    \right)\\
        &= \mathbb{P}\,\left(
                Z_T
                \geq
                \frac{5}{12}T
                \;+\;
                D_\beta
           \right)
        ,
\end{align*}
where
\(
    D_\beta
        = 768\,\cxb^2\,\pim^{-1} \bracks{ p\bracks{\E\normop{\Sigma_R}}^2+4}
\).
Finally, we apply the peeling argument. Set
\(
    d_2
    = \frac{128\,\cxb^2\,\,\log(n+m)}{\log(\eta)\,\pi_{\text{min}}}.
\)
Applying the union bound over the events $\mathcal{B}_l$ with \(T=\eta^{l}d_2\) and $\eta=6/5$:
\begin{align*}
    \mathbb{P}(\mathcal{B})
        &\leq
            \sum_{l=1}^{\infty}\mathbb{P}(\mathcal{B}_l)\\
        &\leq
            4\sum_{l=1}^{\infty}
            \exp\!\left(
                -\frac{\eta^{l}d_2\,\pi_{\text{min}}}{128\,\cxb^2\,}
            \right)\\
        &\leq
            4\sum_{l=1}^{\infty}
            \exp\!\left(
                -\frac{\log(\eta)\,l\,d_2\,\pi_{\text{min}}}{128\,\cxb^2\,}
            \right)
            \tag{using \(\exp(x)\geq x\)}\\
        &\leq
            4\sum_{l=1}^{\infty}
            \exp\!\left(-l\log(n+m)\right)
            \\%\tag{substituting \(d_2\)}\\
        &\leq
            4\,
            \frac{
                \exp\!\left(-\log(n+m)\right)
            }{
                1-\exp\!\left(-\log(n+m)\right)
            }
            \tag{geometric series sum}\\
        &= 4\,
           \frac{(n+m)^{-1}}{1-(n+m)^{-1}}\\[2pt]
        &= \frac{4}{\,n+m-1\,}\\
        &\leq \frac{8}{\,n+m\,},
\end{align*}
where we used $\bracks{n+m-1}^{-1} \leq 2\bracks{n+m}^{-1}$ for any $n,m\ge1$.

Finally, we state the following lemma from \citet{klopp_Matrix_2015} and \citet{robin_Lowrank_2018}.
\begin{lemma} [\citet{robin_Lowrank_2018}, Lemma 4] \label{lemma:Rademacher}
    Let Assumption 3 hold. Then, there exists a constant $c^*$ such that the following inequality holds
    \[\E \normop{\Sigma_R} \le  c^* \bracksb{\sqrt{\pi_{\max} \bracks{n\vee m}}+\sqrt{\log\bracks{n\wedge m}}}.\]
\end{lemma}
\qed

\subsection{Proof of Lemma \ref{lemma:rsc2}}
\label{proof_lemma:rsc2}
We follow similar steps as the previous proof. We will show that, with high probability, the estimation error $\Delta\bM$ lies in a restricted set. By Lemma \ref{lemma:norm_bound2}, it follows that
$\Delta\bM\in\cS(d_5)$ for some $d_5$. Let 
\[
d_5
    = \frac{128\,c^2_m\,\log(n+m)}{\log(6/5)\,\pi_{\text{min}}}.
\] and consider the event $\mathcal{B}$, which corresponds to the failure of the lemma. This event characterizes large deviations of the empirical error $\normfrob{\bomega\circ\bM}$ from its expectation. Specifically, $\mathcal{B}$ occurs if there exists a candidate $\bM$ in the restricted cone for which the observed training error deviates from its expected value by more than half the expectation plus an additive error term $D_M$. It suffices to show that $\mathcal{B}$ occurs with low probability to prove the lemma.  
\[
\mathcal{B} = \Bigl\{
  \exists\,\bM\in\cS\,\Big|\,
  \bigl|
    \|\boldsymbol{\bomega}\circ\bM\|_F^2
    - \mathbb{E}\|\boldsymbol{\bomega}\circ\bM\|_F^2
  \bigr|
  > \tfrac12\mathbb{E}\|\boldsymbol{\bomega}\circ\bM\|_F^2 + D_M
\Bigr\}.
\]
To bound the probability of $\mathcal{B}$, we use the peeling argument of \citet{klopp_Matrix_2015}. Let $\eta = 6/5$. For $l \in \mathbb{N}$, define the shells
\[
\cS_1(l)
  = \Bigl\{
      \bM\in\cS\,\Big|\,
      \eta^{\,l-1}d_5 \leq \mathbb{E}\|\boldsymbol{\bomega}\circ\bM\|_F^2 \leq \eta^{\,l}d_5
    \Bigr\}.
\]
If the event $\mathcal{B}$ occurs for some $\bM \in \cS$, then \(\bM \in \cS_1(l)\) for some index $l$, and
\begin{align} \label{eq:underB1}
\bigl|
\|\boldsymbol{\bomega}\circ\bM\|_F^2
  - \mathbb{E}\|\boldsymbol{\bomega}\circ\bM\|_F^2
\bigr|
&> \tfrac12\mathbb{E}\|\boldsymbol{\bomega}\circ\bM\|_F^2 + D_M \\[4pt] \notag
&> \tfrac12\eta^{\,l-1}d_5 + D_M \\[4pt]
&= \tfrac{5}{12}\eta^{\,l}d_5 + D_M. \notag
\end{align}
For $T > d_5$, consider the set
\[
\cS_2(T)
  = \Bigl\{
      \bM\in\cS\,\Big|\,
      \mathbb{E}\|\boldsymbol{\bomega}\circ\bM\|_F^2 \leq  T
    \Bigr\},
\]
and 

\[
\mathcal{B}_l
  = \Bigl\{
      \exists\,\bM\in{\cS}\,\!\bigl(\eta^{\,l}d_5\bigr)\,\Big|\,
      \bigl|
        \|\boldsymbol{\bomega}\circ\bM\|_F^2
        - \mathbb{E}\|\boldsymbol{\bomega}\circ\bM\|_F^2
      \bigr|
      > \tfrac{5}{12}\eta^{\,l}d_5 + D_M
    \Bigr\}.
\]
Observe that \(\bM \in \cS_1(l)\) implies \(\bM \in \cS_2( \eta^l d_5)\).  Consequently, \eqref{eq:underB1} implies that \(\mathcal{B}_l\) occurs, leading to the inclusion \(
\mathcal{B} \subset \bigcup_{l=1}^{\infty} \mathcal{B}_l
\). It therefore suffices to bound the probability of each event \(\mathcal{B}_l\) and apply the union bound. 

We now derive a concentration bound for $\mathcal{B}_l$. Let
\begin{align*}
Z_T
  &= \sup_{\bM\in\cS_2(T)}
     \bigl|
       \|\boldsymbol{\bomega}\circ\bM\|_F^2
       - \mathbb{E}\|\boldsymbol{\bomega}\circ\bM\|_F^2
     \bigr| \\[4pt]
  &= \sup_{\bM\in\cS_2(T)}
     \Bigl|
       %\sum_{(i,j):\,\bomega_{ij}=1}
       \sum_{i,j}
       \bigl[\bomega_{ij} \bM_{ij}^2 - \pi_{ij}\bM_{ij}^2\bigr]
     \Bigr| \\[4pt]
  &= \sup_{\bM\in\cS_2(T)}
     \bigl| \sum_{i,j}
        (\bomega_{ij} - \pi_{ij})\bM_{ij}^2\bigr |. %\\[4pt]
\end{align*}
We require an upper bound on $\mathbb{E}\bracksc{Z_T}$. Define the random variables  \(c_{ij} = \boldsymbol{\bomega}_{ij}\). These variables satisfy the conditions of Lemma \ref{theorem:klopp}. We may express $Z_T$ as
\[
f\,\bigl(c_{11},\ldots,c_{mn}\bigr)
  = Z_T
  =  \sup_{\bM\in\cS_2(T)}
     \left| \sum_{i,j}
        \bracks{c_{ij} - \pi_{ij}}\bM_{ij}^2\right |
  = \sup_{\bM\in\cS_2(T)}
      \left |\bracksd{
        \operatorname{vec}\,\bracks{\boldsymbol{c}\boldsymbol{-\pi}}}{
        \operatorname{vec}\,\bigl(\bM\bigr)^{2}
      }\right |.
\]
Because the term
\(\bracksd{ \operatorname{vec}\bracks{\boldsymbol{c}-\boldsymbol{\pi}}}{
   \operatorname{vec}\bracks{\bM}^{2}}\)
is an affine function of \(\boldsymbol{c}\), and that every real-valued affine
function is convex, the function
\(f\bracks{\operatorname{vec}\,(\boldsymbol{c})}\) is convex. We now bound the Lipschitz constant of $f$. 

\begin{align*}
\bigl|f({\boldsymbol{c}})-f({\boldsymbol{b}})\bigr|
&=\Bigl\lvert
   \sup_{\bM\in\cS_2(T)}
       \bigl|\sum_{i=1}^{n}\sum_{j=1}^{m}
           (c_{ij}-\pi_{ij})\,\bM_{ij}^2 \bigr|
   -\sup_{\bM\in\cS_2(T)}
       \bigl| \sum_{i=1}^{n}\sum_{j=1}^{m}
           (b_{ij}-\pi_{ij})\,\bM_{ij}^2
 \bigr|\Bigr\rvert \\[4pt]
 &\leq\sup_{\bM\in\cS_2(T)}
 \Bigl\lvert  
       \bigl|\sum_{i=1}^{n}\sum_{j=1}^{m}
           (c_{ij}-\pi_{ij})\,\bM_{ij}^2 \bigr|
   -
       \bigl| \sum_{i=1}^{n}\sum_{j=1}^{m}
           (b_{ij}-\pi_{ij})\,\bM_{ij}^2
 \bigr|\Bigr\rvert \\[4pt]
 &\leq\sup_{\bM\in\cS_2(T)}
 \Bigl\lvert  
       \sum_{i=1}^{n}\sum_{j=1}^{m}
           (c_{ij}-\pi_{ij})\,\bM_{ij}^2 
   -
        \sum_{i=1}^{n}\sum_{j=1}^{m}
           (b_{ij}-\pi_{ij})\,\bM_{ij}^2
 \Bigr\rvert \\[4pt]
&\leq
   \sup_{\bM\in\cS_2(T)}
   \Bigl\lvert
     \sum_{i=1}^{n}\sum_{j=1}^{m}
       \bigl(c_{ij}-b_{ij}\bigr)\,\bM_{ij}^2
   \Bigr\rvert  \tag{Cauchy--Schwarz}\\[4pt]
&\leq
   \sup_{\bM\in\cS_2(T)}
     \sqrt{\sum_{i=1}^{n}\sum_{j=1}^{m}
            \pi_{ij}^{-1}\bigl(c_{ij}-b_{ij}\bigr)^{2}}
     \sqrt{\sum_{i=1}^{n}\sum_{j=1}^{m}
            \pi_{ij}\,\bM_{ij}^{4}}
            \\[4pt]
            &\leq
   \sup_{\bM\in\cS_2(T)}
     \sqrt{\sum_{i=1}^{n}\sum_{j=1}^{m}
            \pi_{ij}^{-1}\bigl(c_{ij}-b_{ij}\bigr)^{2}}
     \sqrt{\sum_{i=1}^{n}\sum_{j=1}^{m} 4c_m^2
            \pi_{ij}\,\bM_{ij}^{2}}
            \\[4pt]
%&\hphantom{\leq{}}%
%    \\[4pt]
&\leq
   2\,c_m\,\pi_{\min}^{-1/2}\,
   \normf{{\boldsymbol{c}}-{\boldsymbol{b}}}\,
   \sqrt{\mathbb{E}\bigl\|\boldsymbol{\bomega}\circ\bM\bigr\|_{F}^{2}} \\[4pt]
&\leq
   2\,c_m\,\pi_{\min}^{-1/2}\, 
   \normf{{\boldsymbol{c}}-{\boldsymbol{b}}}\,\sqrt{T}.
\end{align*}
Consequently, the Lipschitz constant satisfies
\[
\frac{\bigl|f({\boldsymbol{c}})-f({\boldsymbol{b}})\bigr|}
     {\normf{{\boldsymbol{c}}-{\boldsymbol{b}}}}
   \leq 2 c_m\pi_{\min}^{-1/2}\sqrt{T}.
\]
We then set $L=2 c_m\pi_{\min}^{-1/2}\sqrt{T}$ and bound $16L \leq 16\bracks{ T/192 + 192 \, c_m^2 \, \pi_{\min}^{-1}}$ by Young's inequality. Substituting this value into Lemma \ref{theorem:klopp} yields
\begin{align*}
\mathbb{P}\,\Bigl(
  \bigl|Z_T - \mathbb{E}Z_T\bigr|
  \geq \frac{1}{12}T + 3072\,c_m^2\pi_{\min}^{-1} + t
\Bigr)
  &\leq
  4\exp\!\left(
    -\frac{t^{2}\,\pi_{\min}}{8\,T\,c_m^2}
  \right).
\end{align*}
Setting \( t =  T/4 \), we obtain
\begin{align}
\mathbb{P}\,\Bigl(
  \bigl|Z_T - \mathbb{E}Z_T\bigr|
  \geq \tfrac{1}{3}T + 3072\,c^2_m\, \pi_{\min}^{-1}
\Bigr)
  &\leq
  4\exp\!\left(
    -\frac{T\,\pi_{\min}}{128\,c^2_m}
  \right).
  \label{eq:bound_Zt_1_1}
\end{align}
Having established that \( Z_T \) is concentrated around its expectation, we next upper-bound the expectation.
Using the same symmetry and contraction inequality arguments as in the proof of Lemma \ref{lemma:rsc1}, we have
\begin{align}
    \mathbb{E}(Z_T)
        &= \mathbb{E}\!\left[
                \sup_{\bM\in\cS_2(T)}
                \left\lvert
                    \bigl\lVert\boldsymbol{\bomega}\circ\bM\bigr\rVert_F^{2}
                    - \mathbb{E}\bigl\lVert\boldsymbol{\bomega}\circ\bM\bigr\rVert_F^{2}
                \right\rvert
            \right] \notag\\
            & = \mathbb{E}\!\left[
                \sup_{\bM\in\cS_2(T)}
                \left\lvert 
                \sum_{i,j} (\bomega_{ij}-\pi_{ij})\,\bM_{ij}^{2}
                \right\rvert
            \right] \notag\\
        &\leq
            2\,
            \mathbb{E}\!\left[
                \sup_{\bM\in\cS_2(T)}
                \left\lvert
                    \sum_{i,j}^{n,m} E_{ij} \bomega_{ij}\,\bM_{ij}^{2}
                \right\rvert
            \right] \notag\\
        &\leq
            16\,c_m\,
            \mathbb{E}\!\left[
                \sup_{\bM\in\cS_2(T)}
                \left\lvert
                    \sum_{i,j}^{n,m} E_{ij} \bomega_{ij}\,\bM_{ij}
                \right\rvert
            \right] \notag\\
        &= 16\,c_m\,
           \mathbb{E}\!\left[
               \sup_{\bM\in\cS_2(T)}
               \left\lvert
                   \langle \boldsymbol{\Sigma}_R,\,\bM\rangle
               \right\rvert
           \right] \notag\\
        &\leq
           16\,c_m\,
           \mathbb{E}\!\left[
               \sup_{\bM\in\cS_2(T)}
               \normop{\boldsymbol{\Sigma}_R}
               \normnuc{\bM}
           \right] \notag\\
        &\leq \notag
           16\,c_m\,\sqrt{32\,r}\,\normf{\bM}\E{\normop{\boldsymbol{\Sigma}_R}}\\
           &\leq \notag
           16\,c_m\,\sqrt{32\,r\,\pi_{\min}^{-1}\,\E\normfrob{\bomega\circ\bM}}\,\E{\normop{\boldsymbol{\Sigma}_R}} \\
           &\leq
           16\,c_m\,\sqrt{32\,r\,\pi_{\min}^{-1}\, T}\,\E{\normop{\boldsymbol{\Sigma}_R}}. \label{eq:Expec_bound_1}
\end{align}

Substituting the expectation bound \eqref{eq:Expec_bound_1} into the concentration inequality \eqref{eq:bound_Zt_1_1} and using
\[\frac{1}{3}T + 16\,c_m\, \sqrt{32\,r\,\pi_{\min}^{-1}\, T}\,\E{\normop{\boldsymbol{\Sigma}_R}} \le \bracks{\frac{1}{3}+\frac{1}{12}}T + 768 \, c_m^2 \bracks{32\,r\,\pi_{\min}^{-1}}\,\bracks{\E{\normop{\boldsymbol{\Sigma}_R}}}^2,\]
we have 
\begin{align*}
  4\exp\!\left(
    -\frac{T\,\pi_{\min}}{128\,c^2_m}
  \right) &\geq \mathbb{P}\,\left(
                Z_T
                \geq
                \frac5{12}T + 3072\,c_m^2\, \pi_{\min}^{-1} \bracks{8\,r\,\bracks{\E{\normop{\boldsymbol{\Sigma}_R}}}^2+1} 
           \right) \notag\\
        &= \mathbb{P}\,\left(
                Z_T
                \geq
                \frac{5}{12}T
                \;+\;
                D_M
           \right),
\end{align*}
where
\(
    D_M
        = 3072\,c_m^2\, \pi_{\min}^{-1} \bracks{8\,r\,\bracks{\E{\normop{\boldsymbol{\Sigma}_R}}}^2+1}
\).
Finally, we apply the peeling argument. Set
\(
    d_5
    = \frac{128\,c^2_m\,\log(n+m)}{\log(\eta)\,\pi_{\text{min}}}.
\)
Applying the union bound over the events $\mathcal{B}_l$ with \(T=\eta^{l}d_5\) and $\eta=6/5$:
\begin{align*}
    \mathbb{P}(\mathcal{B})
        &\leq
            \sum_{l=1}^{\infty}\mathbb{P}(\mathcal{B}_l)\\
        &\leq
            4\sum_{l=1}^{\infty}
            \exp\!\left(
                -{\eta^{l}d_5\,c^{-2}_m\,\pi_{\text{min}}/128}
            \right)\\
        &\leq
            4\sum_{l=1}^{\infty}
            \exp\!\left(
                -{\log(\eta)\,l\,d_5\,c^{-2}_m\,\pi_{\text{min}}/128}
            \right)
            \tag{using \(\exp(x)\geq x\)}\\
        &\leq
            4\sum_{l=1}^{\infty}
            \exp\!\left(-l\log(n+m)\right)
            \\%\tag{substituting \(d_5\)}\\
        &\leq
            4\,
            \frac{
                \exp\!\left(-\log(n+m)\right)
            }{
                1-\exp\!\left(-\log(n+m)\right)
            }
            \tag{geometric series sum}\\
        &= 4\,
           \frac{(n+m)^{-1}}{1-(n+m)^{-1}}\\[2pt]
        &= \frac{4}{\,n+m-1\,}\\
        &\leq \frac{8}{\,n+m\,},
\end{align*}
where we used $\bracks{n+m-1}^{-1} \leq 2\bracks{n+m}^{-1}$ for any $n,m\ge1$.
\qed

\subsection{Proof of Lemma \ref{lemma:noise}}
\textbf{(i)}
We apply the Bernstein bound for random matrices \citep[Theorem 6.17]{wainwright_Highdimensional_2019}. Because this bound applies only to sums of independent, symmetric matrices, we symmetrize $\bomega\circ\bepsilon$ by defining the following matrix
\[
\tilde{\bepsilon} = \begin{pmatrix}
    0 & \bomega\circ\bepsilon \\
    \tr{\bracks{\bomega\circ\bepsilon}} & 0
\end{pmatrix} \in \R{(n+m)}{(n+m)}.
\]
Observe that $\normop{\bomega\circ\bepsilon} = \normop{\tilde{\bepsilon}}$. We express $\tilde{\bepsilon}$ as a sum of $mn$ independent matrices, one for each entry of $\bepsilon$. Defining the symmetric basis matrices
\[
\boldsymbol{Q}_{ij} = \begin{pmatrix}
    0 & \be i \tr{\bff{}}_j \\
    \bff j \tr{\be{}}_i & 0
\end{pmatrix},
\]
we have
\[
\tilde{\bepsilon} = \sum_{i=1}^n\sum_{j=1}^m \bomega_{ij}\bepsilon_{ij} \boldsymbol{Q}_{ij},
\]
where $\tilde{\bepsilon}_{ij}=\bomega_{ij}\bepsilon_{ij} \boldsymbol{Q}_{ij}$ is a matrix that contains zeros everywhere except in the off-diagonal blocks corresponding to the $\bracks{i,j}$ entry, which contains a zero-mean, $\sigma$-sub-Gaussian random variable. Before applying the Bernstein bound, we compute the operator norm of the matrix variance:
\begin{align*}
    \frac{1}{mn}\normop{\sum_{i=1}^n\sum_{j=1}^m
      \operatorname{Var}\bracks{\tilde{\bepsilon}_{ij}}}
    &= \frac{1}{mn}\normop{\sum_{i=1}^n\sum_{j=1}^m
         \E\bracksc{\bracks{\tilde{\bepsilon}_{ij}}^2}} \\
    &= \frac{1}{mn}\normop{\sum_{i=1}^n\sum_{j=1}^m
         \pi_{ij}\,\E\bracksc{\bepsilon_{ij}^2}\,
         \boldsymbol{Q}_{ij}^2} \\
    &\leq \frac{1}{mn}\normop{\sum_{i=1}^n\sum_{j=1}^m
         \pi_{\max}\,\sigma^2\,\boldsymbol{Q}_{ij}^2} \\
    &= \frac{1}{mn}\left\|\sum_{i=1}^n\sum_{j=1}^m
         \pi_{\max}\,\sigma^2\,\begin{pmatrix}
           \be i\tr{\be{}}_i & 0 \\
           0 & \bff j\tr{\bff{}}_j
         \end{pmatrix}\right\|_{\text{op}} \\
    &= \frac{\pi_{\max}\,\sigma^2}{mn}\left\|\begin{pmatrix}
           m\,\iden{n} & 0 \\
           0 & n\,\iden{m}
         \end{pmatrix}\right\|_{\text{op}} \\
    &= \frac{m \vee n}{mn}\pi_{\max}\,\sigma^2
     = \frac{\pi_{\max}\,\sigma^2}{m \wedge n},
\end{align*}
where the second equality uses the independence of
$\bomega_{ij}$ and $\bepsilon_{ij}$, so that
$\E[\bomega_{ij}^2\,\bepsilon_{ij}^2]
  = \E[\bomega_{ij}]\,\E[\bepsilon_{ij}^2]
  = \pi_{ij}\,\E[\bepsilon_{ij}^2]$;
the inequality uses the sub-Gaussian
property $\E[\bepsilon_{ij}^2]\leq\sigma^2$.
Applying the Bernstein bound yields
\[
\mathbb{P}\bracks{\frac{1}{mn} \normop{\sum_{i=1}^n\sum_{j=1}^m\tilde{\bepsilon}_{ij}}\ge t} \leq
2\, \bracks{m+n} \, \exp\bracksb{\frac{- mnt^2}{2\bracks{{\pi_{\max}\,\sigma^2/\bracks{m\wedge n}}+bt}}} \qquad \text{for all } t\ge 0,
\]
where $b$ is the Bernstein condition parameter. Assuming $t \leq \bracks{\pi_{\max}\,\sigma^2}/\bracks{\bracks{m \wedge n}b}$, the bound simplifies to

\[
\mathbb{P}\bracks{\frac{1}{mn} \normop{\sum_{i=1}^n\sum_{j=1}^m\tilde{\bepsilon}_{ij}}\ge t} \leq
2\, \bracks{m+n} \, \exp\bracksb{\frac{- mnt^2}{4\pi_{\max}\,\sigma^2/\bracks{m\wedge n}}}.
\]

For any $\delta \in \bracks{0,1}$, equating the right-hand side to $\delta$ gives
\begin{align*}
    \delta &= 2\, \bracks{m+n} \, \exp\bracksb{\frac{- mnt^2}{4\pi_{\max}\,\sigma^2/\bracks{m\wedge n}}},\\
    \log\bracks{\delta} &= \log\bracks{2\bracks{m+n}} -\frac{ mnt^2}{4\pi_{\max}\,\sigma^2/\bracks{m\wedge n}}, \\
    t^2 &= \log\bracks{\frac{2\bracks{m+n}}{\delta}} \frac{4\pi_{\max}\,\sigma^2 }{mn\bracks{m\wedge n}},\\
    t &= \sqrt{\log\bracks{\frac{2\bracks{m+n}}{\delta}}} \frac{2\sqrt{\pi_{\max}}\,\sigma }{\sqrt{mn\bracks{m\wedge n}}}.
\end{align*}
Consequently, with probability at least $1-\delta$ we obtain
\begin{align*}
\frac{1}{mn} \normop{\sum_{i=1}^n\sum_{j=1}^m\tilde{\bepsilon}_{ij}} &\le   \sqrt{\log\bracks{\frac{2\bracks{m+n}}{\delta}}} \frac{2\sqrt{\pi_{\max}}\,\sigma }{\sqrt{mn\bracks{m\wedge n}}}, \\
\normop{\sum_{i=1}^n\sum_{j=1}^m\tilde{\bepsilon}_{ij}} &\le 2\sigma  \sqrt{\pi_{\max}\bracks{m\vee n}\log\bracks{\frac{2\bracks{m+n}}{\delta}}}. 
\end{align*}

\textbf{(ii)-(iii)} We prove part (ii) only, as part (iii) follows by
analogous arguments. For notational simplicity, let
$\bQ = \tr\bX[\bomega \circ \bepsilon]$, so that
$\bQ_{ij} = \sum_{l=1}^n \bX_{li}\,\bomega_{lj}\,\bepsilon_{lj}$ for
$i = 1,\dots,p$ and $j = 1,\dots,m$. We claim that $\bQ_{ij}$ is
sub-Gaussian with parameter
$\sigma\sqrt{\sum_{l=1}^n \bX_{li}^2\,\bomega_{lj}}$. Intuitively, this
follows because $\bomega_{lj} \in \{0,1\}$ is bounded and
$\bepsilon_{lj}$ is zero-mean $\sigma$-sub-Gaussian, so the product
$\bomega_{lj}\bepsilon_{lj}$ is itself zero-mean and
$\sigma$-sub-Gaussian. Consequently, $\bQ_{ij}$, as a linear combination
of independent zero-mean sub-Gaussian variables, has a sub-Gaussian
distribution. We now verify this claim explicitly and derive a tail bound
for $\norminf{\bQ}$.

For all $t>0$ and $\lambda>0$, we have
\begin{align}
  \notag
  \mathbb{P}\bracks{\norminf{\bQ} > t}
    &= \mathbb{P}\bracks{\max_{i,j}|\bQ_{ij}| > t} \\
  \notag
    &\leq \sum_{i=1}^p \sum_{j=1}^m
      \mathbb{P}\bracks{|\bQ_{ij}| > t} \\
  \notag
    &= \sum_{i=1}^p \sum_{j=1}^m \bigl(
      \mathbb{P}\bracks{\bQ_{ij} > t}
      + \mathbb{P}\bracks{-\bQ_{ij} > t} \bigr) \\
  \notag
    &\leq  \sum_{i=1}^p \sum_{j=1}^m \bracks{
      \mathbb{P}\bracks{\expo{\lambda\bQ_{ij}}
        > \expo{\lambda\,t}} + \mathbb{P}\bracks{\expo{-\lambda\bQ_{ij}}
        > \expo{\lambda\,t}}} \\
    &\leq \expo{-\lambda\,t}
      \sum_{i=1}^p \sum_{j=1}^m \bracks{
      \E\bracks{\expo{\lambda\bQ_{ij}}}+\E\bracks{\expo{-\lambda\bQ_{ij}}}},
  \label{eq:lemma5_proof_iii}
\end{align}

where the second line uses Boole's inequality and the last two lines use the Chernoff bound.
Next, we bound the moment-generating function by conditioning on
$\bomega_{lj}$. Using the independence of
$\{\bepsilon_{lj}\}_{l=1}^n$, we obtain
\begin{align*}
  \E\bracks{\expo{\lambda\bQ_{ij}}}
    &= \prod_{l=1}^n
      \E\bracks{\expo{\lambda\bX_{li}\,\bomega_{lj}\,\bepsilon_{lj}}} \\
    &= \prod_{l=1}^n \Bigl[
      \mathbb{P}\bracks{\bomega_{lj} = 0}
      + \mathbb{P}\bracks{\bomega_{lj} = 1}\,
        \E\bracks{\expo{\lambda\bX_{li}\,\bepsilon_{lj}}}
    \Bigr] \\
    &= \prod_{l=1}^n \Bigl[
      (1 - \pi_{lj})
      + \pi_{lj}\,
        \E\bracks{\expo{\lambda\bX_{li}\,\bepsilon_{lj}}}
    \Bigr] \\
    &\leq \prod_{l=1}^n \Bigl[
      (1 - \pi_{lj})
      + \pi_{lj}\,
        \expo{\lambda^2 \bX_{li}^2 \sigma^2 / 2}
    \Bigr],
\end{align*}
where the last inequality uses the sub-Gaussian property of
$\bepsilon_{lj}$. We now apply the following inequality: for
$a_1 \in [0,1]$ and $a_2 \geq 0$,
\[
  1 - a_1 + a_1\,\expo{a_2} \leq \expo{a_2},
\]
which holds because
$\expo{a_2} - (1 - a_1 + a_1\,\expo{a_2})
= (1 - a_1)(\expo{a_2} - 1) \geq 0$.
Setting $a_1 = \pi_{lj}$ and
$a_2 = \lambda^2 \bX_{li}^2 \sigma^2 / 2$, we obtain
\begin{align*}
  \E\bracks{\expo{\lambda\bQ_{ij}}}
    &\leq \prod_{l=1}^n
      \expo{\lambda^2 \bX_{li}^2 \sigma^2 / 2} \\
    &= \expo{\lambda^2 \sigma^2
      \sum_{l=1}^n \bX_{li}^2 / 2} \\
    &\leq \expo{\lambda^2 n \sigma^2 / 2},
\end{align*}
where the last step uses $\norminf{\bX} \leq 1$. Because the upper bound on the moment-generating function depends on $\lambda^2$, it holds identically for both $\bQ_{ij}$ and $-\bQ_{ij}$. Substituting this bound back
into~\eqref{eq:lemma5_proof_iii} and
setting $\lambda = t / (n \sigma^2)$, we get
\begin{align*}
  \mathbb{P}\bracks{\norminf{\bQ} > t}
    &\leq 2 \expo{-\lambda\,t}
      \sum_{i=1}^p \sum_{j=1}^m
      \expo{\lambda^2 n \sigma^2 / 2} \\
    &= 2\,p\,m\,\expo{-\lambda\,t + \lambda^2 n \sigma^2 / 2} \\
    &= 2\,p\,m\,\expo{-t^2 / (2 \sigma^2 n)}.
\end{align*}
Finally, for any $\delta \in (0,1)$, setting the right-hand side equal to
$\delta$ and solving for $t$ yields
\begin{align*}
  \delta &= 2\,p\,m\,\expo{-t^2 / (2 \sigma^2 n)}, \\
  t &= \sigma\sqrt{2\,n\,\log\frac{2\,p\,m}{\delta}}.
\end{align*}
\qed

%=========================================================================
%=========================================================================

\section{Theoretical Properties Under the General Case}

In this appendix, we investigate the error bounds under the general case, wherein prior knowledge regarding the structural dependencies within $\truev\btheta$ is available. Following steps analogous to those in Appendix \ref{section:B}, we demonstrate that incorporating similarity matrices into the model does not alter the bounds on $\Delta\bbeta$ and $\Delta\bgamma$. We provide the lemmas requisite for computing the bound on $\Delta\bM$, which generalize Lemmas \ref{lemma:norm_bound2} and \ref{lemma:rsc2} from the special case of identity similarity matrices. We outline the proof to establish an upper bound on $\Delta\bM$. For conciseness, we omit intermediate steps that merely duplicate the derivations presented in Appendix \ref{section:B}. Our focus is restricted to the stages of the proof where the inequalities diverge from the prior analysis. Throughout the remainder of this appendix, we consider the following objective function, under the same optimization problem in \eqref{problem:thm1}.
\begin{align}
    \label{eq:problem_2}
    \loss{\bstheta} = \losssf{\bstheta} + \frac{1}{2}\lambm \bracks{\trace{\tr\bA\bSa\bA} + \trace{\tr\bB\bSb \bB}} + \lambb \normlass{\bbeta} + \lambg \normlass{\bgamma}.
\end{align}
We first introduce the following lemma, which establishes the nuclear norm equivalence for the trace penalties introduced above.
%--------------------------------------------------------------------------------------
\begin{lemma} \label{lemma:nucnormequiv} For $\bM$, $\bA$, $\bB$, $\bSa$, and $\bSb$ defined as in our model, the following holds
\begin{align} \label{lemma:nuclear_equiv}
    \normnuc{\bSa^{1/2}\bM\bSb^{1/2}} = \min_{\bA,\bB:\bM=\bA\tr\bB} \frac{1}{2} \bracks{\trace{\bA^\top\bSa \bA} + \trace{\bB^\top\bSb \bB} }  
\end{align}
\end{lemma}
\begin{proof}
    
By Proposition 1 of \citet{mazumder_spectral_2010}, we have
\begin{align*}2\normnuc{\bQa\bA\tr\bB\bQb} &\le
\trace{\bQa\bA\tr\bA\bQa} + \trace{\bQb\bB\tr\bB\bQb} \\
&= \normfrob{\bQa\bA} + \normfrob{\bQb\bB},
\end{align*}
and by Lemma 6 of \citet{mazumder_spectral_2010}, we obtain the equality
\begin{align*}2\normnuc{\bQa\bA\tr\bB\bQb} 
= \min_{\bA,\bB:\bM=\bA\tr\bB}\bracks{\normfrob{\bQa\bA} + \normfrob{\bQb\bB}}.
\end{align*}
\end{proof}

Lemma \ref{lemma:nuclear_equiv} allows an equivalent reformulation of the objective in \eqref{eq:problem_2} as follows
\begin{align}
    \label{eq:problem_3}
    \loss{\bstheta} = \losssf{\bstheta} + \lambm \normnuc{\bQa\bM\bQb} + \lambb \normlass{\bbeta} + \lambg \normlass{\bgamma}.
\end{align}

The sole distinction between the formulations in \eqref{eq:obj2} and \eqref{eq:problem_3} lies in the penalty applied to $\bM$, which does not influence the penalties on $\bbeta$ and $\bgamma$ or the loss function $\losssf{\bstheta}$. Consequently, Lemmas \ref{lemma:norm_bound1} and \ref{lemma:rsc1}, as well as parts (i) and (ii) of Theorem 1, remain valid. Thus, it suffices to establish an upper bound for $\Delta\bM$. To this end, we first introduce an extra assumption on the similarity matrices.
\begin{assumption}\label{assumption:similarity}
Let $\bK{r}\in\R{n}{n}$ and $\bK{c}\in\R{m}{m}$ be the raw kernel matrices associated with the rows and columns of the response matrix, and let $\bD_r$ and $\bD_c$ denote their diagonal degree matrices, respectively. To prevent spectral decay as the dimensions grow, we normalize the kernels before inverting them. This allows us to control the bounds of the singular values. We define the similarity matrices $\bSa$ and $\bSb$ as the regularized inverses of the normalized kernels
    \[\bSa=\bracks{\bD_r^{-1/2}\bK{r}\bD_r^{-1/2} + \gamma \iden{n}}^{-1},
    \qquad  \bSb=\bracks{\bD_c^{-1/2}\bK{c}\bD_c^{-1/2} + \gamma \iden{m}}^{-1},\]
    where $1 > \gamma > 0$ is a small positive constant. Consequently, $\bSa$ and $\bSb$ are symmetric positive definite, and their singular values are universally bounded in the interval $\bracksc{1/(1+\gamma), 1/\gamma}$, independent of $n$ and $m$. Therefore, there exist universal constants such that the spectral bounds for the square roots of the similarity matrices satisfy
    \begin{align*} 
    \amin &:= \min\bracks{  \sigma_{\min}(\bSa^{1/2}), \sigma_{\min}(\bSb^{1/2})}  \ge \frac{1}{\sqrt{1+\gamma}} > \frac{1}{\sqrt{2}}, \\ 
    \amax &:=\max \bracks{  \sigma_{\max}(\bSa^{1/2}), \sigma_{\max}(\bSb^{1/2})} \le \frac{1}{\sqrt{\gamma}}  < \infty, 
    \end{align*} where $\sigma_{\min}(\cdot)$ and $\sigma_{\max}(\cdot)$ denote the smallest and largest singular values, respectively. 
\end{assumption}

Let $r = \operatorname{rank}\bracks{\truev\bM}$. Recalling the definition of the good event $\mathcal{G}\bracks{\lambm}$, we define the following restricted cone
\begin{align*}
\mathcal{\wtilde{S}}(d_5, a) =& 
\Bigl \{  
\bM \in \mathbb{R}^{n\times m} \, \big | \,
\norminf{\bM} \le 2\,c_m, \;
\E\normfrob{\bomega\circ\bM} \ge d_5, \normnuc{\bM} \leq  a \sqrt{32r} \normf{\bM}  \;
\Bigr\}.
\end{align*}

First, observe that $ \cS\bracks{d_5} = \mathcal{\wtilde{S}}\bracks{d_5, 1}$, implying that $\cS \subset \wtilde{\cS}$. The following Lemma establishes that for $a = \amax^2/\bracks{2\amin^2-1}$, the estimation error satisfies ${\Delta\bM} \in \mathcal{\wtilde{S}}$ for some $d_5$.

\begin{lemma}\label{lemma:norm_bound3}
    Assume that $\lambm$ satisfies the condition
    \[\lambm  \ge 2 \bracks{\normop{\bomega\circ\bX\Delta\bbeta}+\normop{\bomega\circ\Delta\bgamma\bZ}+\normop{\bomega\circ\bepsilon}}.\]
   Then,
    \[ \normnuc{\Delta\bM} \le  \frac{\amax^2}{2\amin^2-1} \sqrt{32r}\normf{\Delta\bM}. \]
\end{lemma}

The next step is to demonstrate that the inequality \eqref{ineq:rsc_M} holds for some $D_M$. Establishing this inequality verifies the RSC condition. By utilizing the value of $D_M$ alongside the optimality condition in \eqref{ineq:Lossdef}, we derive the error bound for $\Delta\bM$ in the general case. We present the generalization of Lemma \ref{lemma:rsc2} for the general case below; the proof is omitted, as it follows the same sequence of steps detailed in Appendix \ref{proof_lemma:rsc2}.

\begin{lemma}\label{lemma:rsc3}
Suppose that $\lambm$ satisfies 

\[\lambm \geq4 \left(\sigma_x\,c_\beta  + \sigma_z \, c_\gamma +  \sigma \,\sqrt{2\,\pi_{\max}\bracks{n\vee m}\log\bracks{\bracks{m+n}/2}} \right),\]
then for any $\Delta \bM \in \mathcal{\wtilde{S}}(d_5,a)$, where $a$ is defined as above and 
\[d_5 = \frac{128\,c^2_m\,\log(n+m)}{\pi_{\text{min}}\,\log(6/5)},\]
then the following inequality holds with a probability at least $1-8(n+m)^{-1}$
\begin{align*}
%\label{ineq: M2}
\|\boldsymbol{\bomega} \circ \Delta\bM\|_{F}^{2} \ge \frac{1}{2}\,\mathbb{E}\bigl\|\boldsymbol{\bomega} \circ \Delta\bM\bigr\|_{F}^{2}
      - \wtilde{D}_M,
\end{align*}
where
\[
  \wtilde{D}_M
        = 3072\, c_m^2\, \pi_{\min}^{-1} \bracks{8\,a^2\,r\,N_1^2+1}.
\]
\end{lemma}

\textbf{Note.} Lemmas \ref{lemma:norm_bound2} and \ref{lemma:rsc2} are special cases of Lemmas \ref{lemma:norm_bound3} and \ref{lemma:rsc3}, respectively, when $\amin=\amax=1$. 
Furthermore, the distinction between $D_M$ and $\wtilde{D}_M$ is that $\wtilde{D}_M$ scales with $a^2$. Consequently, the convexity of the loss function is contingent upon a sufficiently small difference between the maximum and minimum singular values of the similarity matrices.

Having established the RSC condition, we parallel the sequence of steps outlined in Appendix \ref{subsection:bound_M} to obtain an upper bound on $\Delta\bM$ for the general case. As previously noted, redundant intermediate steps are omitted for conciseness. Evaluating the optimality condition in \eqref{ineq:Lossdef} at $\bstheta=\bracks{\estim\bbeta,\estim\bgamma,\truev\bM}$, with the objective function formulated as in \eqref{eq:problem_3}, yields
\begin{align*}
\frac{1}{2} \normfrob{\bomega\circ\bracks{\bY-\estim\btheta}} &\leq \frac{1}{2} \normfrob{\bomega\circ\bracks{\bY-\bX\estim\bbeta-\estim\bgamma\bZ-\truev\bM}} + \lambm \normdiffd{\normnuc}{\bM}{\bQa}{\bQb} \\
&\leq \frac{1}{2} \normfrob{\bomega\circ\bracks{\bY-\bX\estim\bbeta-\estim\bgamma\bZ-\truev\bM}} + \lambm \normnuc{\bQa\estim\bM\bQb-\bQa\truev\bM\bQb}\\
&\leq \frac{1}{2} \normfrob{\bomega\circ\bracks{\bY-\bX\estim\bbeta-\estim\bgamma\bZ-\truev\bM}} + \lambm \normnuc{\bQa\bracks{\estim\bM-\truev\bM}\bQb}\\
&\leq \frac{1}{2} \normfrob{\bomega\circ\bracks{\bY-\bX\estim\bbeta-\estim\bgamma\bZ-\truev\bM}} + \lambm \amax^2 \normnuc{\Delta\bM}.
\end{align*}
Assuming $\lambm  \ge 2 \bracks{\normop{\bomega\circ\bX\Delta\bbeta}+\normop{\bomega\circ\Delta\bgamma\bZ}+\normop{\bomega\circ\bepsilon}}$, we obtain
\begin{align}\notag
    \normfrob{\bomega\circ\Delta\bM} &\leq \lambm \normnuc{\Delta\bM} \bracks{1 + 2 \,\amax^2}\\
    &\leq \lambm \,a\,\bracks{1 + 2 \,\amax^2}\,\sqrt{32\,r} \normf{\Delta\bM}. \label{inequality:bound:M:lambda2}
\end{align}

It is straightforward to verify that \eqref{inequality:bound:M:lambda2} reduces to \eqref{inequality:bound:M:lambda} when we have that $\amin=\amax=1$ in the special case. Subsequently, Lemma \ref{lemma:rsc3} implies that
\begin{align*}
    \normfrob{\Delta\bM} & \leq 2\, \pim^{-1} \bracks{\normfrob{\bomega\circ\Delta\bM}+ \wtilde{D}_M} \\
    &\leq 2\,\pim^{-1} \bracks{ \lambm \,a\,\bracks{1 + 2 \,\amax^2}\,\sqrt{32\,r} \normf{\Delta\bM} + \wtilde{D}_M} \\
    &\leq \frac{1}{2}\pim^{-2} { \lambm^2 \,a^2\,\bracks{1 + 2 \,\amax^2}^2\,{128\,r} + \frac{1}{2} \normfrob{\Delta\bM} + 2\,\pim^{-1}\wtilde{D}_M},
\end{align*}
where the last step follows by applying Young's inequality. Consequently
\begin{align} \notag
    \normfrob{\Delta\bM} &\leq \frac{128}{\pim^2} \bracks{a^2\,r\,\bracks{\lambm^2\bracks{1+2\amax^2}^2+768\,c_m^2\,N_1^2}+96\,c^2_m}\\
    &\lesssim \frac{r\,a^2}{\pim^2}\bracks{\bracks{1\vee \amax^4}\lambm^2 + N_1^2}.
    \label{inequality:bound:M:estim:2}
\end{align}

Observe that the general bound in \eqref{inequality:bound:M:estim:2} reduces to \eqref{inequality:bound:M:estim} in the special case when $\amax=\amin=1$. The parameters $a$ and $\amax^2$ in \eqref{inequality:bound:M:estim:2} represent the additional penalty incurred by incorporating the similarity matrices in the estimation procedure. The parameter $a= \amax^2/\bracks{2\amin^2-1}$ is a modified condition number of the similarity matrices. It imposes a global penalty on the bound based on the spectral spread of the similarity matrices. However, Assumption \ref{assumption:similarity} strictly bounds both quantities. Because $\amax^2 \leq {1/\gamma}$ and $a \leq \bracks{1+\gamma}/\bracks{\gamma-\gamma^2}$, their effect remains finite and universally bounded, provided that $\gamma$ is chosen appropriately away from the boundaries. Under these conditions, the bound achieves the same order as \eqref{inequality:bound:M:estim:last}.  
\qed

This appendix concludes with the proof of Lemma \ref{lemma:norm_bound3}.

\textbf{Proof of Lemma \ref{lemma:norm_bound3}}

Applying inequality \eqref{eq:optimality} to the tuple $\bracks{\estim\bbeta,\estim\bgamma,\truev\bM}$ yields
\begin{align*}
    & \lambm \normdiff{\normnuc}{\bM}{\bQa}{\bQb} \\
     \leq&  \bracksd{\bomega \circ \bracks{\estim\btheta-\bY \pm\truev\btheta}}{-\Delta\bM} \\
    =\,& \bracksd{\bomega \circ \bracks{\Delta\bM + \bX\Delta\bbeta +\Delta\bgamma\bZ - \bepsilon}}{-\Delta\bM} \\
   \leq
    & \bracks{\normop{\bomega\circ\bX\Delta\bbeta}+\normop{\bomega\circ\Delta\bgamma\bZ}+\normop{\bomega\circ\bepsilon}} 
    \normnuc{\Delta\bM}, 
\end{align*}

where we dropped the negative quadratic terms and applied Hölder's inequality. Suppose we choose the regularization parameter such that
\begin{align*}
    \lambm &\ge 2 \bracks{\normop{\bomega\circ\bX\Delta\bbeta}+\normop{\bomega\circ\Delta\bgamma\bZ}+\normop{\bomega\circ\bepsilon}}.
\end{align*}
Define $\truev\bW = \bQa \truev\bM\bQb$ and $\estim\bW = \bQa\estim\bM\bQb$. Because $\bQa$ and $\bQb$ are full-rank positive definite matrices, $\operatorname{rank}\bracks{\truev\bW} = \operatorname{rank}\bracks{\truev\bM} = r$. The main inequality then simplifies to
\begin{align*}
    \lambm\normnuc{\Delta\bM} &\ge 2 \lambm \bracks{\normnuc{\estim\bW}-\normnuc{\truev\bW}}.
\end{align*}
Next, we establish a lower-bound on $\normnuc{\estim\bW}$. To this end, we define the subspace of all matrices in $\R{n}{m}$ with rank $r$. Let $\mathbb{U}$ and $\mathbb{V}$ denote $r$-dimensional vector subspaces, and let $\mathbb{U}^\perp$ and $\mathbb{V}^{\perp}$ denote their orthogonal complements. We subsequently define the following two matrix subspaces
\begin{align*}
    \mathbb{M}\bracks{\mathbb{U}, \mathbb{V}} &:= \bracksb{\bM \in \R{n}{m} \,\,| \,\,\operatorname{rowspan}(\bM)\subseteq \mathbb{V}, \,\,\operatorname{colspan}(\bM) \subseteq \mathbb{U}}, \\
    \overline{\mathbb{M}}^{\perp}\bracks{\mathbb{U}, \mathbb{V}} &:= \bracksb{\bM \in \R{n}{m} \,\,| \,\,\operatorname{rowspan}(\bM)\subseteq \mathbb{V}^\perp, \,\,\operatorname{colspan}(\bM) \subseteq \mathbb{U}^{\perp}}.
\end{align*}
Any matrix in the subspace $\mathbb{M}$ has rank at most $r$; consequently, $\truev\bM \in \mathbb{M}$ and $\truev\bW \in \mathbb{M}$. Let $\overline{\mathbb{M}}$ denote the orthogonal subspace of $\overline{\mathbb{M}}^{\perp}$, which implies that $\mathbb{M} \subset\overline{\mathbb{M}}$ and, by properties of the nuclear norm, that $\overline{\mathbb{M}}$ and $\overline{\mathbb{M}}^{\perp}$ are decomposable. Finally, any matrix in the subspace $\overline{\mathbb{M}}$ has rank at most $2r$. Using this and the triangle inequality, we obtain
\begin{align*}
    \normnuc{\estim\bW} & = \normnuc{\estim\bW+\truev\bW-\truev\bW} \\
    & = \normnuc{\truev\bW+\proj{\Delta\bW}+ \projc{\Delta\bW}}  \\
     & \ge \normnuc{\truev\bW + \projc{\Delta\bW}} - \normnuc{\proj{\Delta\bW}} \\
    & = \normnuc{\truev\bW} + \normnuc{\projc{\Delta\bW}}  -\normnuc{\proj{\Delta\bW}} \\
    & \ge \normnuc{\truev\bW} + \normnuc{\Delta\bW}  - 2 \normnuc{\proj{\Delta\bW}} \\
    & \ge \normnuc{\truev\bW} + \normnuc{\Delta\bW} - 2 \sqrt{2r} \, \normf{\proj{\Delta\bW}} \\
    & \ge \normnuc{\truev\bW} + \normnuc{\Delta\bW} -  \sqrt{8r}\, \normf{\Delta\bW}\\
    & \ge \normnuc{\truev\bW} + \amin^2\normnuc{\Delta\bM} -  \sqrt{8r}\,\amax^2 \normf{\Delta\bM}, 
\end{align*}
where $\proj{\cdot}$ and $\projc{\cdot}$ denote the projection operator onto the subspaces $\overline{\mathbb{M}}$ and $\overline{\mathbb{M}}^\perp$, respectively. Substituting this into the main inequality yields
\begin{align*}
    \normnuc{\Delta\bM} \le \frac{\amax^2}{2\amin^2-1} \sqrt{32r}\normf{\Delta\bM}.
\end{align*}
\qed

  \bibliography{supp_biblio}